\documentclass[11pt]{article}
\usepackage[letterpaper,margin=1in]{geometry}

\usepackage{xcolor} %
\usepackage{ninecolors}\NineColors{saturation=high}
\usepackage{amsmath,amssymb,amsthm,amsfonts,latexsym,bbm,xspace,graphicx,float,mathtools,braket,dirtytalk}
\usepackage{enumitem}
\usepackage{xfrac}   

\newcommand{\len}{\mathrm{len}}
\newcommand{\Mod}{\mathsf{MOD}}
\newcommand{\tdhlf}{\textsc{2D HLF}}

\newcommand{\php}{\textsc{PHP}}
\newcommand{\thr}{\mathsf{THR}}
\newcommand{\e}{\mathsf{E}}
\newcommand{\G}{\mathsf{G}}

\usepackage{siunitx}
\usepackage{booktabs}

\usepackage{yhmath}
\usepackage{subcaption}
\usepackage{sidecap}
\sidecaptionvpos{figure}{c}
\usepackage[linesnumbered,ruled,vlined]{algorithm2e}
\SetKwInput{KwInput}{Input}                %
\SetKwInput{KwOutput}{Output}
\SetKwFor{RepTimes}{repeat}{times}{end}

\usepackage{arydshln}
\usepackage{nicematrix}

\usepackage[colorinlistoftodos]{todonotes}

\makeatletter
\renewcommand*\env@matrix[1][*\c@MaxMatrixCols c]{%
  \hskip -\arraycolsep
  \let\@ifnextchar\new@ifnextchar
  \array{#1}}
\makeatother

\definecolor{ballblue}{rgb}{0.13, 0.67, 0.8}
\definecolor{bittersweet}{rgb}{1.0, 0.44, 0.37}
\definecolor{palered-violet}{rgb}{0.86, 0.44, 0.58}
\definecolor{midnightblue}{rgb}{0.071, 0.271, 0.349}
\definecolor{airforceblue}{rgb}{0.35, 0.51, 0.57}

\usepackage[pagebackref,hypertexnames=false]{hyperref}
\hypersetup{
    colorlinks=true,
    urlcolor=blue3,
    linkcolor=green3,
    citecolor=red3,
}

\usepackage[nameinlink,capitalize]{cleveref}
\crefname{algocf}{Algorithm}{Algorithms}
\Crefname{algocf}{Algorithm}{Algorithms}
\crefname{algocfline}{Line}{Lines}
\Crefname{algocfline}{Line}{Lines}

\renewcommand{\backref}[1]{}

\renewcommand{\backrefalt}[4]{%
\ifcase #1 %
\or
[p.\ #2]%
\else
[pp.\ #2]%
\fi}

\newtheorem{theorem}{Theorem}[section]
\newtheorem*{namedtheorem}{\theoremname}
\newcommand{\theoremname}{testing}

\newtheorem{lemma}[theorem]{Lemma}

\newtheorem{proposition}[theorem]{Proposition}

\newtheorem{corollary}[theorem]{Corollary}

\newtheorem{question}[theorem]{Question}

\theoremstyle{definition}
\newtheorem{definition}[theorem]{Definition}

\newtheorem{remark}[theorem]{Remark}

\renewcommand{\Pr}{\mathop{\bf Pr\/}}
\newcommand{\E}{\mathop{\bf E\/}}

\newcommand{\poly}{\mathrm{poly}}

\newcommand{\polylog}{\mathrm{polylog}}
\newcommand{\size}{\mathrm{size}}

\newcommand{\sgn}{\mathrm{sgn}}

\newcommand{\R}{\mathbb R}

\newcommand{\N}{\mathbb N}
\newcommand{\Z}{\mathbb Z}
\newcommand{\F}{\mathbb F}

\newcommand{\BQP}{\mathsf{BQP}}
\newcommand{\PH}{\mathsf{PH}}
\newcommand{\PTIME}{\mathsf{P}}

\newcommand{\NP}{\mathsf{NP}} 
\newcommand{\TC}{\mathsf{TC}}
\newcommand{\AC}{\mathsf{AC}}
\newcommand{\ac}{\AC}
\newcommand{\sym}{\SYM}
\newcommand{\SYM}{\mathsf{SYM}}

\newcommand{\NEXP}{\mathsf{NEXP}}

\newcommand{\GCC}{\mathsf{GCC}} 
\newcommand{\GC}{\mathsf{GC}} \newcommand{\gc}{\GC}
 
\newcommand{\ACC}{\mathsf{ACC}} \newcommand{\acc}{\ACC}
\newcommand{\ACCz}{\mathsf{ACC^0}} 
\newcommand{\QNC}{\mathsf{QNC}}
\newcommand{\NC}{\mathsf{NC}}
 
\newcommand{\OR}{\mathsf{OR}} 
\newcommand{\NOT}{\mathsf{NOT}} 
\newcommand{\AND}{\mathsf{AND}} 
\newcommand{\BQLOGTIME}{\mathsf{BQLOGTIME}}
\newcommand{\g}{\mathsf{G}}
\newcommand{\CKT}{\mathsf{CKT}}
\newcommand{\DT}{\mathsf{DT}}

\newcommand{\MAJ}{\mathsf{MAJ}} \newcommand{\maj}{\MAJ}

\newcommand{\PSPACE}{\mathsf{PSPACE}}
\newcommand{\PP}{\mathsf{PP}}

\newcommand{\eps}{\varepsilon}

\newcommand{\bits}{\{-1,1\}}
\newcommand{\bitz}{\{0,1\}}

\newcommand{\calC}{\mathcal{C}}
\newcommand{\calD}{\mathcal{D}}
\newcommand{\calE}{\mathcal{E}}
\newcommand{\calF}{\mathcal{F}}

\newcommand{\calL}{\mathcal{L}}

\newcommand{\calP}{\mathcal{P}}

\newcommand{\calR}{\mathcal{R}}
\newcommand{\calS}{\mathcal{S}}

\newcommand{\quasipoly}{\mathrm{quasipoly}}

\newcommand{\abs}[1]{\lvert #1 \rvert}
\newcommand{\Abs}[1]{\Bigl\lvert #1 \Bigr\rvert}

\newcommand{\ignore}[1]{}

\newcounter{termcounter}[equation]
\renewcommand{\thetermcounter}{\the\numexpr\value{equation}+1\relax.\roman{termcounter}}
\crefname{term}{term}{terms}
\creflabelformat{term}{#2\textup{(#1)}#3}

\makeatletter
\def\term{\@ifnextchar[\term@optarg\term@noarg}%
\def\term@optarg[#1]#2{%
  \textup{#1}%
  \def\@currentlabel{#1}%
  \def\cref@currentlabel{[][2147483647][]#1}%
  \cref@label[term]{#2}}
\def\term@noarg#1{%
  \refstepcounter{termcounter}%
  \textup{\thetermcounter}%
  \cref@label[term]{#1}}
\makeatother

\usepackage{enumitem}

\title{Improved Circuit Lower Bounds and\\Quantum-Classical Separations}

\author{Sabee Grewal\thanks{\href{mailto:sabee@cs.utexas.edu,vmkumar@cs.utexas.edu}{\texttt{\{sabee, vmkumar\}@cs.utexas.edu}}. Department of Computer Science, The University of Texas at Austin.}\and Vinayak M. Kumar}

\date{}

\renewcommand\footnotemark{}
\begin{document}

\maketitle

\begin{abstract}
We continue the study of the circuit class $\GC^0$, which augments $\AC^0$ with unbounded-fan-in gates that compute arbitrary functions inside a sufficiently small Hamming ball but must be constant outside it.
While $\GC^0$ can compute functions requiring exponential-size circuits, Kumar (CCC 2023) showed that switching-lemma lower bounds for $\AC^0$ extend to $\GC^0$ with no loss in parameters.

We prove a parallel result for the polynomial method: any lower bound for $\AC^0[p]$ obtained via the polynomial method extends to $\GC^0[p]$ without loss in parameters. As a consequence, we show that the majority function $\maj$ requires depth-$d$ $\GC^0[p]$ circuits of size $2^{\Omega(n^{1/2(d-1)})}$, matching the best-known lower bounds for $\AC^0[p]$. This yields the most expressive class of non-monotone circuits for which exponential-size lower bounds are known for an explicit function.
We also prove a similar result for the algorithmic method, showing that $\e^\NP$ requires exponential-size $\GCC^0$ circuits, extending a result of Williams (JACM 2014).

Finally, leveraging our improved classical lower bounds, we establish the strongest known unconditional separations between quantum and classical circuit classes. 
We separate $\QNC^0$ from $\GC^0$ and $\GC^0[p]$ in various settings and show that $\BQLOGTIME$ is not contained in $\GC^0$. As a consequence, we construct an oracle relative to which $\BQP$ lies outside uniform $\GC^0$, extending the Raz–Tal oracle separation between $\BQP$ and $\PH$ (STOC 2019).
\end{abstract}

\hypersetup{linktocpage}
\setcounter{tocdepth}{2}
\tableofcontents

\section{Introduction}

Proving superpolynomial circuit lower bounds for an explicit function is a longstanding challenge in computer science. 
It remains one of our only viable approaches to resolving the $\PTIME \stackrel{?}{=} \NP$ question \cite{aaronson2016p}.
Beyond this central goal, circuit lower bounds also find applications throughout complexity theory, for example, in structural complexity \cite{furst1984parity,haastad1986computational,aaronson2010bqp,rossman2015average,raz2022oracle}, proving unconditional quantum advantage \cite{bravyi2018quantum,watts2019exponential,grier2019interactive}, and pseudorandomness \cite{nisan1994hardness,impagliazzo1997p}. 

Motivated in part by the relativization barrier of Baker, Gill, and Solovay \cite{baker1975relativizations}, 
considerable effort was put forth in the mid-1970s to early 1980s to prove circuit lower bounds for explicit functions.
After a burst of progress \cite{ schnorr1974zwei,paul19752,stockmeyer1976combinational,schnorr19803n,blum19812,blum1983boolean}, the best lower bound for an explicit function was $3n - o(n)$. 
The current state of the art is $3.1n - o(n)$, and the (seemingly) marginal improvement in the leading constant was highly nontrivial to obtain \cite{demenkov2011elementary,find2016better,golovnev2016weighted,li20221}. 

A ``bottom-up'' approach to circuit lower bounds has also been explored, where the goal is to prove lower bounds for highly restricted circuits, then slightly relax those restrictions and repeat. 
This approach has led to two techniques: switching lemmas (or more broadly, the method of random restrictions) \cite{ajtai198311,furst1984parity,yao1985separating,haastad1986computational} and the polynomial method \cite{razborov1987lower,smolensky1987algebraic}. 
The former technique has been used to show lower bounds against $\AC^0$, constant-depth circuits of $\AND$, $\OR$, and $\NOT$ gates with unbounded fan-in. The latter technique has been used to prove lower bounds against $\AC^0[p]$, constant-depth circuits that include unbounded fan-in $\Mod_p$ gates, where $p$ is prime, in addition to $\AND$, $\OR$, and $\NOT$ gates.\footnote{$\Mod_p$ outputs $0$ iff the sum of the input bits is congruent to $0 \pmod p$.}

Alas, this bottom-up approach stalled in the late 1980s. 
Furthermore, the natural proofs barrier of Razborov and Rudich \cite{razborov1994natural} showed that the random restriction and polynomial methods fail to prove superpolynomial-size lower bounds against $\TC^0$, constant-depth, polynomial-size circuits of $\AND$, $\OR$, $\NOT$ and $\MAJ$ gates with unbounded fan-in---a circuit class far weaker than polynomial-depth, polynomial-size circuits.\footnote{$\MAJ$ outputs $1$ iff at least half of the input bits are $1$.}
Additionally, Aaronson and Wigderson \cite{aaronson2009algebrization} identified a third barrier, the algebrization barrier, another hurdle any new lower bound technique must overcome.

The gold standard in circuit complexity is the development of new lower bound techniques that circumvent known barriers.
A shining example is Williams' algorithmic method, which led to breakthrough $\ACC^0$ lower bounds \cite{wil14acc0}.\footnote{$\ACC^0$ is the union of $\AC^0[m]$ for all $m$.}
However, new techniques are few and far between.
In this work, we take a complementary approach: rather than seeking new techniques, we aim to refine our understanding of existing ones.
By examining how and where current methods fail, we hope to gain insight into what future breakthroughs might require.
Broadly, our work is driven by two motivating questions: 

\begin{question}\label{question:first}
What is the strongest circuit class for which current techniques can still yield nontrivial lower bounds? 
\end{question}

\begin{question}\label{question:second}
Is there a unifying perspective that captures existing techniques, revealing a common structure or property they all exploit?
\end{question}

\subsection{Our Results in a Nutshell}
An early attempt to unify and extend lower bound techniques was made by Yao \cite{yao1989circuits}, who observed that certain lower bounds hold even when circuits are augmented with \emph{local computation}, i.e., bounded fan-in gates that compute arbitrary functions.
For example, Yao showed that Razborov’s monotone circuit-size lower bound for $k$-Clique on $n$ vertices \cite{razborov1985lower} holds even when the monotone circuits are allowed arbitrary monotone gates of fan-in $n^{1/100}$ (whereas Razborov's original lower bound assumed gates of fan-in 2). 
In a follow-up work, Jukna \cite{jukna1990monotone} showed that Razborov's lower bound holds for arbitrary monotone gates of fan-in $n$ as long as the minterm of each gate is at most $(n/\log n)^{2/3}$.

Beyond proving lower bounds for more expressive circuit classes, the study of local computation has also been used to analyze the limitations of lower bound techniques, a perspective taken by Chen, Hirahara, Oliveira, Pich, Rajgopal, and Santhanam \cite{chen2022beyond}. 
At a high level, the idea is as follows: if a lower bound technique for $\AC^0$ also applies to some larger class $\mathsf{C}$, it suggests that the technique is insensitive to the differences between $\AC^0$ and $\mathsf{C}$.
By analyzing this insensitivity more carefully, one can hope to refine the technique and obtain stronger lower bounds against $\AC^0$.

The notion of locality studied in prior work---arbitrary computation over a small number of input bits---does not generalize constant-depth circuits with \emph{unbounded} fan-in.
For example, even a single unbounded fan-in $\OR$ gate cannot be implemented by a constant-depth circuit with only bounded fan-in gates.
To extend the line of investigation pursued by Yao, Jukna, and Chen et al. to the unbounded fan-in setting, we must identify a notion of locality that is compatible with unbounded fan-in gates.

Recently, Kumar \cite{kumar2023tight} introduced the $\G(k)$ gate: an unbounded fan-in gate that can compute an arbitrary function within a Hamming ball of radius $k$ but must be constant outside it.
In this work, we propose interpreting the $\G(k)$ gate as defining a new notion of locality---one that is especially well-suited to the unbounded fan-in setting.
To see this, observe that $\AND$, $\OR$, and $\NOT$ can be viewed as special cases of this model: each computes a function that depends only on inputs within a Hamming ball of radius $0$, and is constant elsewhere.
Thus, the circuit class $\GC^0(k)$, constant-depth circuits built from $\G(k)$ gates, naturally generalizes $\AC^0$.
Moreover, since arbitrary bounded fan-in gates are also special cases of $\G(k)$ gates, this definition subsumes earlier models of local computation studied by Yao and by Chen et al., while extending them to include unbounded fan-in.\footnote{
Let us briefly compare our model with the more traditional notion of locality, i.e., the arbitrary bounded fan-in model considered in Yao's work. 
Specifically, consider constant-depth circuits composed of unbounded fan-in $\AND$, $\OR$, and $\NOT$ gates, along with arbitrary gates of bounded fan-in $k$---call this class $\mathsf{YAO}^0$. This offers a natural point of comparison with our $\GC^0(k)$ model.

A natural question is whether $\GC^0(k)$ can be simulated by $\mathsf{YAO}^0$. However, a simple counting argument shows that this is not the case: there exist individual $\G(k)$ gates that require exponential-size $\mathsf{YAO}^0$ circuits to implement. Indeed, a size-$s$ $\mathsf{YAO}^0$ circuit on $n$ input bits can be encoded by $s(k\log(n+s) + 2^k)$ bits: each of the $s$ gates is specified by its $k$ inputs and the length-$2^k$ truth table. In contrast, a $\G(k)$ gate of fan-in $n$ requires $\binom{n}{\le k} = \Omega((n/k)^k)$ bits to specify. Thus $s$ must be $(n/k)^{\Omega(k)}$, which is exponential in $n$ when $k = n^\eps$---the regime of interest in this work. 
}

The main result of \cite{kumar2023tight} was to prove a novel switching lemma for $\GC^0$, which implies lower bounds for $\GC^0$ that are just as strong as those known for $\AC^0$. The core takeaway is captured by the following informal theorem:

\begin{theorem}[Main result of {\cite{kumar2023tight}}, Informal]\label{thm:intro-kumar}
If one can prove size-$s$ lower bounds against depth-$d$ $\AC^0$ using a switching lemma, then one can prove size-$s$ lower bounds against depth-$d$ $\GC^0(k)$ even when $k = 0.1 n^{1/d}$ (for a possibly different hard function).
\end{theorem}

This result is surprising because, in this regime of $k$, a simple counting argument shows that $\GC^0(k)$ can compute functions requiring exponential-size Boolean circuits (see \cref{thm:gk-incomp-TC}).
In the spirit of Yao \cite{yao1989circuits} and Jukna \cite{jukna1990monotone}, \cref{thm:intro-kumar} yields new lower bounds for a strictly more powerful class of circuits. 
But in the spirit of Chen et al.\cite{chen2022beyond}, the result also illuminates the limitations of the technique itself. In particular, it shows that the switching lemma cannot distinguish between $\AC^0$ and $\GC^0(k)$. In other words, the technique applies equally well to both classes, despite the latter's significantly greater computational power.

The first contribution of this work is to show the analogous result for the polynomial method. 

\begin{theorem}[Improved circuit lower bounds, Informal]
Define $\GC^0(k)[p]$ as the class of constant-depth $\GC^0(k)$ circuits augmented with unbounded fan-in $\Mod_p$ gates.
If one can prove size-$s$ lower bounds against depth-$d$ $\AC^0[p]$ using the polynomial method, then one can prove size-$s$ lower bounds against depth-$d$ $\GC^0(k)[p]$ even when $k = 0.1 n^{1/2d}$ (for a possibly different hard function).
\end{theorem}

Towards addressing \cref{question:first}, our result yields exponential-size circuit lower bounds against $\GC^0(k)[p]$ in a regime where this class can compute functions requiring exponential-size Boolean circuits. 
In particular, our results give the least restricted class of non-monotone circuits for which we have exponential-size circuit lower bounds against an explicit function (see \cref{remark:explicit-lower-bound} for further detail). 
Notably, $\AC^0[p]$ and $\GC^0(k)[p]$ provably do not satisfy a switching lemma, so our lower bounds could not have been achieved by prior work.\footnote{It is natural to wonder if existing $\AC^0[p]$ lower bounds already imply our results for $\GC^0[p]$. We explain at length in \cref{subsec:intro:our-new-circuit-lower-bound} why this is not the case.} 

Towards addressing \cref{question:second}, a central conceptual contribution of this work is to identify a broader notion of locality---namely, arbitrary computation restricted to small Hamming balls---as the key property exploited by both the switching lemma and the polynomial method. 
Strikingly, both techniques operate at a level that is agnostic to the precise gate types involved, so long as the computation remains sufficiently local in this Hamming-ball sense.
This is particularly surprising, as the switching lemma and the polynomial method are deeply different in nature---combinatorial versus algebraic---yet both extend naturally to $\g(k)$ gates.

Because our result shows that the polynomial method cannot distinguish between $\AC^0[p]$ and $\GC^0(k)[p]$, it can also be interpreted as identifying a \emph{barrier}---much like relativization, naturalization, and algebrization \cite{baker1975relativizations,razborov1994natural,aaronson2009algebrization}.
Specifically, if a function $f$ can be computed by size-$s$ $\GC^0[p]$ circuits, then neither the polynomial method nor the switching lemma can be used to prove a stronger than size-$s$ lower bound against $\AC^0[p]$. Otherwise, by our results, such a lower bound would lift to $\GC^0[p]$---contradicting the assumed existence of a small $\GC^0[p]$ circuit for $f$.
This perspective may help explain why certain lower bounds remain elusive, such as proving tight lower bounds for $\MAJ$ against $\AC^0[p]$.  

In addition to these conceptual contributions and new circuit lower bounds, we also present several related results. We prove analogous (but weaker) results for the algorithmic method. Furthermore, our new lower bounds have a range of applications, including to learning theory and quantum-classical separations. We discuss these in detail in the following subsection.

\subsection{Our Results in Detail}

\subsubsection{Our New Circuit Lower Bound}\label{subsec:intro:our-new-circuit-lower-bound}
Our first result uses the polynomial method to prove exponential-size lower bounds for $\GC^0(k)[p]$ circuits.

\begin{theorem}[{$\GC^0(k)[p]$ lower bound, Restatement of \cref{thm:gc0-lower-bound}}]\label{thm:intro-gc0-lower-bound}
    Let $p$ and $q$ be distinct prime numbers, and let $k=O(n^{1/2d})$.
    Any depth-$d$ $\GC^0(k)[p]$ circuit that computes either $\MAJ$ or $\Mod_q$ on $n$ input bits  must have size $2^{\Omega\left(n^{1/2(d-1)}\right)}$.
\end{theorem}

Notably, this lower bound \emph{matches} the best-known bound for depth-$d$ $\AC^0[p]$. 

\paragraph{Is $k=O(n^{1/2d})$ optimal?} 
The locality $k = O(n^{1/2d})$ in \cref{thm:intro-gc0-lower-bound} is optimal up to a factor of $2$ in the exponent; specifically, there is a gap between $1/2d$ and $1/d$.
This is because $\Mod_q$ over $n$ bits can be computed by a depth-$d$ circuit of size $O(n^{1-1/d})$ using $\Mod_q$ gates of fan-in $n^{1/d}$---that is, by a $\GC^0(n^{1/d})$ circuit.

\paragraph{Why the na\"ive aproach fails}
It is natural to ask whether our lower bound for $\GC^0(k)[p]$ could be recovered by simulating such circuits within $\AC^0[p]$ and applying known lower bounds. This naïve approach, however, fails: it incurs an unavoidable blow-up in size and therefore yields much weaker bounds.

Suppose we have a depth-$d$ size-$s$ $\GC^0(k)[p]$ circuit with $s = 2^{\Theta(n^{\frac{1}{2(2d-1)}}/k)}$. To simulate it in $\AC^0[p]$, each $\G(k)$ gate could have fan-in up to $s$, and upon expressing each one as a CNF or DNF of size $s^{O(k)}$, we obtain a depth-$2d$, size-$2^{O(n^{\frac{1}{2(2d-1)}})}$ $\AC^0[p]$ circuit. This blow-up is inherent: by a counting argument,\footnote{The number of CNFs/DNFs of size $t$ is at most $2^{nt}$, while a single $\G(k)$ gate of fan-in $s$ requires $s^{\Omega(k)}$ bits to specify. Thus, representing all such gates requires $t = s^{\Omega(k)}$.} there exists $\G(k)$ gates of fan-in $s$ requiring size $s^{\Omega(k)}$ when expressed as a CNF/DNF. Hence, any size-$s$ $\GC^0(k)[p]$ circuit including such a gate will have a size-$s^{O(k)}$ simulating circuit.

Known $\ac^0[p]$ lower bounds \cite{razborov1987lower,smolensky1987algebraic} imply depth-$2d$ size-$2^{O(n^{\frac{1}{2(2d-1)}})}$ $\ac^0[p]$ circuits cannot compute majority. Combining this with our simulation implies a size $s=2^{\Theta(n^{\frac{1}{2(2d-1)}}/k)}$ lower bound for $\GC^0(k)[p]$, which is far weaker than our exponential bound in \cref{thm:intro-gc0-lower-bound} due to the $1/k$ factor in the exponent. 
Even for constant $k$ this is weaker than our $2^{\Omega(n^{1/2d})}$ bound, and when $k \ge n^{1/(4d-2)}$ it yields no nontrivial bound at all. By contrast, \cref{thm:intro-gc0-lower-bound} gives $2^{\Omega(n^{1/2(d-1)})}$ lower bounds for all $k \le O(n^{1/2d})$, bypassing the simulation bottleneck entirely.

One might also ask whether saving on the depth blow-up from $d$ to $2d$ could salvage the simulation. The answer is no: the real obstacle is the size blow-up, which persists regardless of depth.\footnote{A size-$t$ $\AC^0[p]$ circuit is describable in $O(t^2\log t)$ bits. As a $\G(k)$ gate of fan-in $s$ requires $s^{\Omega(k)}$ bits to specify, it follows there exists $\G(k)$ gates that require size $s^{\Omega(k)}$ $\AC^0[p]$ circuits (regardless of depth). } For completeness, we note that some depth reduction is possible, but it seems challenging to avoid some constant factor blow-up in depth. If our $\GC^0(k)$ circuit had no $\Mod_p$ gates, expanding even layers into CNFs and odd layers into DNFs collapses to depth $d+1$. In the presence of $\Mod_p$ gates, a similar argument yields depth $3d/2$, and whether further collapse to $d+1$ is possible is, to the best of our knowledge, a challenging open problem.\footnote{For example, consider a depth-$d$ $\GC^0(k)[p]$ circuit where even layers are $\Mod_p$ gates and odd layers are $\G(k)$ gates not in $\G(k-1)$. Expanding the $\G(k)$ gates does not obviously allow collapse due to the sandwiching layers of $\Mod_p$ gates, leading to depth $3d/2$.} In any case, even with such reductions, the $s^{O(k)}$ size blow-up rules out recovering our bounds through a na\"ive simulation.

\subsubsection{Related Classical Results}

We now outline the key ingredients in the proof of \cref{thm:intro-gc0-lower-bound}, along with our results on the algorithmic method, PAC learning of $\GC^0(k)[p]$ circuits, and a new multi-output multi-switching lemma for $\GC^0(k)$.

\paragraph{Proof Ingredients for \cref{thm:intro-gc0-lower-bound}}
The key lemma in our argument is to show that any $\g(k)$ gate can be computed by a probabilistic polynomial of extremely low degree (\cref{def:prob-poly}). 

\begin{lemma}[Restatement of \cref{lem:gateapprox}]
\label{lem:intro-gateapprox}
    For any prime $p$ and $\g(k)$ gate $G$ of fan-in $n$, there is an $\eps$-probabilistic $\F_p$-polynomial of degree $O(k + \log(1/\eps))$ computing $G$.
\end{lemma}

This upper bound is, in fact, optimal: 

\begin{lemma}[Restatement of \cref{lem:prob-poly-lower-bound}]\label{lem:intro-prob-poly-lower-bound}
There exists a $\g(k)$ gate that requires probabilistic degree $\Omega(k + \log(1/\eps))$.
\end{lemma}

The tightness of our degree bound in \cref{lem:intro-gateapprox} is crucial for obtaining $\GC^0(k)[p]$ lower bounds that \emph{match} the best-known $\AC^0[p]$ lower bounds. Anything even slightly suboptimal would not suffice! 
For example, had the degree been modestly larger---say, $O(k \log(1/\eps))$---the resulting lower bound in \cref{thm:intro-gc0-lower-bound} would degrade with increasing $k$.

 We use \cref{lem:intro-gateapprox} to prove that $\GC^0(k)[p]$ can be approximated by proper $\F_p$ polynomials (i.e., polynomials that have Boolean outputs when the inputs are Boolean, see \cref{def:proper-polynomials}).

\begin{theorem}[Restatement of \cref{thm:probpoly}]
\label{thm:intro-probpoly}
Let $p$ be a prime.
For any constant-depth-$d$ size-$s$ $\GC^0(k)[p]$ circuit, there exists an proper polynomial $q(x)\in\F_p[x_1,\dots, x_n]$ with \[\deg(q)\le O\left(k^d + \log^{d-1}s\right)\]
such that 
\[\Pr_{x\sim U_n}[q(x) \neq  C(x)]\le 0.1.\]
\end{theorem}

Combining \cref{thm:intro-probpoly} with the well-known fact that any $\F_p$ polynomial approximating $\MAJ$ or $\Mod_q$ must have degree $\Omega(\sqrt{n})$ yields our \cref{thm:intro-gc0-lower-bound}.

\begin{remark}\label{remark:explicit-lower-bound}
\Cref{thm:intro-gc0-lower-bound} gives the least restricted class of non-monotone circuits for which we have exponential-size lower bounds for an explicit function.
In particular, the result applies to $\GC^0(k)[p] \cap \TC^0$. 
Consequently, $\GC^0(k)[p] \cap \TC^0$ currently represents the largest subclass of $\TC^0$ for which we have superpolynomial-size lower bounds. 
As a concrete example, consider $\AC^0[2]$ augmented with $\thr_k$ gates, which output $1$ if the Hamming weight of the input exceeds $k$, and $0$ otherwise.
This class lies within $\GC^0(k)[p] \cap \TC^0$, and our result also yields exponential lower bounds against it.
\end{remark}

\paragraph{Algorithmic Method}
In a celebrated result, Williams \cite{wil14acc0} used the algorithmic method to prove that there are languages in $\e^\NP$ and $\NEXP$ that require exponential-size $\ACC^0$ circuits. 
Recall that $\e^\NP$ is the class of languages that can be decided by a Turing machine in time $2^{O(n)}$ with access to an $\NP$ oracle.
In this work, we prove that there are languages in $\e^\NP$ that require exponential-size $\GCC^0(k)$ circuits, where $\GCC^0(k) \coloneqq \bigcup_{m \in \N} \GC^0(k)[m]$.\footnote{A similar result can be shown for $\NEXP$, but we focus on $\e^\NP$ because we get a stronger size-depth tradeoff.}

\begin{theorem}[{$\e^\NP \not\subseteq \GCC^0(k)$}, Restatement of \cref{thm:main-gcc0}]\label{thm:intro-main-gcc0}
    For every constant $d$, there is a $\delta >0$ such that for all $k \le O(n^{\delta/\log n})$, there is language in $\e^\NP$  that fails to have $\GCC^0(k)$ circuits of depth $d$ and size $\exp\left(\Omega(n^\delta/k) \right)$.
 \end{theorem}

As for $\GC^0(k)[p]$, one might again consider expanding a depth-$d$, size-$s$ $\GCC^0(k)$ circuit into a depth-$2d$, size-$s^{O(k)}$ $\ACC^0$ circuit by converting each $\g(k)$ gate into a CNF. Applying Williams’ lower bound in this setting would yield strictly weaker results than our \cref{thm:intro-main-gcc0}; we elaborate on this in \cref{subsec:ryan-williams}. Nonetheless, our current bound still incurs a $1/k$ loss in the exponent, and it remains an open question whether $\ACC^0$ lower bounds can be lifted to $\GCC^0(k)$ \emph{without} such degradation.

For a circuit class $\mathsf{C}$, the $\mathsf{C}$-\textsc{CircuitSAT} problem asks whether a given circuit $C \in \mathsf{C}$ has a satisfying input $x \in \{0,1\}^n$ with $C(x) = 1$.
The algorithmic method shows that faster-than-brute-force algorithms for $\mathsf{C}$-\textsc{CircuitSAT} yield circuit lower bounds for $\mathsf{C}$.
Accordingly, our lower bound for $\GCC^0$ follows from a fast algorithm for $\GCC^0(k)$-\textsc{CircuitSAT}, which we obtain by generalizing Williams’ algorithm for $\ACC^0$-\textsc{CircuitSAT}.

\begin{theorem}[{$\GCC^0(k)$-\textsc{CircuitSAT} algorithm}, Restatement of \cref{thm:gccfasteval}]
\label{thm:intro-gccfasteval}
For every $d>1$ and certain $\eps = \eps(d)$
, the satisfiability of depth-$d$ $\GCC^0(k)$ circuits with $n$ inputs and $2^{n^\eps/k}$ size can be determined in time $2^{n-\Omega(n^\delta/k)}$ for some $\delta > \eps$.
\end{theorem}

The key ingredient in our faster $\GCC^0$-\textsc{CircuitSAT} algorithm is a randomness-efficient probabilistic circuit for computing $\g(k)$ gates.
While our earlier probabilistic polynomial construction (from \cref{lem:intro-gateapprox}) yields degree-$O(k)$ approximations using $\poly(n)$ random bits, 
this construction uses too many random bits, and attempting to use it in the algorithmic method would yield a $\GCC^0$-\textsc{CircuitSAT} algorithm that is too slow. Furthermore, the randomness is used in a complicated manner, making it unclear how to convert it from a probabilistic polynomial into a probabilistic circuit.

To address this, we design a new probabilistic circuit of degree $O(k^2 \log^2 n)$ that computes any $\g(k)$ gate using only $O(k^2 \log^2 n)$ random bits. This generalizes a construction of Allender and Hertrampf \cite{allender1994uniform}, and trades a modest increase in degree for exponential savings in randomness, which is crucial for obtaining a faster algorithm for \textsc{CircuitSAT}.

\begin{theorem}[Restatement of \cref{thm:gk-depth-2}]\label{thm:intro-gk-depth-2}
    Let $q$ be a prime. Any $\G(k)$ gate on $n$ bits can be computed by a depth-2 probabilistic circuit using $O(k^2\log^2 n\log(1/\eps))$ random bits, and consists of a $\Mod_q$ of fan-in $2^{O(k^3\log^2n\log(1/\eps))}$ at the top, and $\AND$ gates of fan-in $O(k^3\log^2n\log(1/\eps))$ at the bottom layer. Furthermore, the circuit can be constructed in $2^{O(k^3\log^2 n\log(1/\eps))}$ time.
\end{theorem}

\paragraph{PAC Learning $\GC^0(k)[p]$}
Using a framework of Carmosino, Impagliazzo, Kabanets, and Kolokolova \cite{carmosino2016learning}, we give a quasipolynomial time learning algorithm for $\GC^0(k)[p]$ in the PAC model over the uniform distribution with membership queries (\cref{def:learning-model}). 

\begin{theorem}[{Learning $\GC^0(k)[p]$ in quasipolynomial time}, Restatement of \cref{thm:pac-learn}]\label{cor:intro-pac}
For every prime $p$ and $k = O(n^{1/2d})$, there is a randomized algorithm that, using membership queries, learns a given $n$-variate Boolean function $f \in \GC^0(k)[p]$ of size $s_f$ to within error $\eps$ over the uniform distribution, in time $\quasipoly(n, s_f, 1 / \eps)$. 
\end{theorem}

Using circuit lower bounds to obtain learning algorithms dates back to the seminal work of Linial, Mansour, and Nisan \cite{linial1993constant} where they gave a quasipolynomial time learning algorithm for $\AC^0$ in the PAC model over the uniform distribution (hereafter, the ``LMN algorithm''). 
One can interpret the LMN algorithm as exploiting the existence of a natural property that is useful against $\AC^0$ (in the sense of Razborov and Rudich \cite{razborov1994natural}, see \cref{def:natural-property}).

Carmosino, Impagliazzo, Kabanets, and Kolokolova \cite{carmosino2016learning} prove that for any circuit class $\mathsf{C}$ containing $\AC^0$, a natural property that is useful against $\mathsf{C}$ implies a quasipolynomial time learning algorithm for $\mathsf{C}$ in the PAC model over the uniform distribution with membership queries. 
It is not hard to show that our $\GC^0(k)[p]$ lower bound (\cref{thm:intro-gc0-lower-bound}) is natural, which implies \cref{cor:intro-pac}. 

\begin{theorem}[Informal version of \cref{thm:gc0-is-natural}]\label{thm:intro-gc0-is-natural}
   For every prime $p$ and $k = O(n^{1/2d})$, there is a natural property of $n$-variate Boolean functions that is useful against $\GC^0(k)[p]$ circuits of depth $d$ and of size up to 
   $\exp\left(\Omega(n^{1/2d})\right)$.
\end{theorem}

\paragraph{A New Multi-Output Multi-Switching Lemma For $\GC^0(k)$}
In \cref{subsec:quantum-intro}, we describe new separations between quantum and classical circuits. 
One such separation relies on a new multi-switching lemma for $\GC^0(k)$ tailored to handle circuits with multiple outputs. 
The details of the switching lemma are quite technical, and we refer the interested reader to \cref{subsec:muli-switching} for details. 
The general switching lemma is stated in \cref{thm:multioutputlemma}.
We heavily rely on Kumar's multi-switching lemma \cite{kumar2023tight}, which we use in a black-box manner.

This strengthened switching lemma allows us to show that $\GC^0(k)$ circuits have exponentially small correlation with a particular search problem that can be solved by constant-depth quantum circuits. While a similar separation could be obtained using only Kumar’s switching lemma, our new version yields significantly stronger bounds on the correlation.

\subsubsection{Improved Quantum-Classical Separations}\label{subsec:quantum-intro}
A central goal in quantum complexity theory is to identify problems that are tractable for quantum computers but intractable for classical ones. 
One way to formalize this goal is to show that $\BQP$ (Bounded-Error Quantum Polynomial Time) strictly contains $\PTIME$ (Polynomial Time). 
Alas, even showing that $\PTIME$ is strictly contained in $\PSPACE$ is currently beyond the reach of complexity theory. 

One can separate $\BQP$ from $\PTIME$ \emph{conditionally}, for example, under the assumption that there is no polynomial-time algorithm for factoring integers \cite{shor1997polynomial,regev2024efficient}. 
There is also a long line of research that separates quantum and (randomized) classical computation in the black-box model \cite{bernstein1993quantum,simon1997power,aaronson2015forrelation}. 

Yet another option (and the one that is most relevant to this work) is to look at restricted models of computation. 
In a groundbreaking work, Bravyi, Gosset, and K\"onig \cite{bravyi2018quantum} exhibited a search problem that is solvable by $\QNC^0$ (constant-depth bounded-fan-in quantum circuits), but is hard for $\NC^0$ (constant-depth bounded-fan-in classical circuits). This is an \emph{unconditional separation} between a quantum and classical model of computation.

Since then, there has been a lot of progress \cite{watts2019exponential,legall2019average,coudron2019trading,grier2019interactive,bravyi2020quantum,grier2021interactive,grilo2024power}. We briefly summarize the state of the art prior to this work: there is a decision problem that separates $\BQLOGTIME$ (\cref{def:bqlogtime}) and quasipolynomial-size $\AC^0$ \cite{raz2022oracle}; a search problem that separates $\QNC^0$ and exponential-size $\AC^0$ \cite{watts2019exponential}; and a search problem that separates $\QNC^0\mathsf{/qpoly}$ and polynomial-size $\AC^0[p]$ for any prime $p$ \cite{watts2019exponential, grilo2024power}. 
Recall that $\QNC^0\mathsf{/qpoly}$ is the class of constant-depth bounded-fan-in quantum circuits that start with a quantum advice state, i.e., an input-independent quantum state of choice.
Grier and Schaeffer \cite{grier2019interactive} also obtain a separation between $\QNC^0$ and exponential-size $\AC^0[p]$ for an interactive problem.
Finally, Bravyi, Gosset, K\"onig, and Temamichel \cite{bravyi2020quantum} and Grier, Ju, and Schaeffer \cite{grier2021interactive} showed that these separations hold even when the quantum circuits are subject to certain types of noise.\footnote{Watts and Parham \cite{watts2024unconditionalquantumadvantagesampling} also studied unconditional separations for input-independent sampling problems. In this work, we focus on computational problems that have inputs and outputs.} 

Building on this line of work, we can subsume all previously known separations between quantum and classical circuits. 
In particular, we show that even if we allow arbitrary unbounded fan-in local circuits (i.e., $\GC^0$ and its extensions), these circuits are still not powerful enough to simulate constant-depth quantum circuits. 
We re-use the problems used to obtain the above quantum-classical separations; our contribution is to strengthen all of the lower bounds to hold for $\GC^0(k)$ or $\GC^0(k)[p]$.
All of our separations are exponential, meaning that the problems can be solved with polynomial-size quantum circuits but require exponential-size classical circuits.
Previously the best separation between $\QNC^0\mathsf{/qpoly}$ and polynomial-size $\AC^0[p]$ circuits. 
In the remainder of this subsection, we state our separations in more detail.

\paragraph{$\BQLOGTIME$ vs. $\GC^0$}
In \cref{subsec:raz-tal}, we exhibit a promise problem that separates $\BQLOGTIME$ from $\GC^0(k)$. 

\begin{theorem}[Restatement of \cref{cor:bqlogtime-not-in-gac0}]\label{cor:intro-bqlogtime-not-in-gac0}
There is a promise problem in $\BQLOGTIME$ (\cref{def:bqlogtime}) that is not solvable by constant-depth $\GC^0(k)$ for $k = \frac{O(n^{1/4d})}{(\log n)^{\omega(1)}}$ and size $\quasipoly(n)$. 
\end{theorem}

By well-known reductions, this implies an oracle relative to which $\BQP$ is not contained in the class of languages decided by uniform $\GC^0$ circuit families.

\begin{corollary}[Restatement of \cref{thm:oracle-separation-bqp}]\label{thm:intro-oracle-separation-bqp}
There is an oracle relative to which $\BQP$ is not contained in the class of languages decidable by uniform families of circuits $\{C_n\}$, where for all $n \in \N$, $C_n$ is a size-$2^{n^{O(1)}}$ depth-$d$ $\GC^0(k)$ circuit with $k \in \frac{2^{n/4d}}{n^{\omega(1)}}$. 
\end{corollary}

Raz and Tal \cite{raz2022oracle} showed that $\BQLOGTIME \not\subseteq \AC^0$, which implied an oracle relative to which $\BQP$ is not contained in the class of languages decided by uniform families of size-$2^{n^{O(1)}}$ constant-depth $\AC^0$ circuits. It is well-known that this class is precisely the polynomial hierarchy $\PH$. Hence, because $\GC^0(k)$ contains $\AC^0$ (and can even compute functions that require exponential-size $\AC^0$ circuits), \cref{thm:intro-oracle-separation-bqp} is a generalization of the relativized separation of $\BQP$ and $\PH$.

One reason Raz and Tal \cite{raz2022oracle} is such a striking result is that it shows even the enormous power of $\PH$ fails to simulate quantum computation in a relativizing way. 
This is made more precise in the beautiful follow-up work of Aaronson, Ingram, and Kretschmer \cite{aaronson2022acrobatics} who show (among many other results) that there is an oracle relative to which $\PTIME = \NP$ but $\BQP = \PTIME^{\# \PTIME}$. 
In words, they show that even in a world where $\NP$ is easy, $\BQP$ can still be extremely powerful. 
Our oracle separation result complements these results (and relies on Raz and Tal). 

We give one concrete implication of \cref{thm:oracle-separation-bqp}.
Namely, we show that there is an oracle relative to which $\BQP$ is outside of hierarchies of counting classes, where the counting classes can count whether there are a small number of accepting witnesses. 
This is perhaps surprising because $\BQP \subseteq \PP$ relative to all oracles \cite{adleman1997quantum}. Hence, we show that it is actually necessary to count a larger number of witnesses to simulate $\BQP$ in a relativizing way. The counting classes are defined in \cref{def:Counting-quant,def:biased-ch}, and the oracle separation is given in \cref{cor:biased-ch-separation}.

\paragraph{$\QNC^0$ vs. $\GC^0$} In \cref{subsec:qnc-gc0}, we exhibit a search problem that separates $\QNC^0$ from $\GC^0(k)$.
Our separation is based on the 2D Hidden Linear Function (\tdhlf) problem (\cref{def:tdhlf}) introduced by Bravyi, Gosset, and K\"onig \cite{bravyi2018quantum}.

\begin{theorem}[Restatement of \cref{thm:watts-main}]\label{thm:intro-watts-main}
The \tdhlf\ problem (\cref{def:tdhlf}) on $n$ bits cannot be solved by a constant-depth-$d$ size-$\exp(O(n^{1/4d}))$ $\GC^0(k)$ circuit with $k=O(n^{1/4d})$. 
Furthermore, for the same value of $k$, there exists an (efficiently samplable) input distribution on which any $\GC^0_d(k)$ circuit (or $\GC^0_d(k)\mathsf{/rpoly}$ circuit) of size at most $\exp(n^{1/4d})$ only solves the \tdhlf\ problem with probability at most $\exp(-n^c)$ for some $c > 0$. 
\end{theorem}

\cref{thm:intro-watts-main} generalizes the separation between $\QNC^0$ and $\AC^0$ obtained by Watts, Kothari, Schaeffer, and Tal \cite{watts2019exponential}. The proof requires a new multi-output multi-switching lemma for $\GC^0(k)$, which we prove in \cref{subsec:muli-switching}.

Using the frameworks developed by Bravyi et al.\ \cite{bravyi2020quantum} and Grier et al.~\cite{grier2021interactive}, we show in \cref{thm:noisy-qnc0-gc0} that this separation holds even when the quantum circuits are subjected to certain types of noise.

\paragraph{$\QNC^0\mathsf{/qpoly}$ vs. $\GC^0[p]$}
In \cref{subsec:gc2,subsec:gcp}, we exhibit a family of search problems that separates $\QNC^0\mathsf{/qpoly}$ from $\GC^0(k)[p]$ for all primes $p$. 
The family of search problems is a generalization of the Parity Bending problem introduced by Watts, Kothari, Schaeffer, and Tal \cite{watts2019exponential} and was also studied in a recent work of Grilo, Kashefi, Markham, and Oliveira \cite{grilo2024power}.

\begin{theorem}[Restatement of \cref{thm:intro-gcp}]{\label{thm:intro-gcp}}
    For any prime $p$, there is a search problem that is solvable by $\QNC^0\mathsf{/qpoly}$ with probability $1-o(1)$, but any $\GC^0(k)[p]\mathsf{/rpoly}$ circuit of depth $d$ and size at most $\exp\left(O(n^{1/2.01d}) \right)$ with $k = O(n^{1/2d})$ cannot solve the search problem with probability exceeding $n^{-\Omega(1)}$.
\end{theorem}

Previously the best separations were between polynomial-size $\QNC^0$ and \emph{polynomial-size} $\AC^0[p]$ obtained in the works of Watts et al.~\cite{watts2019exponential} and Grilo et al.~\cite{grilo2024power}.
Our \cref{thm:intro-gcp} is a separation between polynomial-size $\QNC^0$ and \emph{exponential-size} $\GC^0(k)[p]$.  

\paragraph{Interactive $\QNC^0$ vs. $\GC^0(k)[p]$}
Grier and Schaeffer \cite{grier2019interactive} studied quantum-classical separations that can be obtained in certain interactive models.
Among some conditional results, they obtain an unconditional separation between $\QNC^0$ and $\AC^0[p]$ for all primes $p$. 
We generalize their separation to $\GC^0(k)[p]$.

\begin{theorem}[Restatement of \cref{thm:grier-schaeffer}]\label{thm:intro-grier-schaeffer}
Let $k = O(n^{1/2d})$.
There is an interactive task that $\QNC^0$ circuits can solve that depth-$d$, size-$s$ $\GC^0(k)[p]$ circuits cannot for $s \leq \exp\left(O(n^{1/2.01d})\right)$.
\end{theorem}

\subsection{Open Problems}\label{subsec:open-problems}

Combined with the work of Kumar \cite{kumar2023tight}, we now know that $\AC^0$ size lower bounds from the combinatorial technique of switching lemmas, as well as $\AC^0[p]$ lower bounds using the algebraic technique of probabilistic polynomials, both lift \emph{losslessly} to $\GC^0$ and $\GC^0(k)[p]$, respectively. It is extremely surprising to us that both techniques, while extremely different in flavor, generalize so cleanly to $\g(k)$ gates. 
This observation raises many questions about how $\g(k)$ gates can help us understand the limitations of our lower bound techniques.

\begin{itemize}
    \item Do $\g(k)$ gates exactly capture the switching lemma technique as well as the probabilistic polynomial technique? 
    This would let us know whether there is an even more general class of gates that capture the power of these techniques.
    
    \item Can we use $\g(k)$ gates (or its generalizations derived from the last item) to show barrier results for current lower bounds we have? 
    For example, implementing explicit functions in $\GC^0(k)$ or $\GC^0(k)[p]$ would demonstrate a limitation on the size lower bounds achievable for $\AC^0$ or $\AC^0[p]$ via switching lemmas or the polynomial method.

    \item Can lower bounds for $\e^{\NP}$ using Williams' algorithmic method be lifted losslessly from $\ACC^0$ to $\GCC^0$? 
\end{itemize}

There are also general questions about how $\GC^0$ and their counterparts fit in the landscape of circuit classes.
\begin{itemize}
    \item How do $\GC^0(k)$, $\GC^0(k)[p]$, and $\GCC^0(k)$ compare to more classical circuit classes like $\NC^1$ and  $\TC^0$? 
    We know that when $k=n$, $\GC^0(k)$ can compute any function, and when $k = 1$, $\GC^0(k) = \AC^0$. What is the smallest $k$ such that $\TC^0\subset \GC^0(k)$? We know this is true when $k\ge n/2$, but is it true for smaller $k$? 
    Similar questions can be raised for $\GC^0(k)[p]$.

    \item Can we get stronger quantum-classical separations? 
    Specifically, can we obtain separations between $\QNC^0$ and $\GC^0(k)[p]$ without giving the quantum circuit an advice state? 
    
    \item \cite{kumar2023tight} gave a natural subclass of $\g(k)$ consisting of biased linear threshold gates. 
    Are there other natural gates contained in $\g(k)$?

\end{itemize}

\paragraph{Concurrent Work}
An independent and concurrent work of Hsieh, Mendes, Oliveira, and Subramanian \cite{hsieh2024concurrent} overlaps with our work in one way. 
They give an exponential separation between $\GC^0(k)$ and $\QNC^0$, which is essentially the same as our separation (\cref{thm:watts-main}).\footnote{Hsieh et al.\ denote $\GC^0(k)$ by $\mathsf{bPTF}^0[k]$.} 
Like us, they also prove a new muli-output multi-switching lemma for $\GC^0(k)$ (\cref{thm:multioutputlemma}) to obtain their separation. 
The similarity in our arguments comes from the fact that we both use the exponential separation between $\AC^0$ and $\QNC^0$ of Watts, Kothari, Schaeffer, and Tal  \cite{watts2019exponential} as a starting point. 

Hsieh et al.\ also show that their separation holds if the quantum circuits are subjected to a certain noise model, which we also do in \cref{thm:noisy-qnc0-gc0}. 
This noise-robustness result follows from applying the framework introduced by Bravyi, Gosset K\"onig, and Temamichel \cite{bravyi2020quantum} and further developed by Grier, Ju, and Schaeffer \cite{grier2021interactive}. Hsieh et al.\ also study extending this framework to prime-dimensional qudits. 

\section{Preliminaries}
We presume the reader is familiar with common concepts in the theory of computation (circuit complexity and quantum computing, in particular).
All prerequisite knowledge can be found in standard textbooks such as \cite{goldreich2008computational,arora2009computational,nielsen2002quantum}.

We obey the following notation and conventions throughout.
For a positive integer $n$, $[n]\coloneqq \{1,\dots,n\}$.
For us, the natural numbers do not include $0$, i.e., $\N \coloneqq \{1, 2, 3, \ldots\}$.
Define $\binom{n}{\leq k} \coloneqq \sum_{i=0}^k \binom{n}i$.
For $S\subseteq[n]$ and $x \in \{0,1\}^n$, define $x^S \coloneqq \prod_{i\in S}x_i$.
Let $\quasipoly(n)$ denote all functions that have an upper bound of the form $2^{O(\log^c n})$ for some constant $c$.

We denote the Hamming weight of a string $x\in\bitz^n$ as $\abs{x} = \sum_{i} x_i$. 
More generally, for $x \in \F_q^n$ (for some prime $q$), $\abs{x} = \sum_i x_i \pmod q$.
The \emph{Hamming distance} between $x, y \in \bitz^n$ is $\Delta(x,y) = \abs{\{ i \in [n] : x_i \neq y_i \}}$.
The \emph{Hamming ball of radius $k$} is the set $\{x\in\bitz^n: |x|\le k\}$, and \emph{Hamming ball of radius $k$ centered at $c$} is the set $\{ x \in \bitz^n : \Delta(x, c) \leq k \}$.

For a distribution $\calD$ over support $S$, $x \sim \calD$ denotes sampling an $x \in S$ according to the distribution $\calD$.
For a set $S$, we denote drawing a sample $s \in S$ uniformly at random by $s \sim S$.
$U_n$ denotes the uniform distribution over length-$n$ bit strings. 
For a distribution $\calD$ and a function $f$, $\E[f(\calD)] \coloneqq \E_{x \sim \calD}[f(x)]$.
For two discrete distributions $p$ and $q$ supported on $S$, the total variation distance (also called the statistical distance) is defined as $\frac{1}{2}\sum_{s \in S}\abs{p(s) - q(s)}$.

We also use Fermat's little theorem. 

\begin{theorem}[Fermat's little theorem]\label{thm:fermat}
For any integer $a \not\equiv 0 \pmod p$ for a prime $p$, $a^{p-1} \equiv 1 \pmod p$.
\end{theorem}

All circuit classes studied in this work are constant depth, and $d$ always denotes a constant.
Circuits are comprised of layers of gates. When we refer to the ``top'' of a classical circuit, we are referring to the last layer of the circuit. In particular, for a Boolean-valued circuit, the top is a single gate. The ``bottom'' of a circuit refers to the first layer of gates.

For an integer $m$, the $\Mod_m$ gate is the unbounded fan-in Boolean gate that outputs $0$ iff the sum of the input bits is congruent to $0 \pmod m$.  
The $\MAJ$ gate computes the majority function, i.e., the unbounded fan-in Boolean gate that outputs $1$ iff the majority of the input bits are $1$.
The $\thr^k$ gate is the unbounded fan-in Boolean gate that outputs $1$ iff the Hamming weight of the input is $> k$.

Recall the following well-studied circuit classes:
\begin{itemize}
    \item $\NC^i$: $O(\log^i n)$-depth circuits of bounded fan-in $\AND$, $\OR$, and $\NOT$ gates. 
    \item $\AC^i$: $O(\log^i n)$-depth circuits of unbounded fan-in $\AND$, $\OR$, and $\NOT$ gates. 
    \item $\AC^i[p]$: $O(\log^i n)$-depth circuits of unbounded fan-in $\AND$, $\OR$, $\NOT$, and $\Mod_p$ gates. 
    \item $\ACC^i$: The union of $\AC^i[m]$ for all $m$. 
    \item $\TC^i$: $O(\log^i n)$-depth circuits of unbounded fan-in $\AND, \OR, \NOT,$ and $\MAJ$  gates. 
    \item $\QNC^i$: $O(\log^i n)$-depth quantum circuits of bounded fan-in quantum gates. 
    \item $\mathsf{SIZE}(f(n))$: fan-in-2 Boolean circuits of size $O(f(n))$. 
\end{itemize}
$\NC \coloneqq \bigcup_i \NC^i$, and $\AC$ and $\TC$ are defined analogously. It is known that $\NC = \AC = \TC$. The size of a circuit is the number of gates in the circuit besides $\NOT$ gates.
We always specify the circuit size when relevant; however, if the size is not explicitly mentioned, it should be assumed to be polynomial.

A search problem (also called a relation problem or relational problem) is a computational problem with many valid outputs, as opposed to a function problem which only has one valid output for each input. 
A two-round interactive problem is a computational problem where in the first round you are given an input and produce an output and in the second round, you produce another input and output. The correctness of an interactive algorithm is based on the entire transcript of the interaction, and a computational device solving an interactive problem gets to keep state from the first round during the second round. 

In a common abuse of notation we use e.g. $\AC^0$ or $\GC^0(k)[p]$ to interchangeably talk about a type circuit or a class of (decision/relation/interactive) problems, where the context clarifies what we are referring to.

We also will use probabilistic circuits.

\begin{definition}\label{def:prob-circuit}
A \emph{probabilistic circuit} that computes a function $f:\bitz^n\to\bitz$ is a circuit $C$ that takes input $x\in\bitz^n$ and uniformly random bits $r$, and satisfies the property that for all $x\in \bitz^n$, \[\Pr_{r}[C(x,r) \neq f(x)]\le \eps.\] 
\end{definition}

\subsection{The \texorpdfstring{$\g(k)$}{G(k)} Gate}\label{subsec:gk-crash-course}

The $\g(k)$ gate is an unbounded fan-in gate with the following behavior. The circuit designer chooses a Hamming ball $B_{k,c}$ of radius $k$ centered at $c$. On input $x \in \bitz^n$, if $x \in B_{k,c}$, the $\g(k)$ gate can compute any function $f$ of the circuit designer's choosing. 
Otherwise, the $\g(k)$ gate outputs a constant $c \in \{0,1\}$ of the designer's choosing.
We define $\GC^0$ as the class of constant-depth circuits comprised of $\G(k)$ gates.

One can equivalently define the $\g(k)$ gate as an unbounded fan-in gate that computes within the Hamming ball of radius $k$ centered at $0^n$. This is because one can use this gate to implement $\NOT$. Then one can shift the center of the Hamming ball by appropriately applying $\NOT$ gates to the input bits.
We typically use this definition in our proofs, because it yields cleaner arguments.

The value of $k$ controls the power of the $\g(k)$ gate. 
When $k$ is a constant, it is easy to see that a single $\g(k)$ gate can be computed by a depth-two polynomial-size $\AC^0$ circuit. 
When $k = n$, a single $\g(k)$ gate can compute any function.
Much of the landscape between $k = O(1)$ and $k=n$ is not yet understood, which we discuss further in \nameref{subsec:open-problems}.

We also emphasize that the circuit designer can use the $\g(k)$ gate however they like. 
On the tamer side, the $\g(k)$ gate can, e.g., evaluate parity on $k$ bits or majority on $2k$ bits, and, on the wilder side, it can, e.g., evaluate uncomputable functions like the halting function (with the caveat that it must output a constant if the input is not within the relevant Hamming ball). 

The $\g(k)$ gates capture natural gates as special cases. For example, $\g(k)$ gates naturally generalize $\AND$ and $\OR$ gates to biased linear threshold gates. 
Let $\theta, w_1, \dots, w_n \in \R$, with the $w_i$'s sorted such that $\abs{w_1}\leq \abs{w_2}\leq \dots \leq \abs{w_n}$. Let $f(x) = \sgn(\sum_{i=1}^n w_ix_i - \theta)$. If $\sum_{i > k} \abs{w_i} - \sum_{i \leq k} \abs{w_i} < \abs{\theta}$, then $f$ can be computed by a $\g(k)$ gate \cite[Theorem A.1]{kumar2023tight}.
Kumar \cite{kumar2023tight} showed that circuits comprised of biased linear threshold gates interpolate between $\AC^0$ and $\TC^0$ as the parameter $k$ varies. 
We note that there is a connection between linear threshold functions and neural networks that dates back to the 1940s \cite{mcculloch1943logical}, and there is a precise connection between feed-forward neural networks and $\TC^0$ circuits \cite{muroga1971threshold, maass1991computational} (see also \cite[Section 2.5.1]{arunachalam2021quantum}). 
Circuits with $\g(k)$ gates capture a subset of neural networks whose activation functions are \emph{biased} linear threshold functions.

\section{Approximating \texorpdfstring{$\g(k)$}{G(k)} Gates by Low-Degree Polynomials}\label{sec:raz-smo}

We show that any $\g(k)$ gate can be approximated by proper low-degree polynomials. 
To discuss our results in more detail, we must introduce some terminology.

\begin{definition}[Proper polynomial]\label{def:proper-polynomials}
Let $q$ be a prime number. 
A polynomial $p(x) \in \F_q[x_1,\dots,x_n]$ is proper when $p(x) \in \{0,1\}$ for all inputs $x \in \{0,1\}^n$.
\end{definition}

\begin{definition}[$\eps$-approximating polynomial]\label{def:approx-poly} 
An $\eps$-approximate polynomial for a function $f:\bitz^n \to \bitz$ is a proper polynomial $p$ such that 
\[
\Pr_{x \sim U_n}[f(x) \neq p(x)] \leq \eps.
\]
\end{definition}

\begin{definition}[$\eps$-probabilistic polynomial]\label{def:prob-poly}
   An $\eps$-probablistic polynomial of degree $d$ for a function $f: \bitz^n \to \bitz$ is a distribution $\calP$ over proper polynomials of degree $\le d$ such that for every $x\in\bitz^n$, \[\Pr_{p\sim \calP}[p(x) \neq f(x)]\le \eps.\] 
\end{definition}

In \cref{subsec:raz-smo}, we show that $\g(k)$ gates can be approximated by low-degree polynomials.
As a consequence, we show that any $\GC^0(k)[q]$ circuit can be approximated by low-degree polynomials, generalizing the Razborov-Smolensky polynomial method \cite{razborov1987lower, smolensky1987algebraic} for $\AC^0[q]$ to $\GC^0(k)[q]$. 
This allows us to prove circuit lower bounds for $\GC^0(k)[q]$; we discuss this application and others in \cref{sec:apps-to-classical,sec:quantum}. 

In \cref{subsec:gcc-approx}, we construct probabilistic polynomials for $\g(k)$ gates that use very few bits of randomness. 
The randomness-efficiency of our construction will be essential to invoke the algorithms-to-lower-bounds technique of Williams \cite{wil14acc0}, which we do in \cref{subsec:ryan-williams}.

\subsection{Approximating \texorpdfstring{$\gc^0[p]$}{GC0[p]} by Low-Degree Polynomials}\label{subsec:raz-smo}
We show that size-$s$ $\GC^0(k)[q]$ circuits can be well-approximated by $\F_q$-polynomials of degree $\poly(k, \log s )$. To do so, we need a standard lemma stating that one can interpolate a truth table on a radius $k$ Hamming ball by a degree-$k$ polynomial. 
We give a proof for convenience.

\begin{lemma}
\label{lem:interpolate}
    For any $f:\bitz^n\to\F_q$ and prime $q$, there exists a unique $\F_q$-polynomial $p$ with $\deg(p)\le k$ such that for all $x\in\bitz^n$ with $\abs{x}\le k$, $f(x) = p(x)$. Furthermore, this polynomial can be constructed in $n^{O(k)}$ time.
\end{lemma}

\begin{proof}
Consider the $\mathbb{F}_q$-linear system of equations given by $\sum_{\lvert S \rvert \le k} c_S a^S = f(a)$ for each $a \in \{0,1\}^n$ such that $\lvert a \rvert \leq k$. These equations are linearly independent, and since the number of equations equals the number of variables, there is a unique set of coefficients $\{c_S\}$ that satisfies this system. Therefore, the polynomial with these coefficients, $p(x) \coloneqq \sum_{\lvert S \rvert \le k} c_S x^S$, is the desired polynomial. Furthermore, these coefficients can be retrieved in $n^{O(k)}$ time via Gaussian elimination on the $\binom{n}k$ linear equations.
\end{proof}

Next, we prove a technical lemma that says there are low-degree probabilistic polynomials for $\g(k)$ gates. 
Our construction uses probabilistic polynomials for $\thr^k$, where $\thr^k$ is the unbounded fan-in gate that outputs $1$ iff the Hamming weight of the input is $> k$.

\begin{lemma}[{\cite[Theorem 3]{srinivasan_et_al:LIPIcs.FSTTCS.2019.28}}]\label{lem:thr-prob-poly}
For any prime $q$, there is an $\eps$-probabilistic $\F_q$ polynomial of degree $O(\sqrt{k\log(1/\eps)} + \log(1/\eps))$ that computes $\thr^k$.
\end{lemma}

\begin{lemma}
\label{lem:gateapprox}
    For any $\g(k)$ gate $G$ of fan-in $n$ and constant prime $q$, there is an $\eps$-probabilistic $\F_q$-polynomial of degree $O(k + \log(1/\eps))$ computing $G$.\footnote{By \emph{constant prime}, we mean that $q$ does not grow with $n$. In particular, the $O(\cdot)$ expressions may hide factors depending on $q$.}
\end{lemma}
\begin{proof}
Because $G\in \g(k)$, we can express its behavior as \[G(x) = \begin{cases}c & |x| > k \\ f(x) & |x|\le k \end{cases} \] for an arbitrary $f:\bitz^n\to\bitz$ and $c\in \bitz$.  
By \Cref{lem:interpolate}, there exists a (deterministic) degree-$k$ polynomial $p(x)$ that matches $f(x)-c$ when $|x|\le k$. Furthermore, by \cref{lem:thr-prob-poly}, there exists a probabilistic polynomial $Q(x)$ of degree $O(\sqrt{k\log(1/\eps)} + \log(1/\eps))$ that computes $\thr^k$ with error $\eps$.

Consider the probabilistic polynomial 
\[P(x) \coloneqq (p(x)(1-Q(x)) + c)^{q-1}.\] Notice that $\deg(P) = O(k + \sqrt{k\log(1/\eps)} + \log(1/\eps)) = O(k + \log(1/\eps))$, and that the support of $P$ is over proper polynomials by Fermat's Little Theorem (\cref{thm:fermat}). 

When $\abs{x} \le k$, observe that $\Pr[Q(x) = 0]\ge 1-\eps$.  Hence, with probability at least $1-\eps$, we have \[P(x) = (p(x) + c)^{q-1} = f(x)^{q-1} = f(x) = G(x),\]
where we use the fact that $p(x) = f(x) - c$ when $\abs{x} \le k$ and the third equality follows from Fermat's Little Theorem (\cref{thm:fermat}).

When $\abs{x} > k$, $\Pr[Q(x) = 1] \ge 1-\eps$. Therefore, with probability at least $1-\eps$, we have \[P(x) = c^{q-1} = c = G(x).\] Thus in either case, it follows that $P$ computes $G$ with error $\le \eps$.
\end{proof}

We can also show that the degree of the probabilistic polynomial in \cref{lem:gateapprox} is optimal.

\begin{lemma}\label{lem:prob-poly-lower-bound}
There exists a $\g(k)$ gate that requires probabilistic degree $\Omega(k + \log(1/\eps))$.
\end{lemma}
\begin{proof}
To show a probabilistic degree lower bound of $d$ against a $\g(k)$ gate $G$, it suffices by Yao's minimax principle to construct a hard distribution $\calD$ supported over $\bitz^n$ such that for any degree-$d$ polynomial $p$, $\Pr_{x\sim \calD}[p(x) \neq G(x)] > \eps$. We will show a lower bound of $\max(k/2,\log(1/\eps)) = \Omega(k + \log(1/\eps))$ by showing there exists a gate $\g(k)$ gate $G$ and hard distributions $\calD_1$ and $\calD_2$ such that any polynomial $\eps$-approximating $G$ under $\calD_1$ requires degree $\ge k/2$, and any polynomial $\eps$-approximating $G$ under $\calD_2$ requires degree $\lfloor \log(1/\eps)\rfloor$. We will use the probabilistic method and pick $G\in \g(k)$ uniformly at random.

We will set $\calD_1$ to be uniform over all strings $x$ with $|x|\le k$.
For a fixed polynomial $p$, we see by a Chernoff bound that 
\[\Pr_{G}\left[\Pr_{x\sim\calD_1}[p(x)\neq G(x)] < \eps\right] \le e^{-\frac{1}4\binom{n}{\le k}}.\] 
Union bounding over all $q^{\binom{n}{\le k/2}}$ degree-$(k/2)$ $\F_q$-polynomials tells us that $G$ cannot be computed by any degree-$k/2$ polynomial with error $\eps$ with probability 
\[
\ge 1-q^{\binom{n}{\le k/2}}
\cdot e^{-\frac{1}4\binom{n}{\le k}} \ge 1- e^{-\Omega\left(\binom{n}{\le k}\right)}.
\]

Now let $\calD_2$ be the sample $1^k0^{n-k-\lfloor \log(1/\eps)\rfloor}y$, where $y\sim U_{\lfloor \log(1/\eps)\rfloor}$. Notice that with probability $1/2$, $G'(y):= G(1^k0^{n-k-\lfloor \log(1/\eps)\rfloor}y)$ is either an $\AND$ or $\OR$ up to negation (and with the other $1/2$ probability it is constant). Furthermore, if there exists even one $y$ such that \[p(1^k0^{n-k-\lfloor \log(1/\eps)\rfloor}y)\neq G(1^k0^{n-k-\lfloor \log(1/\eps)\rfloor}y),\] then $\Pr_{x\sim\calD_2}[p(x)\neq G(x)]\ge \frac{1}{2^{\lfloor \log(1/\eps)\rfloor}} > \eps.$ Therefore, any polynomial $p$ $\eps$-approximating $G$ under $\calD_2$ must have the restricted polynomial $p'(y) := p(1^k0^{n-k-\lfloor \log(1/\eps)\rfloor}y)$ exactly compute $G'$. Conditioning on $G$ being an $\AND/\OR$ up to negation, we note that the $\AND/\OR$ over $m$ variables has $\F_q$-degree $m$, and so $\deg(p)\ge\deg(p') = \deg(G') = \lfloor \log(1/\eps)\rfloor$.

Consequently by a union bound, a randomly picked $G$ will require degree $k/2$ to approximate under $\calD_1$ and degree $\lfloor \log(1/\eps)\rfloor)$ to approximate under $\calD_2$ with probability $\ge \frac{1}2  - e^{-\Omega\left(\binom{n}{\le k}\right)} > 0$. Hence, our desired $G$ exists, and the lower bound holds.
\end{proof}

We are now ready to show the main theorem. 
Namely, that proper low-degree $\F_q$ polynomials can approximate 
any $\GC^0(k)[q]$ circuit. 

\begin{theorem}
\label{thm:probpoly}
Let $q$ be a constant prime.
For any $\GC^0_d(k)[q]$ circuit $C$ of size $s$, there exists a proper polynomial $p(x)\in\F_q[x_1,\dots, x_n]$ with $\deg(p)\le O\left((k+\log(1/\eps)(k+\log(s/\eps))^{d-1}\right)$ such that 
\[\Pr_{x\sim U_n}[p(x) \neq  C(x)]\le \eps.\]
\end{theorem}
\begin{proof}
We will construct a probabilistic low-degree polynomial for each gate in the circuit. By composing these polynomials according to the structure of the circuit, we will obtain a probabilistic low-degree polynomial for the entire circuit. 
This final probabilistic polynomial is the low-degree polynomial approximating the circuit.

For each gate $G\in C$ with fan-in $n_G$, we will associate a probabilistic low-degree polynomial $P_G$ that approximates it. If $G = \NOT$, then $n_G = 1$ and we set $P_G(x) = x + 1$. If $G = \Mod_q$, then we set $P_G(x) = \sum x_i$. If $G\in \g(k)$ and $G$ is not the top gate, we will set $P_G$ to be the probabilistic polynomial with degree $O(k+\log(2s/\eps))$ that computes $G$ with error probability at most $\eps/s$, as given by \Cref{lem:gateapprox}. Otherwise if $G$ is the top gate, we will set $P_G$ to be the probabilistic polynomial with degree $O(k+\log(2/\eps))$ that computes $G$ with error probability at most $\eps/2$.
Note that for all gates $G$ below the top gate in the circuit and all inputs $x$, $\Pr[G(x)\neq P_G(x)]\le \eps/2s$, and $\deg(P_G)\le O(k+\log(s/\eps))$, whereas the top gate $G$ satisfies  $\Pr[G(x)\neq P_G(x)]\le \eps/2$ with $\deg(P_G)\le O(k+\log(2/\eps))$. 

Now, if we replace each gate $G$ with the probabilistic polynomial $P_G$ and compose the polynomials together, we get a probabilistic polynomial $P$ with $\deg(P)\le O((k+\log(2/\eps))(k+\log(2s/\eps))^{d-1})$. 
Fix an input $x$ to the circuit. Let $x_G\in \bitz^{n_G}$ be the bits of $x$ read by gate $G$. 
If $P_G(x_G) = G(x_G)$ for all gates $G$ in $C$, then $P(x) = C(x)$. Therefore, by a union bound and accounting for the larger error on the top gate, we have that
\[\Pr_{p\sim P}[p(x)\neq C(x)]\le \sum_{G}\Pr[P_G(x_G) = G(x_G)]\le \frac{\eps}2 + \frac{\eps}{2s}\cdot s = \eps.\] 
Since $x$ was arbitrary, the above holds for \emph{all} $x$, which means 
\[\eps \ge \E_{x}[\Pr_{p\sim P}[p(x)\neq C(x)]] = \E_{p\sim P}[\Pr_x[p(x)\neq C(x)]].\] 
Hence, by an averaging argument, there exists a polynomial $p$ in the support of $P$ that agrees with $C(x)$ on all but an $\eps$ fraction of inputs.
\end{proof}

\subsection{Probabilistic Circuits for \texorpdfstring{$\G(k)$}{G(k)} Gates With Very Few Random Bits}\label{subsec:gcc-approx}

We prove that $\g(k)$ gates can be approximated by a randomness-efficient depth-$2$ probabilistic circuit (\cref{def:prob-circuit}) comprised of $\AND$ gates of small fan-in in the bottom layer and a $\Mod_q$ gate for any prime $q$ in the top layer, generalizing a prior work of Allender and Hertrampf \cite{AHdepthreduction94}.
This result will be crucial for invoking the lower bound technique of Williams \cite{wil14acc0} as we do in \cref{subsec:ryan-williams}.   

Depth-$2$ probabilistic circuits with $\AND$ gates at the bottom and a $\Mod_q$ gate at the top are an instance of probabilistic $\F_q$-polynomials. In particular, if the $\AND$ gates all have fan-in at most $d$, then these depth-$2$ circuits are probabilistic polynomials of degree $d$.
Therefore, one can view the main result of this subsection (\cref{thm:gk-depth-2}) as a version of \Cref{lem:gateapprox} that uses very few bits of randomness.
To compare, our \cref{lem:gateapprox} uses $\poly(n)$ random bits to construct a probabilistic $\F_q$-polynomial of degree $O(k)$ for any $\g(k)$ gate. 
In this section, we use $O(k^2 \log^2 n)$ random bits to construct a probabilistic $\F_q$-polynomial of degree $O(k^3 \log^2 n)$ for any $\g(k)$ gate.
So, at the cost of a $\poly(k ,\log n)$ factor in the degree, we can obtain an exponential savings in the number of random bits used in our construction.

Our construction uses the following theorem of Valiant and Vazirani.

\begin{theorem}[\cite{valiantvazirani}]
\label{thm:valvaz}
 Let $n \in \N$ and let $S \subseteq \{0, 1\}^n$ be a nonempty set. Suppose $w_1, w_2, \ldots, w_n$ are randomly chosen from $\{0, 1\}^n$. Let $S_0 = S$ and let

\[ S_i = \{v \in S : \langle v,  w_1\rangle = \langle v , w_2\rangle  = \ldots = \langle v , w_i\rangle = 0 \}, \text{ for each } i \in [n]\]
(where the dot product of two vectors $v, w$ of length $n$ is $\langle v, w\rangle = \sum_{j=1}^n v_j w_j \pmod 2$). Let $P_n(S)$ be the probability that $|S_i| = 1$ for some $i \in \{0, \ldots, n\}$. Then $P_n(S) \geq \frac{3}{4}$.
\end{theorem}

We start by constructing a depth $5$ circuit and then reducing it to depth 2.  

\begin{theorem}\label{thm:thrifty-raz-smo}
    Let \( q \) be a constant prime number. Any $\G(k)$ gate on $n$ bits can be computed by a uniform family of probabilistic circuits of size \( n^{O(k)}\log(1/\eps) \), with \( O(k^2\log^2 n\log(1/\eps)) \) random bits and error $\eps$.
    Furthermore, the circuit has the following structure from top to bottom.
    \begin{itemize}
        \item The first layer (the top output gate) is an $\AND$ of fan-in $O(k\log n\log(1/\eps))$.
        \item The second layer consists of $\Mod_p$ gates with fan-in $n^{O(k)}$.
        \item The third layer consists of $\AND$ gates of fan-in $k$.
        \item The fourth layer consists of $\Mod_p$ gates of size $n^{O(k)}$.
        \item The fifth layer consists of $\AND$ gates of fan-in $O(k\log n)$.
    \end{itemize}
    Furthermore, this circuit can be constructed in $n^{O(k)}$ time.
\end{theorem}

\begin{proof}
    Let $G$ be an arbitrary $\G(k)$ gate. Assume that $G(x) = 0$ for $|x|>k$. Otherwise, we can construct a circuit $C$ computing $\neg G$, and then negate it by using a $\Mod_q$ gate connected to $C$ and $q-1$ 1's. 
    We begin by describing our circuit construction (with some commentary to help digest the circuit's behavior). A rigorous analysis of the construction will then follow.
    
    \textbf{Construction.} 
    It will be helpful to think of the circuit as the $\AND$ of two subcircuits, $C_1$ and $C_2$. On inputs $x$ with $|x|\le k$, $C_1$ will compute $G$ exactly while $C_2$ will output $1$. On the remaining inputs with $|x|>k$, $C_1$ will have arbitrary behavior while $C_2$ will output $0$ with high probability over the probabilistic bits. 
    
    Our circuit $C_1$ is the low-degree polynomial constructed in \Cref{lem:interpolate}. 
    This degree-$k$ polynomial can be constructed in $n^{O(k)}$ time, represented as a depth-2 circuit with fan-in-$k$ $\AND$ gates at the bottom, one $\Mod_q$ gate at the top, and requires no random bits. 
    
    Next, we describe the circuit $C_2$ layer-by-layer, from the inputs to the output gate.
    Define $m\coloneqq \lfloor \log \binom{n}{k+1} \rfloor + 1$. 
    In the first layer, we will have $n + m^2$ bits as input: the input $x$ along with $m^2$ random bits. 
    Identify the random bits as $m$ vectors $w_1\dots w_m\in \bitz^m$.
    Arbitrarily associate each $S\in \binom{[n]}{k+1}$ with a distinct bit string in $\bitz^m$, and denote the length-$m$ bit string associated with $S$ by $(S_1, S_2, \dots, S_m)$. 
    We can then define $\langle S, w_i\rangle \coloneqq \sum S_i w_i \pmod 2$.

    In the second layer, we will compute $\langle S, w_i\rangle \coloneq \sum_{j=1}^m x_iw_{i,j} \pmod 2$ for each $S\in \binom{[n]}{k+1}$ and $i\in [m]$. 
    Each $\langle S, w_i \rangle$ can be computed with a single $\oplus$ gate with fan-in $\le m$ by adding a wire from $w_{ij}$ to the gate iff $S_j = 1$.
    To turn $\oplus$ into $\Mod_q$ and $\AND$ gates, each $\oplus$ gate can be expanded into a DNF of size $\binom{n}k ^{O(1)} = n^{O(k)}$. Because at most one of the bottom-layer $\AND$ clauses can be satisfied simultaneously, we can replace the top $\OR$ gate with a $\Mod_q$ gate. 
    This conversion is done for each $\oplus$ gate, so, in total, we have $\binom{n}k\cdot m$ depth-2 subcircuits of size $n^{O(k)}$, where each subcircuit has a layer of fan-in-$m$ $\AND$ gates in the bottom layer and a single $\Mod_q$ gate at the top. 
    Denote the $\Mod_q$ gate computing $\langle S, w_i\rangle$ by $A_{S,i}$.

    In the third layer, for all $S\in \binom{[n]}{k+1}$ and $0\le \ell \le m$ we will compute the predicates 
    \[B_{S,\ell}\coloneqq \mathbbm{1}\left\{(x^S = 1) \wedge (\forall i\le \ell, \langle S, w_i\rangle = 0)\right\}. \] 
    These predicates are easily computed using the $A_{S,i}$'s. In particular, to compute $B_{S,k}$, take the $\AND$ of $x_i$ for $i\in S$, as well as the $A_{S,i}$ for all $i\le k$. 
    This uses a single $\AND$ gate of fan-in $O(m)$. 
    Notice that if $|x|\le k$, $B_{S,\ell}$ is false for all $S,\ell$.

    In the fourth layer, for $0\le i\le m$, we will compute the predicates 
    \[
    D_\ell \coloneqq \mathbbm{1}\left\{ \Abs{\{S \in \binom{[n]}{k} : x^S = 1 \text{ and } \forall i \leq \ell, \langle S, w_i \rangle = 0 \}} \not\equiv 1 \pmod q \right\}.
    \] 
    In words, $D_\ell$ is $1$ iff the number of sets $S\in\binom{[n]}k$ such that $x^S=1$ and $\forall i\le \ell, \langle S, w_i\rangle = 0$ is \emph{not} one more than a multiple of $q$.  
    This is accomplished by taking the $\Mod_q$ of $B_{S,\ell}$ for all $S$, along with $q-1$ $1$'s. Notice if $|x| > k$, then the set of all $S\in \binom{[n]}{k+1}$ such that $x^S=1$ is nonempty. Hence by \Cref{thm:valvaz}, with probability $\ge 1/4$ there will exist some $\ell$ such that there is exactly one $S$ with $x^S=1$ and $\forall i\le \ell, \langle S, w_i\rangle = 0$. In this case, we will have that $D_\ell = 0$.

    In the fifth layer, we simply take the $\AND$ of all the $D_\ell$, which will have fan-in $m$. By the analysis above, we know this $\AND$ gate will output $0$ with probability $\ge 1/4$ when $|x| > k$.

    We also note that by algorithmically constructing $C_2$ exactly in the manner we described, we can produce $C_2$ in $n^{O(k)}$ time.

    \textbf{Analysis.} Consider an input $x$ to $C$. 
    
    If $|x|\le k$, then we know by our construction of $C_1$ and \Cref{lem:interpolate} that $C_1(x) = G(x)$, and from our construction of $C_2$ that $B_{S,\ell}$ is $0$ for all $S$ and $\ell$. 
    It is clear that if all $B_{S,\ell}$'s $0$, then all the $D_\ell$'s must be $1$. Therefore, $C_2(x) = 1$.  
    Hence, in this case we have, \[C(x) = C_1(x)\wedge C_2(x) = G(x)\wedge 1 = G(x).\]
    
    Now if $|x| > k$, $C_1(x)$ may be arbitrary, but, as argued above, $C_2(x) = 0$ with probability $\ge 1/4$. 
    We can amplify the error probability of $C$ by replacing $C_2$ with $C_2'$, which is an $\AND$ of $O(\log(1/\eps))$ copies of $C_2$. 
    It is easy to see that the behavior of $C$ is preserved when $|x|\le k$. Now when $|x| > k$, \[\Pr[C(x) = 0] = \Pr[C_1(x)\wedge C_2'(x) = 0]\ge \Pr[C_2'(x)=0]\ge 1-(3/4)^{O(\log (1/\eps))}\ge 1-\eps.\] 
    
    We have shown that our circuit $C$ has the desired behavior: computing $G$ with error $\eps$. $C_1$ has size and construction runtime $n^{O(k)}$ and uses no random bits, and $C_2$ has size and construction runtime $n^{O(k)}$ and uses $m^2$ random bits. 
    Hence $C$ will have size and construction runtime $O(n^{O(k)}\log(1/\eps))$ and use $O(k^2\log^2 n\log(1/\eps))$ random bits.

    One can also verify easily that the construction has the desired structure (upon collapsing the cluster of $\AND$ gates at the top of the circuit, and trivially extending circuit $C_1$ past layer 2 using fan-in one gates).
\end{proof}

To shorten this construction to depth-2, we use the following depth-reduction lemma of Allender and Hertrampf \cite{AHdepthreduction94}.

\begin{lemma}[{\cite[Lemma 3]{AHdepthreduction94}}]
\label{lem:ahdepthred}
    Let $q$ be prime. Then every depth-$4$ circuit consisting of
\begin{itemize}
\item one $\Mod_p$ gate with fan-in \( s_1 \) on the top level,
    \item $\AND$ gates with fan-in \( t \) on the second level,
    \item $\Mod_p$ gates with fan-in \( s_2 \) on the third level, and
    \item $\AND$ gates with fan-in \( r \) on the last level
\end{itemize}
can be converted into a depth-$2$ circuit that is a $\Mod_p$ of $s_1 \cdot s_2^{t \cdot (p-1)}$  $\AND$ gates, each with fan-in \( r \cdot t \cdot (p-1) \). Furthermore, this conversion can be done in $O(s_1s_2^{t(p-1)} + rt)$ time.
\end{lemma} 

By applying this lemma twice to our depth-5 probabilistic circuit, we get the following depth-2 probabilistic circuit approximating a $\G(k)$ gate.

\begin{theorem}\label{thm:gk-depth-2}
    Let $q$ be a constant prime. Any $\G(k)$ gate on $n$ bits can be computed by a depth-2 probabilistic circuit using $O(k^2\log^2 n\log(1/\eps))$ random bits, and consists of a $\Mod_q$ of fan-in $2^{O(k^3\log^2n\log(1/\eps))}$ at the top, and $\AND$ gates of fan-in $O(k^3\log^2n\log(1/\eps))$ at the bottom layer. Furthermore, the circuit can be constructed in $2^{O(k^3\log^2 n\log(1/\eps))}$ time.
\end{theorem}

\begin{proof}
    We take the construction of \Cref{thm:thrifty-raz-smo} and apply \Cref{lem:ahdepthred} to all the depth-4 subcircuits. This yields a circuit with an $\AND$ of fan-in $O(k\log n\log(1/\eps))$ at the top, followed by $\Mod_p$ gates of fan-in $n^{O(k)}\cdot  n^{O(k\cdot k(q-1))} = 2^{O(k^2\log n)}$ in the next layer, followed by a final layer of $\AND$ gates of fan-in $O(k^2\log n)$.

    We now apply \Cref{lem:ahdepthred} again on this resulting circuit, where we add a dummy fan-in 1 $\AND$ gate at the top. This gives a depth-2 circuit whose top gate is a $\Mod_q$ of fan-in $2^{O(k^3\log^2n\log(1/\eps))}$, and whose bottom layer are $\AND$ gates of fan-in $O(k^3\log^2n\log(1/\eps))$ as desired.
\end{proof}

\section{Applications to Classical Complexity}\label{sec:apps-to-classical}

\cref{thm:probpoly,thm:gk-depth-2} generalize the seminal works of Razborov \cite{razborov1987lower}, Smolensky \cite{smolensky1987algebraic}, and Allender-Hertrampf \cite{AHdepthreduction94}, which have found use throughout theoretical computer science for nearly four decades. 
We expect most (if not all) of these applications to hold equally well for $\GC^0(k)[p]$ and $\GCC^0$, given our results in the previous section. 
To illustrate this, we have selected three applications to present here. 

In \cref{subsec:avg-case}, we prove average-case lower bounds against $\GC^0(k)[p]$. 
In particular, we prove that exponential-size circuits are necessary for a $\GC^0(k)[p]$ circuit to compute $\MAJ$ or $\Mod_q$ for any prime $q \neq p$. This was the original application of the theorems of Razborov and Smolensky.   

In \cref{subsec:ryan-williams}, we prove that $\e^\NP$ does not have non-uniform $\GCC^0$ circuits of exponential size. This generalizes the celebrated result of Williams \cite{wil14acc0}.

Finally, in \cref{subsec:learning-gc}, we apply a framework of Carmosino, Impagliazzo, Kabanets, and Kolokolova \cite{carmosino2016learning} to give a quasipolynomial time learning algorithm for $\GC^0(k)[p]$ in the PAC model over the uniform distribution with membership queries.

\subsection{Average-Case Lower Bounds for \texorpdfstring{$\GC^0[q]$}{GC0[q]}}\label{subsec:avg-case}

We prove that exponential-size $\GC^0(k)[q]$ circuits are necessary to compute $\MAJ$ and $\Mod_r$ for any prime $r \neq q$.
Our lower bounds generalize the lower bounds of Razborov \cite{razborov1987lower} and Smolensky \cite{smolensky1987algebraic} and follow the same structure. 
The lower bound argument has two main pieces: (1) $\GC^0(k)[q]$ circuits can be approximated by low-degree polynomials and (2) $\MAJ$ and $\Mod_r$ gates require large degree to be approximated by a polynomial. 
The former result was shown in \cref{thm:probpoly}, and the latter is a result of Razborov and Smolensky.

\begin{proposition}[\cite{razborov1987lower,smolensky1987algebraic}]\label{prop:low-deg-mod-maj}
Let $q$ and $r$ be distinct prime numbers, and let $F \in \{\MAJ, \Mod_r\}$. 
For all degree-$t$ polynomials $p(x) \in \F_q[x_1,\dots,x_n]$,  
\[
\Pr_{x \in \bitz^n}\left[p(x) = F(x) \right] \leq \frac{1}{2} + O\left(\frac{t}{\sqrt{n}}\right).
\]
\end{proposition}

We can prove correlation bounds against $\GC^0(k)[q]$ by combining \cref{thm:probpoly} and \cref{prop:low-deg-mod-maj}.

\begin{theorem}[{Correlation bounds against $\GC^0(k)[q]$}]\label{thm:cor-bound-gc0p}
Let $F\in \{\MAJ, \Mod_r\}$. 
For any depth-$d$ size-$s$ $\GC^0(k)[q]$ circuit $C$, we have 
\[\Pr_{x\in\bitz^n} [C(x) = F(x)]\le \frac{1}{2} + O\left(\frac{(k+\log n)(k +\log(ns))^{d-1}}{\sqrt{n}}\right) + \frac{1}{n}.\]
\end{theorem}
\begin{proof}
   By \cref{thm:probpoly}, there exists a polynomial $p(x) \in \F_q[x_1,\dots,x_n]$ with degree \\ $O\left((k+\log(1/\eps))(k +\log(s/\eps)\right)^{d-1})$ such that 
   \[
   \Pr_{x \in \bitz^n}\left[p(x) =  \neg C(x) \right] \ge 1-\eps.
   \]
   Then
   \begin{align*}
    \Pr_{x\in\bitz^n} [ C(x) = F(x)] &= \Pr_{x\in\bitz^n} [\neg C(x) \neq F(x)]  \\
    &\le \Pr_{x\in\bitz^n}[p(x) \neq F(x)] + \Pr_{x\in\bitz^n}[p(x)\neq \neg C(x)] \\ 
    &\leq \Pr_{x\in\bitz^n}[1 - p(x) = F(x)] + \eps \\
    &\le \frac{1}{2} + O\left(\frac{(k+\log(1/\eps))(k +\log(s/\eps))^{d-1}}{\sqrt{n}} \right) + \eps, 
   \end{align*}
   where the second inequality follows from the fact that $\Pr_{x \in \bitz^n}\left[p(x) =  \neg C(x) \right] \ge 1-\eps$ and the third inequality follows from \cref{prop:low-deg-mod-maj}.
   The result follows from setting $\eps = 1/n$.
\end{proof}

As a corollary, we get a lower bound for $\GC^0(k)[q]$. 

\begin{corollary}\label{thm:gc0-lower-bound}
    Let $q$ and $r$ be distinct prime numbers, let $F \in \{\MAJ, \Mod_r\}$, and let $k = \Theta(n^{1/2d})$. 
    Any depth-$d$ $\GC^0(k)[q]$ circuit that computes $F$ must have size
    $2^{\Omega\left(n^{1/2(d-1)}\right)}$.
\end{corollary}
\begin{proof}
   Let $C$ be a size-$s$, depth-$d$ $\GC^0(k)[q]$ circuit $C$ that can compute $F \in \{\MAJ, \Mod_r\}$. We have \[1 = \Pr[C(x) = F(x)] \le  \frac{1}2 + O\left(\frac{(k+\log n)(k +\log(ns))^{d-1}}{\sqrt{n}}\right) + \frac{1}{n}.\] 
   By solving for $s$, we can conclude that 
   \[
   s \geq 2^{\Omega\left(n^{1/2(d-1)} - k^{d/(d-1)} \right)}. 
   \]
   Plugging in $k$ gives the desired result.
\end{proof}

We can improve our average-case lower bounds for $\GC^0(k)[q]$ to average-case lower bounds for $\GC^0(k)[q]\mathsf{/rpoly}$. 
Recall that $\mathsf{/rpoly}$ means the circuit gets random advice as additional input. 
In other words, one gets to choose a probability distribution over polynomially many bits that depends on the input size (but not the specific input), and the circuit gets to draw one sample from this distribution. 

\begin{theorem}[Average-case lower bound for {$\GC^0(k)[q]$}]\label{thm:mod-p-correlation-bounds}
Let $q$ and $r$ be distinct prime numbers, and let $F \in \{\Mod_r, \MAJ\}$.
    There exists an input distribution on which any $\GC^0(k)[q]\mathsf{/rpoly}$ circuit of depth $d$, $k = O(n^{1/2d})$, and size at most $\exp\left(O\left(n^{1/2.01d} \right)\right)$ only computes $F$ with probability $\frac{1}{2} + \frac{1}{n^{\Omega(1)}}$.
\end{theorem}
\begin{proof}
Toward a contradiction, assume that for all input distributions, there exists a $\GC^0_d(k)[q]/\mathsf{rpoly}$ circuit with $k = O(n^{1/2d})$ and size $2^{\Omega(n^{1/2.01d})}$ that computes $F$ with probability $1/2 + \eps$ for $\eps = 1/n^{o(1)}$.
Then Yao's minimax principle implies that there exists a distribution over $\GC^0_d(k)[q]$ circuits that computes $F$ with probability $1/2 + \eps$ on every input.
By drawing $O(1/\eps^2)$ samples from this distribution and taking the majority vote of their outputs, we obtain a new circuit that computes $F$ with probability $0.99$ on every input. 
Recall that one can compute majority on $m$ bits with a size-$2^{O(n^{1/d})}$ $\AC^0$ circuit \cite{haastad2014correlation}. 
Therefore, since $O(1/\eps^2) = n^{o(1)}$, the majority of the $\GC^0_d(k)[q]$ circuits can be computed in depth $d$ and size $2^{n^{o(1)}}$, which doubles the depth of the original circuit and only increases the size by a negligible amount. 

Next, we amplify the success probability from $0.99$ to $1 - \exp(-n)$, for some $\exp(-n) < 2^{-n}$, by sampling $O(n)$ circuits that succeed with probability $0.99$ and taking their majority vote. 
Since the circuits succeed with probability $0.99$, it is easy to see that a $0.99$-fraction of the votes will be $0$'s or $1$'s with high probability. 
Hence, the approximate majority construction of Ajtai and Ben-or \cite{ajtai1984theorem} suffices, which can be performed by a polynomial-size $\AC^0$ circuit.\footnote{For the unfamiliar reader, the approximate majority circuit will output $``1"$ when at least a $0.75$-fraction of the inputs are $1$, $``0"$ when at most a $0.25$-fraction of the inputs are $0$, and behave arbitrarily otherwise.}

Because this distribution over $\GC^0_d(k)[q]$ circuits fails to compute $F$ with probability less than $2^{-n}$, it follows by union bounding over all $2^n$ inputs that there exists one circuit in the distribution that computes $F$ on all inputs. 
Hence, we have constructed a $\GC^0_d(k)[q]$ circuit of depth $2d + O(1)$, $k=O(n^{1/2d})$, and size $\exp(n^{1/2.01d})$, contradicting \cref{thm:gc0-lower-bound}.
\end{proof}

\subsection{Non-Uniform \texorpdfstring{$\GCC^0$}{GCC0} Lower Bounds}\label{subsec:ryan-williams}

We prove that there are languages in $\e^\NP$ that fail to have polynomial-size $\GCC^0(k)$ circuits for certain values of $k$ (which are stated carefully in \cref{thm:main-gcc0}).
Recall that $\e$ is the class of languages that can be decided by a Turing machine in time $2^{O(n)}$.
This generalizes the breakthrough work of Williams \cite{wil14acc0} who proved that there are languages in $\NEXP$ and $\e^\NP$ that fail to have polynomial-size $\ACC^0$ circuits.
Here we focus on $\e^\NP$ instead of $\NEXP$ because we get a stronger size-depth tradeoff. 
We note that similar arguments can show that $\NEXP$ fails to have $\GCC^0(k)$ circuits.

These lower bounds are based on Williams' algorithmic method, which, in short, connects the existence of fast algorithms for the \textsc{CircuitSAT} problem to circuit lower bounds. 

\begin{definition}[$\calC$-\textsc{CircuitSAT}]
Given as input a description of a $\calC$ circuit $C$, the $\calC$-\textsc{CircuitSAT} problem is to decide whether there exists an input $x \in \{0,1\}^n$ such that $C(x) = 1$.
\end{definition}

The algorithmic method only works for ``nice'' circuit classes. 

\begin{definition}[Nice circuits \cite{wil14acc0}]
    A \emph{nice} circuit class $\calC$ is a collection of circuit families that:
    \begin{itemize}
        \item contain $\AC^0$: for every circuit family in $\AC^0$, there is an equivalent circuit family in $\calC$, and
        \item is closed under composition: 
        for $\{C_n\}$, $\{D_n\} \in \calC$ and any integer $c$, the circuit family obtained by feeding $n$ input bits to $n^c + c$ copies of $C_n$ and feeding the outputs into $D_{n^c + c}$ is also in $\calC$. 
    \end{itemize}
\end{definition}

Every well-studied circuit class is nice, and it is easy to see that $\GCC^0$ is nice too. 

We can now formally state the essence of the algorithmic method. Specifically, fast algorithms for $\calC$-\textsc{CircuitSAT} imply circuit lower bounds for $\calC$.

\begin{theorem}[{\cite[Theorem 3.2]{wil14acc0}}]
\label{thm:algtolowerbound}
    Let $S(n) \leq 2^{n/4}$ and let $\calC$ be a nice circuit class. 
    There is a $c > 0$ such that, if $\mathcal{C}$-\textsc{CircuitSAT} instances with at most $n + c \log n$ variables, depth $2d + O(1)$, and $O(n S(2n) + S(3n))$ size can be solved in $O(2^n / n^c)$ time, then $\e^{\NP}$ does not have non-uniform $\mathcal{C}$ circuits of depth $d$ and $S(n)$ size.
\end{theorem}

To apply \cref{thm:algtolowerbound} and obtain our $\GCC^0$ lower bound, we will give fast algorithms for $\GCC^0$-\textsc{CircuitSAT}, showing that the algorithmic method of Williams also lifts from $\ACC^0$ to $\GCC^0$. 
As a starting point, we will recall the $\ACC^0$ satisfiability algorithm and then extend the necessary parts to $\GCC^0(k)$.
Let $\SYM^+$ be the class of depth-two circuits with a layer of $\AND$ gates at the bottom and some symmetric function at the top.
The $\ACC^0$-\textsc{CircuitSAT} algorithm can be modularized as follows. 
Given as input a description of a size-$s$ depth-$d$ $\ACC^0$ circuit (that is comprised of $\AND$, $\OR$, $\NOT$, and $\Mod_m$ gates for a fixed $m$), the algorithm performs the following four steps. 

\begin{enumerate}
    \item Turn each $\Mod_m$ gate into an $\AND$ of $\Mod_p$'s of $\AND$'s, where all gates have constant fan-in and $p$ is some prime dividing $m$. This takes $s^{O(1)}$ time.
    
\item Replace each $\OR$ gate with a probabilistic circuit consisting of a $\Mod_p$ of $2^{\poly(\log s)}$ $\AND$s, each of fan-in $\poly(\log s)$. Call the resulting circuit $C$. $C$ uses $\poly(\log s)$ random bits.

\item Convert $C$ into a $\SYM^+$ circuit $C'$ of size $2^{O(\poly(\log s))}$ whose top symmetric gate can be evaluated in time $2^{O(\poly(\log s))}$.

\item Run a $\SYM^+$\textsc{-CircuitSAT} algorithm on $C'$.
\end{enumerate}

To design a $\GCC^0(k)$-\textsc{CircuitSAT} algorithm, it suffices to modify only the second step in the above blueprint to handle $\g(k)$ gates.
In particular, we will use our \cref{thm:gk-depth-2} to turn a $\G(k)$ gate into a probabilistic circuit with only $\Mod_p$ gates and bounded fan-in $\AND$s with comparable parameters to Step 2 above. (In particular, our circuit will have the same size and $\AND$ fan-in, but with $k\log s$ in place of $\log s$.) 

Now we will prove that Step 2 above holds for $\g(k)$ gates.
We first recall the $\ACC^0$ theorems established in \cite{wil14acc0} that we will use in a black-box manner. In these theorems, we will fix a function $f(d) \coloneqq 2^{O(d)}$ that quantifies the size-depth tradeoffs in these theorems. 
This will be important to track the size-depth improvements we obtain in our $\GCC^0(k)$ lower bounds.

\begin{theorem}[\cite{allender1994uniform, wil14acc0}]
\label{thm:probcirctosym+}
Let $f:\N\to\N$ be a function where $f(d)= 2^{O(d)}$ and let $t \in \N$. 
Let $C$ be a probabilistic circuit with depth $2d=O(1)$, size $2^{t^4}$, no $\OR$ or $\Mod_m$ gates for any composite $m$, and $\AND$ gates of fan-in at most $t^4$ that computes a function with $t^3$ probabilistic inputs and error probability $1/3$. There is an algorithm that, given $C$, outputs an equivalent $\SYM^+$ circuit of size $2^{O(t^{f(d)})}$. The algorithm takes at most $2^{O(t^{f(d)})}$ time.

Furthermore, if the number of $\AND$s in the $\SYM^+$ circuit that evaluate to 1 is known, then the symmetric function in the $\SYM^+$ circuit can be evaluated in $2^{O(t^{f(d)})}$ time.
\end{theorem}

Williams transforms a size-$s$, depth-$d$ $\ACC^0$ circuit into a $\sym^+$ circuit by replacing each $\OR/\AND$ gate with a depth-2 probabilistic circuit with $\AND$ gates of bounded fan-in and then applying \Cref{thm:probcirctosym+} with $t \gets O(\log s)$. 
This is formalized in the following lemma.

\begin{lemma}[\cite{AHdepthreduction94,allender1994uniform,wil14acc0}]
\label{thm:accztosym+}
Let $f: \N \to \N$ be a function where $f(d) = 2^{O(d)}$.
There is an algorithm that, given an $\acc^0$ circuit of depth $d = O(1)$ and size $s$, 
outputs an equivalent $\SYM^{+}$ circuit of size $2^{O(\log^{f(d)} s)}$.
The algorithm takes $2^{O(\log^{f(d)} s)}$ time.

Furthermore, if the number of $\AND$s in the $\SYM^+$ circuit that evaluate to 1 is known, then the symmetric function in the $\SYM^+$ circuit can be evaluated in $2^{O(\log^{f(d)} s)}$ time.
\end{lemma}

We will get a similar conversion for size-$s$ depth-$d$ $\GCC^0$ circuits by replacing $\G(k)$ gates with our newly constructed depth-2 probabilistic circuits from \Cref{thm:gk-depth-2}, which are comparable in size and identical in depth to the $\AND/\OR$ probabilistic circuit construction used to prove \cref{thm:accztosym+}. 
This allows us to use \Cref{thm:probcirctosym+} with $t \gets O(k\log s)$.

\begin{theorem}
\label{thm:gccztosym+}
Let $f: \N \to \N$ be a function where $f(d) = 2^{O(d)}$.
There is an algorithm that, given a $\GCC^0(k)$ circuit of depth $d = O(1)$ and size $s$, outputs an equivalent $\sym^{+}$ circuit of size $2^{O((k\log s)^{f(d)})}$. 
The algorithm takes at most $2^{O((k\log s)^{f(d)})}$ time.

Furthermore, if the number of $\AND$s in the $\SYM^+$ circuit that evaluate to 1 is known, then the symmetric function in the $\SYM^+$ circuit can be evaluated in $2^{O((k\log s)^{f(d)})}$ time.
\end{theorem}

\begin{proof}
    Let $C$ be the given circuit. As in the $\ACCz$ case, we will identically use Step 1 to convert all $\Mod_m$ gates into $\Mod_p$ gates, with $p$ prime, in $s^{O(1)}$ time (see \cite[Appendix A]{wil14acc0} for specific details). Denote this new circuit $C'$. At this point we will now use \Cref{thm:gk-depth-2} to replace each $\G(k)$ gate with a probabilistic circuit that computes the gate except with probability $\eps\coloneqq 1/3s$ and uses the \emph{same random bits} (versus having a fresh supply per gate), which can be done in time $s\cdot 2^{O(k^3\log^3s)}$. Since the fan-in of each $\G(k)$ gate is at most $s$ and $\eps = 1/3s$, it follows that each $\G(k)$ gate is replaced by a depth-2 probabilistic circuit of size $2^{O(k^3\log^3 s)}$ consisting of $\Mod_p$ gates with $p$ prime, and $\AND$ gates of fan-in $O((k\log s)^3)$. Furthermore, the circuit uses $O(k^2\log^3 s)$ random bits altogether. Notice by a union bound, there is at most $s(1/3s) = 1/3$ probability that one of the $s$ probabilistic subcircuits substituted in is faulty. Therefore, the resulting circuit computes $C$ with probability $\ge 2/3$. We finally apply \Cref{thm:probcirctosym+} to construct the desired $\SYM^+$ circuit in the desired time complexity.
\end{proof}

The algorithm in \cref{thm:gccztosym+} is the transformation in Step 2 above.
Hence, all that remains to get our lower bound is to put the pieces together.
To do so, we need the following evaluation algorithm, which takes a $\SYM^+$ circuit as input and outputs its truth table. 

\begin{lemma}[\cite{wil14acc0}]
\label{lem:sym+eval}
    There is an algorithm that, given a $\sym^+$ circuit of size $s \leq 2^{0.1n}$ and $n$ inputs with a symmetric function that can be evaluated in $\poly(s)$ time, runs in $(2^n + \poly(s))\poly(n)$ time and prints a $2^n$-bit vector $V$ which is the truth table of the function represented by the given circuit. That is, $V[i] = 1$ iff the $\sym^+$ circuit outputs 1 on the $i$th variable assignment.
\end{lemma}

This gives us our fast $\GCC^0(k)$-\textsc{CircuitSAT} algorithm. 
Recall that $f: \N \to \N$ in the theorems below is a function $f(d) = 2^{O(d)}$.

\begin{theorem}
\label{thm:gccfasteval}
    For every $d>1$ and  $\eps = \eps(d) \coloneqq .99/f(d)$, the satisfiability of depth-$d$ $\GCC^0(k)$ circuits with $n$ inputs and $2^{n^\eps/k}$ size can be determined in time $2^{n-\Omega(n^\delta/k)}$ for some $\delta > \eps$.
\end{theorem}
\begin{proof}
    Consider $C$, a depth-$d$ $\GCC^0$ circuit of size $2^{n^\eps/k}$. For any $\ell\in [n]$, we can create circuit $C'$ of depth $d+1$, size $s2^{\ell}$ over $n-\ell$ inputs by taking $2^{\ell}$ copies of $C$, plugging in a distinct assignment of the first $\ell$ bits into each copy, and then taking the $\OR$ of them. Notice that $C$ is satisfiable iff $C'$ is.  

    We now apply \Cref{thm:gccztosym+} on $C'$ to get an equivalent $\SYM^+$ circuit $C''$, which is a symmetric function of $s''\le 2^{(k(\ell + \log s))^{f(d)}}$ $\AND$s. By \Cref{lem:sym+eval} and the fact the symmetric function can be computed in $\poly(s'')$ time, it follows that upon setting $\ell \coloneqq \log s = n^{\eps}/k$, we get an algorithm that runs in $O(2^{n-\ell}\poly(n)) = 2^{n-\Omega(n^\delta/k)}$ for some $\delta > \eps$.
\end{proof}

Our circuit satisfiability algorithm implies the following lower bound.

\begin{theorem}[$\e^\NP \not\subseteq \GCC^0$]\label{thm:main-gcc0}
    For every $d$, there is a constant $C>1$ and $\delta = \delta(d) \coloneqq 1/Cf(2d)$, such that for all $k \le O(n^\delta/\log n)$, there exists a language in $\e^{\NP}$ that fails to have $\GCC^0(k)$ circuits of depth $d$ and size $\exp\left(\Omega(n^\delta/k)\right)$.
 \end{theorem}

\begin{proof}
    By \Cref{thm:gccfasteval}, we know for every $d$, the satisfiability of depth-$d$ $\GCC^0(k)$ of size $2^{O(n^{.99/f(d)})}$ on $n$ inputs can be solved in $2^{n-\Omega(n^\eps/k)}$ time for some $\eps > 1/4f(d)$. Now by \Cref{thm:algtolowerbound}, we know there exists a constant $c>0$ such that if $\GCC^0(k)$-\textsc{Circuit SAT} instances with $n+c\log n$ variables, depth $2d+O(1)$, and size $s = n2^{(2n)^\delta} + 2^{(3n)^\delta}$ can be solved in time $O(2^n/n^c)$, then $E^{NP}$ doesn't have non-uniform $\GCC^0(k)$ circuits of depth $d$ and size $2^{n^\delta}$. Since $f(d) = 2^{O(n)}$, we know $f(2d+O(1)) \le Cf(2d)$ for some constant $C$. Consequently, for $\delta = 1/Cf(2d)$, we can indeed solve depth $2d+O(1)$ and size $n2^{(2n)^\delta} + 2^{(3n)^\delta}\le \exp\left(O(n^{\frac{.99}{f(2d+O(1))}})\right)$ $\GCC^0$ circuits over $n+c\log n$ inputs in time $2^{(n+c\log n) - \Omega((n+c\log n)^\eps/k)} = O(2^n/n^c)$ for small enough constant $c$ (by using the assumption $n^\delta/k=\Omega(\log n)$), yielding the desired lower bound.
\end{proof}

We conclude with some remarks about the extent of our contribution. 
The Williams lower bound of $\e^\NP \not\subseteq \ACC^0$ suffices to prove that there exist languages in $\e^\NP$ that fail to have polynomial-size $\GCC^0$ circuits (or even exponential-size $\GCC^0$ circuits for some small enough exponential function). 
This is achieved by na\"lively transforming the $\GCC^0$ circuit to an $\ACC^0$ circuit. 
Specifically, suppose we have a size-$s$ depth-$d$ $\GCC^0(k)$ circuit, and then we transform each $\g(k)$ gate into a CNF (or DNF, it does not matter). 
The resulting circuit will be a size-$s^k$ depth-$2d$ $\ACC^0$ circuit.
Then, after applying the lower bound for depth-$d$ size-$\exp(\Omega(n^{1/f(2d)}))$ $\ACC^0$ circuits\footnote{This is the lower bound proved by Williams \cite{wil14acc0}. It is also a special case of \cref{thm:main-gcc0} with $k=1$.}, we obtain a separation between $\e^{\NP}$ and depth-$d$ $\GCC^0(k)$ circuits of size $\exp(O(n^{1/Cf(4d)}/k))$. 

In our \cref{thm:main-gcc0}, we get a separation between $\e^\NP$ and depth-$d$ $\GCC^0(k)$ circuits of size $\exp(O(n^{1/Cf(2d)}/k))$. 
The difference is the $f(2d)$ in \cref{thm:main-gcc0} vs. $f(4d)$ in the na\"ive approach that appear in the exponent of the exponent of the circuit size. 
Because $f$ is an exponential function as well, the difference is then a factor of $2$ in the exponent of the exponent of the exponent. Hence, using our result yields an improvement in the triple exponent in the size-depth tradeoff compared to the na\"ive approach.

\subsection{PAC Learning \texorpdfstring{$\GC^0[p]$}{GC0[p]}}\label{subsec:learning-gc}

Carmosino, Impagliazzo, Kabanets, and Kolokolova \cite{carmosino2016learning} gave a quasipolynomial time learning algorithm for $\AC^0[p]$ in the PAC model over the uniform distribution with membership queries. 
We recall their result in more detail and argue that there is a quasipolynomial time learning algorithm for $\GC^0(k)[p]$. 

To begin, we establish some notation and define the learning model. 
For a circuit class $\Lambda$ and a set of size functions $\calS$, $\Lambda[\calS]$ denotes the set of size-$\calS$ $n$-input circuits of type $\Lambda$.  
For a Boolean function $f: \bitz^n \to \bitz$ and $\eps \in [0,1]$, 
$\widetilde{\mathsf{CKT}}_n(f, \eps)$ denotes the set of
all circuits that compute $f$ on all but an $\eps$ fraction of inputs. 

\begin{definition}[Learning model]\label{def:learning-model}
    Let $\calC$ be a class of Boolean functions.
    An algorithm $A$ PAC-learns $\calC$ if for any $n$-variate $f \in \calC$ and for any $\eps, \delta > 0$, given membership query access to $f$, algorithm $A$ prints with probability at least $1-\delta$ over its internal randomness a circuit $C \in \widetilde{\mathsf{CKT}}_n(f,\eps)$. The runtime of $A$ is measured as a function of $T(n, 1/\eps, 1/\delta, \,\size(f))$.
\end{definition}

Carmosino et al.\ establish a connection between learning and natural proofs \cite{razborov1994natural}. We recall the definition of natural proofs here for convenience. 
Let $F_n$ be the collection of all Boolean functions on $n$ variables. 
$\Lambda$ and $\Gamma$ denote complexity
classes. 
A \emph{combinatorial property} is a sequence of subsets of $F_n$ for each $n$.

\begin{definition}[Natural property \cite{razborov1994natural}]\label{def:natural-property}
A combinatorial property $R_n$ is $\Gamma$-natural against $\Lambda$ with density $\delta_n$ if it satisfies the following three conditions:

\begin{itemize}
    \item \textbf{Constructivity:} The predicate $f_n \stackrel{?}{\in} R_n$ is computable in $\Gamma$.
    \item \textbf{Largeness:} $\abs{R_n} \geq \delta_n \abs{F_n}$.
    \item \textbf{Usefulness:} For any sequence of functions $f_n$, if $f_n \in \Lambda$ then $f_n \notin R_n$, almost everywhere. 
\end{itemize}
\end{definition}

A proof that some explicit function is not in $\Lambda$ is called $\Gamma$-natural against $\Lambda$ with density $\delta_n$ when it involves a $\Gamma$-natural property $R_n$ that is useful against $\Lambda$ with density $\delta_n$.
Razborov and Rudich \cite{razborov1994natural} showed that the Razborov-Smolensky lower bound proofs are $\NC^2$-natural against $\AC^0[p]$, where, roughly speaking, the natural property contains functions that cannot be approximated by low-degree polynomials (see \cite[Section 3]{razborov1994natural} and \cite[Section 5]{carmosino2016learning} for further details). 
An immediate implication of our lower bounds (\cref{thm:gc0-lower-bound,thm:mod-p-correlation-bounds}) is that the same property is $\NC^2$-natural against $\GC^0(k)[p]$. 

\begin{theorem}\label{thm:gc0-is-natural}
   For every prime $p$, there is an $\NC^2$-natural property of $n$-variate Boolean functions, with largeness at least $1/2$, that is useful against $\GC^0(k)[p]$ circuits of depth $d$ and of size up to 
   $\exp\left(\Omega(n^{1/2d})\right)$ where $k = O(n^{1/2d})$.
\end{theorem}

Carmosino et al.\ \cite{carmosino2016learning} prove the following connection between natural properties and PAC learning algorithms over the uniform distribution with membership queries.

\begin{theorem}[{\cite[Theorem 5.1]{carmosino2016learning}}]\label{thm:natural-proof-to-learning}
Let $\Lambda$ be any circuit class containing $\AC^0[p]$ for some prime $p$. Let $R$ be a $\PTIME$-natural property, with largeness at least $1/5$, that is
useful against $\Lambda[u]$, for some size function $u: \N \to \N$. 
Then there is a randomized algorithm
that, given oracle access to any function $f : \bitz^n \to \bitz$ from $\Lambda[s_f]$, produces a circuit
$C \in \widetilde{\mathsf{CKT}}(f, \eps)$ in time $\poly(n, 1/\eps, 2^{u^{-1}\poly(n,1/\eps,s_f)})$.
\end{theorem}

By combining \cref{thm:gc0-is-natural}, \cref{thm:natural-proof-to-learning}, and the basic fact that $\AC^0[p] \subseteq \GC^0(k)[p]$ for all primes (and prime powers) $p$, we get the following learning algorithm for $\GC^0(k)[p]$.

\begin{corollary}[{Learning $\GC^0(k)[p]$ in quasipolynomial time}]\label{thm:pac-learn}
Let $k = O(n^{1/2d})$.
For every prime $p$, there is a randomized algorithm that, using membership queries, learns a given $n$-variate Boolean function $f \in \GC^0(k)[p]$ of size $s_f$ to within error $\eps$ over the uniform distribution, in time $\quasipoly(n, s_f, 1 / \eps)$. 
\end{corollary}

\section{Applications to Quantum Complexity}\label{sec:quantum}

We study the implications of our lower bounds for $\GC^0[p]$ and $\GC^0$ on quantum complexity theory.
Specifically, we show exponential separations between shallow quantum circuits and both $\GC^0[p]$ and $\GC^0$, surpassing all previously known separations between quantum and classical circuits. 
We emphasize that these separations are \emph{unconditional} and our results generalize the prior work in this area \cite{bravyi2018quantum, watts2019exponential, bravyi2020quantum, grier2021interactive, raz2022oracle, grilo2024power}.

For convenience, we summarize the separations we obtain in this section. We say a separation is exponential when polynomial-size quantum circuits can solve a certain problem but even some exponential-size classical circuits cannot.
In this section, we exhibit (formal definitions and arguments are given within the corresponding subsection): 
\begin{itemize}
    \item A promise problem separating $\BQLOGTIME$ and $\GC^0(k)$ (\cref{cor:bqlogtime-not-in-gac0}).
    \item A relation problem separating $\QNC^0$ and $\GC^0(k)$ (\cref{thm:watts-main}). 
    \item A relation problem separating $\QNC^0\mathsf{/qpoly}$ and $\GC^0(k)[p]$ for any prime $p$ (\cref{thm:gc2,thm:gcp}). 
    \item An interactive problem separating $\QNC^0$ and $\GC^0(k)[p]$ for any prime $p$ (\cref{thm:grier-schaeffer}). 
\end{itemize}
Our separations are all exponential (i.e., the problems can be solved by polynomial-size $\QNC^0$ circuits but are hard for exponential-size classical circuits), and \cref{thm:watts-main,thm:gc2,thm:gcp} prove average-case lower bounds.

In addition to our results in \cref{sec:raz-smo,sec:apps-to-classical}, our quantum-classical separations require a few new classical ingredients. We prove a \emph{multi-output} multi-switching lemma for $\GC^0$ (\cref{thm:multioutputlemma}), which generalizes the multi-switching lemma proved by Kumar \cite{kumar2023tight} to multi-output $\GC^0$ circuits. 
Our result is based on the multi-switching lemmas for $\AC^0$ that were proven by H\aa stad \cite{haastad2014correlation} and Rossman \cite{ros17entropyswitch}, and is based on the proof of the $\AC^0$ multi-output multi-switching lemma established in \cite{watts2019exponential}. 

We also prove that a single $\g(k)$ gate can compute functions that are not computable in $\NC = \AC = \TC$ when $k = \log^{\omega(1)}n$ (\cref{thm:gk-incomp-nc}). We use this to show that certain $\GC^0(k)[p]$ circuits are incomparable to $\NC^1$ (\cref{cor:gc0k-nc1-incomp}), which is needed in the proof of \cref{thm:grier-schaeffer}.

\subsection{Pushing Raz \& Tal: \texorpdfstring{$\BQLOGTIME \not\subseteq \GC^0$}{BQLOGTIME ⊄ GC0}}\label{subsec:raz-tal}

In a breakthrough work, Raz and Tal \cite{raz2022oracle} showed that $\BQP$ is not in $\PH$ relative to an oracle.
An unconditional separation between $\BQLOGTIME$ and $\AC^0$ is at the core of their result. Specifically, they give a distribution that is \emph{pseudorandom} (i.e., cannot be distinguished from the uniform distribution) for $\AC^0$ circuits, but not for $\BQLOGTIME$ circuits. By well-known reductions, this implies their oracle and circuit separations.  
We show that their distribution is also pseudorandom for $\GC^0$ circuits. Hence, by the same reductions, we can conclude that $\BQLOGTIME \not\subseteq \GC^0$.
We begin with a formal definition of $\BQLOGTIME$.

\begin{definition}\label{def:bqlogtime}
   $\BQLOGTIME$ is the class of promise problems $\Pi = (\Pi_{\textsc{Yes}}, \Pi_{\textsc{No}})$ that are decidable, with bounded error probability, by a $\mathsf{LOGTIME}$-uniform family of quantum circuits $\{C_n\}_{n\in\N}$, where each $C_n$ is an $n$-qubit quantum circuit with $O(\log n)$ gates that are either (i) input query gates (i.e., gates that map $\ket{i}\ket{z}$ to $\ket{i}\ket{z \oplus x_i}$ where $x = x_1 \dots x_n$ is the input string) or (ii) standard quantum gates from a fixed, finite gate set. 
\end{definition}

Let $\calD_{\textsc{Raz-Tal}}$ denote the distribution over $\{-1,1\}^{2N}$ described in \cite[Section 4]{raz2022oracle} (see also \cite[Section 2]{wu2022stoch}). 
Raz and Tal showed that if $\calD_{\textsc{Raz-Tal}}$ is sufficiently pseudorandom, then one can obtain separations from $\BQLOGTIME$.

\begin{lemma}[\cite{raz2022oracle}]\label{lemma:abstract-raz-tal}
Let $\calF$ be a class of Boolean functions 
$f: \{\pm 1\}^{2N} \to \{0,1\}$. Suppose that for each $f \in \calF$, 
\[
    \Abs{\E[f(\calD_{\textsc{Raz-Tal}})] - \E[f(U_{2N})]} \leq \left(\frac{1}{\log N}\right)^{\omega(1)}.
\]
Then $\BQLOGTIME \not\subseteq \calF$.
\end{lemma}

Furthermore, Raz and Tal showed that the desired pseudorandomness property follows from understanding the \emph{second-level Fourier growth}, i.e., the $\ell_1$-norm of the Fourier coefficients on the second level. 

\begin{lemma}[\cite{raz2022oracle}, {\cite[Theorem 4.4]{wu2022stoch}}]\label{lem:second-level}
   Let $f: \{\pm 1\}^{2N} \to \{0,1\}$ be a Boolean function. For $L > 0$, suppose that for any restriction $\rho$, 
   \[
   \sum_{\substack{S \subseteq [2N]\\\abs{S}=2}} \abs{\widehat{f_\rho}(S)} \leq L.
   \]
   Then, 
   \[
   \Abs{\E[f(\calD_{\textsc{Raz-Tal}})] - \E[f(U_n)] } \leq \frac{2\eps L}{\sqrt{N}}.
   \]
\end{lemma}

In prior work, Kumar \cite{kumar2023tight} gave upper bounds on the Fourier growth of $\GC^0$-computable functions. 

\begin{lemma}[{\cite[Theorem 5.14]{kumar2023tight}}]\label{lemma:kumar-switching}
    Let $f : \{\pm 1\}^{n} \to \{\pm 1\}$ be computable by a size-$m$ $\GC^0_d(k)$ circuit. 
    Then, for all $\ell \in \N$, the following is true for some universal constants $C_1, C_2 > 0$:
    \[
    \sum_{\substack{S\subseteq[n]\\\abs{S}=\ell}} \abs{\widehat{f}(x)} \leq C_1  (C_2 \cdot k(k + \log m)^{d-1})^\ell.
    \]
    In particular, for some universal constant $C> 0$,
     \[
    \sum_{\substack{S\subseteq[n]\\\abs{S}=2}} \abs{\widehat{f}(x)} \leq C   k^2(k + \log m)^{2(d-1)}.
    \]
\end{lemma}

We can now start combining these ingredients to obtain the claimed separation. 

\begin{proposition}[Generalization of {\cite[Theorem 7.4]{raz2022oracle}}]\label{prop:gc0-pseudorandom}
    Let $f : \{\pm 1\}^{2N} \to \{\pm 1\}$ be a size-$m$ $\GC^0_d(k)$ circuit. 
    Then there is a universal constant $C > 0$ such that
    \[
    \Abs{\E[f(\calD_{\textsc{Raz-Tal}})] - \E[f(U_n)]} \leq \frac{C \eps  k^2 (k + \log m)^{2(d-1)}}{\sqrt{N}}.
    \]
\end{proposition}
\begin{proof}
    Combine \cref{lem:second-level,lemma:kumar-switching}.
\end{proof}

Combining \cref{lemma:abstract-raz-tal,prop:gc0-pseudorandom} yields the following two corollaries.

\begin{corollary}[Generalization of {\cite[Corollary 7.5]{raz2022oracle}}]\label{cor:d-fools-gac0}
    Let $f : \{\pm 1\}^{2N} \to \{\pm 1\}$ be a $\GC^0(k)$ circuit of constant depth and size $\quasipoly(N)$. 
    For $\eps = O\left(\frac{1}{ \log N}\right)$ and $k = \frac{O(N^{1/4d})}{\log^{\omega(1)} N}$, 
    \[
    \Abs{\E[f(\calD_{\textsc{Raz-Tal}})] - \E[f(U_{2N})]} \leq \frac{1}{\log^{\omega(1)} N}.\footnote{Note that $\eps \in \Omega(1/\log N)$ is necessary for the $\BQLOGTIME$ to succeed with a large enough probability. See \cite[Section 6]{raz2022oracle} for further detail.} 
    \]
\end{corollary}

\begin{corollary}[Generalization of {\cite[Corollary 1.6]{raz2022oracle}}]\label{cor:bqlogtime-not-in-gac0}
There is a promise problem in $\BQLOGTIME$ that is not solvable by constant-depth $\GC^0(k)$ for $k = \frac{O(n^{1/4d})}{\log^{\omega(1)} n}$ and size $\quasipoly(n)$, where $n$ is the input size. 
\end{corollary}

Our circuit separation also says something about oracle separations. 
By standard techniques, \cref{cor:d-fools-gac0,cor:bqlogtime-not-in-gac0}  imply an oracle $A$ relative to which $\BQP^A \not\subseteq \mathsf{C}^A$ for any class of languages $\mathsf{C}$ that can be decided by a uniform family of constant-depth, exponential-size $\GC^0$ circuits.\footnote{The notion of uniformity here is sometimes called direct connect uniform \cite[Definition 6.28]{arora2009computational} or highly uniform \cite[Exercise 3.8]{goldreich2008computational}.}

\begin{corollary}[Generalization of {\cite[Corollary 1.5]{raz2022oracle}}]\label{thm:oracle-separation-bqp}
There is an oracle relative to which $\BQP$ is not contained in the class of languages decidable by uniform families of circuits $\{C_n\}$, where for all $n \in \N$, $C_n$ is a size-$2^{n^{O(1)}}$ depth-$d$ $\GC^0(k)$ circuit with $k \in \frac{2^{n/4d}}{n^{\omega(1)}}$. 
\end{corollary}

The proof is the same as \cite[Appendix A]{raz2022oracle} but the step where they apply their Theorem 1.2 should be replaced with our \cref{cor:d-fools-gac0}. Hence, we omit the details. Similar proofs were also given by Aaronson \cite{aaronson2010bqp} and Fefferman, Shalteil, Umans, and Viola \cite{fefferman2012beating}, which were based on an earlier work of Bennett and Gill \cite{charles1981relative}.

It is well-known that $\PH$ is the class of languages decided by uniform families of size-$2^{n^{O(1)}}$ constant-depth $\AC^0$ circuits (see e.g., \cite[Theorem 6.29]{arora2009computational}).
Therefore, the separation of $\BQP$ and $\PH$ is a special case of our theorem, because $\AC^0 \subseteq \GC^0(k)$ for all $k \geq 0$. 

Because $\g(k)$ gates can compute many functions, \cref{thm:oracle-separation-bqp} can be instantiated in many ways. 
We give one concrete example separating $\BQP$ from a biased version of the counting hierarchy, which we now define. 
First, we define existential and universal counting quantifiers.
Similar definitions date back to \cite{wagner1986complexity,toran1991,allender1993counting}. 
For a bit string $x$, let $\len(x)$ denote the length of $x$.   

\begin{definition}[Counting quantifiers]\label{def:Counting-quant}
Let $\mathsf{C}$ be a class of languages, and let $k: \N\to\N$ be a function.
Define $\exists_k \cdot \mathsf{C}$ to be the set of all languages $L$ such that there is some polynomial $p$ and a language $C \in \mathsf{C}$ such that $x \in L \iff$
\[
\abs{\{y \in \bitz^{p(\len(x))} : \langle x, y \rangle \in C \}}> k(\len(x)). 
\]

Define $\forall_k \cdot \mathsf{C}$ to be the set of all languages $L$ such that there is some polynomial $p$ and a language $C \in \mathsf{C}$ such that for $x \in \bitz^n$, $x \in L \iff$
\[
\abs{\{y \in \bitz^{p(\len(x))} : \langle x, y \rangle \notin C \}}\le k(\len(x)). 
\]
\end{definition}

We note that $\exists_0 = \exists$ and $\forall_0 = \forall$.

We can now define the $k$-biased counting hierarchy. 
For two functions $f_1,f_2: \N\to\N$, we say $f_1\le f_2$ when $\forall n, f_1(n)\le f_2(n)$. 

\begin{definition}[Biased counting hierarchy]\label{def:biased-ch}
Let $k:\N\to\N$ be a function. The $k$-biased counting hierarchy $\mathsf{CH}(k)$ is the smallest family of language classes satisfying:
\begin{itemize}
    \item[(i)] $\PTIME \in \mathsf{CH}(k)$,
    \item[(ii)] If $L \in \mathsf{CH}(k)$, then $\exists_{k'} \cdot L$ and $\forall_{k'} \cdot L \in \mathsf{CH}(k)$ for all $k':\N\to\N$, $k' \leq k$. 
\end{itemize}
\end{definition}

It is a well-known fact that the polynomial hierarchy $\PH$ can be characterized by alternating $\exists_0$ and $\forall_0$ quantifiers, so $\mathsf{CH}(0) = \PH$. 
As mentioned above, there is also a well-known characterization of $\PH$ by $\AC^0$ circuits. 
Roughly speaking, the $\exists_0$ quantifiers are replaced by $\OR$ gates, and the $\forall_0$ quantifiers are replaced by $\AND$ gates.  
Let $k$-$\OR$ be the gate that is $1$ iff $> k$ input bits are $1$.
Similarly, let $k$-$\AND$ be the gate that is $0$ iff $> k$ input bits are $0$.  
Observe $\OR = 0$-$\OR$ and $\AND = 0$-$\AND$, and that $k$-$\AND$ and $k$-$\OR$ are $\g(k)$ gates up to negations (specifically, one can construct $k$-$\AND$ with $\NOT$ and $k$-$\OR$ via De Morgan's law).  
So, in \emph{exactly} the same manner as $\PH$, for any class $\mathsf{C} \in \mathsf{CH}(k)$, one can construct a $\GC^0(k)$ circuit that decides an $L \in \mathsf{C}$ by replacing the $\exists_k$ quantifiers with $k$-$\OR$ gates and the $\forall_k$ quantifiers with $k$-$\AND$ gates.
By doing this standard construction, one obtains the following corollary of \cref{thm:oracle-separation-bqp}.

\begin{corollary}\label{cor:biased-ch-separation}
There is an oracle $A$ relative to which $\BQP^A \not\subseteq \mathsf{CH}(k)^A$ for $k(n) =\frac{2^{\Theta(n)}}{n^{\omega(1)}}$. 
\end{corollary}

This result is perhaps surprising considering that $\BQP \subseteq \PP$ relative to all oracles \cite{adleman1997quantum} and $\PP$ is the first level of the standard counting hierarchy.
Thus, our \cref{cor:biased-ch-separation} shows that a relativizing simulation of $\BQP$ requires being able to count a larger number of witnesses (exponential in the input instance size), as one can in $\PP$.

More broadly, \cref{thm:oracle-separation-bqp} separates $\BQP$ from many complexity classes that contain $\PH$ and are incomparable with $\PP$; the specific complexity classes one gets depends on how the $\G(k)$ gates are used in the uniform circuit families. Above we gave an example where the $\g(k)$ gates are all $k$-$\AND$ and $k$-$\OR$ gates.

\subsection{\texorpdfstring{Separation Between $\QNC^0$ and $\GC^0$}{Separation Between QNC0 and GC0}}\label{subsec:qnc-gc0}

We exhibit a search problem with input size $n$ that can be solved by $\QNC^0$ circuits (i.e., polynomial-size, constant-depth quantum circuits with bounded fan-in gates), but not by size-$s$ $\GC^0(k)$ circuits for $s \leq \exp(n^{1/4d})$ and $k=O(\log s)$. 
In particular, we show strong average-case lower bounds against $\GC^0$ for this search problem, i.e., that any $\GC^0$ circuit can only succeed on an $\exp(-n^c)$ fraction of input strings for some $c > 0$. In \cref{subsec:noisy-qnc0}, we show that our separation holds even when the quantum circuits are subject to noise.

Our result builds on prior work of Bravyi, Gosset, and K\"onig \cite{bravyi2018quantum} and Watts, Kothari, Schaeffer, and Tal \cite{watts2019exponential}. In particular, we use the same search problems introduced in these works and prove that they are average-case hard for $\GC^0$.
To prove our lower bound, we prove a new multi-switching lemma for multi-output $\GC^0$ circuits in \cref{subsec:muli-switching}.

Bravyi et al.\ introduced the 2D Hidden Linear Function Problem and showed that it can be solved by $\QNC^0$ circuits.

\begin{definition}[2D Hidden Linear Function Problem, \tdhlf\, \cite{bravyi2018quantum}]\label{def:tdhlf}
   Let $b \in \{0,1,2,3\}^n$ be a vector and let $A \in \{0,1\}^n$ be a symmetric matrix describing an $n\times n$\, $2$D grid, i.e., $A_{ij} = 1$ when vertices $i$ and $j$ are connected on the 2D grid. Define $q: \F_2^n \to \Z_4$ as $q(u) \coloneqq u^TAu + b^Tu \pmod 4$. Define $\calL_q$ as 
   \[
   \calL_q \coloneqq \left\{ u \in \F_2^n : \forall v\in \F_2^n,\, q(u \oplus v) = q(u) + q(v) \pmod 4 \right\}.
   \]
   $\oplus$ denotes the addition of binary vectors modulo two, and the addition $q(u) + q(v)$ is modulo four. 
   Bravyi, Gosset, and K\"onig \cite{bravyi2018quantum} show that the restriction of $q$ onto $\calL_q$ is a linear form, i.e., there exists a $z \in \F_2^n$ such that $q(x) = 2 z^Tx \pmod 4$ for all $x \in \calL_q$. 
   Given $A \in \{0,1\}^{n \times n}$ and $b \in \{0,1,2,3\}^n$ as input, the \textsc{2D Hidden Linear Function} (\tdhlf) problem is to output one such $z \in \F_2^n$. 
\end{definition}

Subsequently, Watts et al.\ \cite{watts2019exponential} introduced the Parallel Parity Halving Problem and showed that it reduces to 2D HLF.

\begin{definition}[Parallel Parity Halving Problem, $\php_{n,m}^r$\,\cite{watts2019exponential}]\label{def:php}
Given $r$ length-$n$ strings $x_1,\dots,x_r \in \{0,1\}^n$ as input, promised that each $x_i$ has even parity, output $r$ length-$m$ strings $y_1,\dots,y_r\in\{0,1\}^m$ such that, for all $i \in [p]$, 
\[
\abs{y_i} \equiv \frac{1}{2}\abs{x_i} \pmod 2.
\]
\end{definition}

\begin{lemma}[{\cite[Theorem 26, Corollary 30]{watts2019exponential}}]\label{lem:php-reduction-2dhlf}
    $\php_{m,n}^n$ reduces to \tdhlf.
\end{lemma}

Hence, to obtain our separation between $\QNC^0$ and $\GC^0(k)$ it suffices to prove that $\php$ is hard for $\GC^0$, which we do in the remainder of this subsection. Doing so yields the following result. 

\begin{theorem}[Generalization of {\cite[Theorem 1]{watts2019exponential}}]\label{thm:watts-main}
The \tdhlf\ problem on $n$ bits cannot be solved by a size-$\exp(O(n^{1/4d}))$ $\GC^0_d(k)$ circuit with $k=O(n^{1/4d})$. 
Furthermore, there exists an (efficiently samplable) input distribution on which any $\GC^0_d(k)$ circuit (or $\GC^0_d(k)\mathsf{/rpoly}$ circuit) of size at most $\exp(n^{1/4d})$ only solves the \tdhlf\ problem with probability at most $\exp(-n^c)$ for some $c > 0$. 
\end{theorem}

In \cref{subsec:muli-switching}, we prove a multi-switching lemma for multi-output $\GC^0$ circuits necessary for our lower bound. 
In \cref{subsec:gc0-lower-bound}, we prove that $\php$ is average-case hard to compute for $\GC^0$ circuits, yielding \cref{thm:watts-main}.
Finally, in \cref{subsec:noisy-qnc0}, we generalize our result further to obtain an exponential separation between \emph{noisy} $\QNC^0$ circuits and $\GC^0(k)$, building on the work of Bravyi, Gosset, K\"onig, and Temamichel \cite{bravyi2020quantum} and Grier, Ju, and Schaeffer \cite{grier2021interactive}.

\subsubsection{A Multi-Switching Lemma for \texorpdfstring{$\GC^0$}{GC0}}\label{subsec:muli-switching}

We prove a multi-output multi-switching lemma for $\GC^0(k)$, building on prior works of Rossman \cite{ros17entropyswitch}, H\aa stad \cite{haastad2014correlation}, and Kumar \cite{kumar2023tight}.
We must first establish some notation, following Rossman \cite{ros17entropyswitch} and Watts et al.\ \cite[Appendix A]{watts2019exponential}. Our proof involves the following classes of functions:
\begin{itemize}
    \item $\DT(w)$ is the class of depth-$w$ decision trees.
    \item $\CKT(k; d; s_1,\dots,s_d)$ is the class of depth-$d$ $\GC^0(k)$ circuits with $s_i$ gates at the $i$th layer of the circuit for all $i \in \{1,\dots, d\}$. Note that these circuits have $s_d$ many output bits.
    \item $\CKT(k; d; s_1, \dots, s_d) \circ \DT(w)$ is the class of circuits in $\CKT(k; d; s_1,\dots,s_d)$ whose inputs are labeled by depth-$w$ decision trees. Note that these functions have $s_d$ many output bits.
    \item $\DT(t) \circ \CKT(k; d; s_1, \dots, s_d) \circ \DT(w)$ is the class of depth-$t$ decision trees whose leaves are labeled by functions in $\CKT(k;d;s_1,\dots,s_d) \circ \DT(k)$. Note that these functions have $s_d$ many output bits.
    \item $\DT(w)^m$ is the class of $m$-tuples of depth-$k$ decision trees. This function has $m$ many output bits, where each output bit is computed by an element of $\DT(w)$. 
    \item $\DT(t) \circ \DT(w)^m$ is the class of depth-$t$ decision trees where each leaf is labeled by an $m$-tuple of depth-$k$ decision trees. Note that these functions have $m$ many output bits.
\end{itemize}

In the remainder of this subsection, we will build to the multi-switching lemma by combining ingredients from Rossman \cite{ros17entropyswitch} and Kumar \cite{kumar2023tight}.
To begin, we need the following lemma that says, with high probability, a depth-$\ell$ decision tree will reduce in depth under random restriction. 

\begin{lemma}[{\cite[Lemma 20]{ros17entropyswitch}}]
\label{thm:syntactic}
For a depth-$\ell$ decision tree $T \in \DT(\ell)$, 
\[\Pr_{\rho\sim \calR_p} [ T|_\rho\text{ has depth}\ge t]\le (2ep\ell/t)^t.\]
\end{lemma}

We also need the multi-switching lemma for $\GC^0$.

 \begin{lemma}[{\cite[Theorem 4.8, Lemma 4.9]{kumar2023tight}}]\label{thm:gc0-multiswitch}
 Let $f\in  \CKT(k;d; s_1,\dots, s_d)\circ \DT(w)$, then
\[\Pr_{\rho\sim \calR_p}[f|_\rho\notin \DT(t-1)\circ \CKT(k;d-1; s_2,\dots, s_d)\circ \DT(r)]\le 4(64(2^ks_1)^{1/r}pw)^t.\]
\end{lemma}
\begin{proof}
This follows immediately from \cite[Theorem 4.8, Lemma 4.9]{kumar2023tight}. 
We include the details for completeness.
The bottom two layers of $f$ are $s_1$ elements of $\G(k) \circ \DT(w)$, i.e., $\g(k)$ gates whose inputs are labeled by depth-$w$ decision trees.   
\cite[Lemma 4.9]{kumar2023tight} shows that $\g(k) \circ \DT(w)$ is equivalent to $\g(k) \circ \AND_w$, i.e., a depth-2 circuit whose bottom layer has fan-in-$w$ $\AND$ gates that feed into a $\g(k)$ gate one the top layer. Hence, the $s_1$ $\G(k)\circ \DT(w)$ substructures in $f$ can be viewed as $s_1$ $\G(k)\circ \AND_w$ subcircuits. 
To complete the proof, apply \cite[Theorem 4.8]{kumar2023tight} to these $s_1$ subcircuits. 
\end{proof}

We can now show that under random restriction elements of 
$\DT(t-1)\circ \CKT(k;d; s_1,\dots, s_d)\circ \DT(w)$ simplify to elements of 
$\DT(t-1)\circ \CKT(k;d-1; s_2,\dots, s_d)\circ \DT(r)$ with high probability.

\begin{lemma}[Generalization of {\cite[Lemma 24]{ros17entropyswitch}}]
\label{thm:combinedmultiswitch}
 Let $f\in \DT(t-1)\circ \CKT(k;d; s_1,\dots, s_d)\circ \DT(w)$, then
\[\Pr_{\rho\sim \calR_p}[f|_\rho\notin \DT(t-1)\circ \CKT(k;d-1; s_2,\dots, s_d)\circ \DT(r)]\le 5(64(2^ks_1)^{1/r}pw)^t. \]
\end{lemma}

\begin{proof}
Say $f$ is computed by a depth-$(t-1)$ decision tree $T$, where each leaf $\ell$ is labeled by a circuit $C_\ell\in \CKT(k;d; s_1, s_2,\dots, s_d)\circ \DT(w)$. 
Let $\calE_1$ be the event $T|_{\rho}$ has depth $\le \lfloor t/2\rfloor - 1$, and let $\calE_2$ be the event $C_\ell|_\rho\in \DT(\lceil t/2 \rceil - 1)\circ \CKT(k;d-1;s_2,\dots, s_d)\circ \DT(r)$ for all leaves $\ell$ of $T$. Note that \[\calE_1\wedge \calE_2 \implies f|_\rho\in \DT(t-1)\circ \CKT(k;d-1;s_2,\dots, s_d)\circ \DT(r).\] 
By \Cref{thm:syntactic}, we know \[\Pr_{\rho\sim \calR_p}[\neg\calE_1] \le (2ep(t-1)/\lceil t/2 \rceil)^{\lceil t/2 \rceil} \le (4ep)^{t/2}.\]
By \Cref{thm:gc0-multiswitch} and a union bound, we have
\begin{align*}
        \Pr_{\rho\sim \calR_p}[\neg\calE_2] & \le \sum_{\text{leaves }\ell} \Pr[C_\ell|_\rho\notin \DT(\lceil t/2\rceil - 1)\circ \CKT(k;d-1;s_2,
        \dots, s_d)\circ \DT(r)] \\ & \le \sum_\ell 4(64(2^ks_1)^{1/r}pw)^t  \\ 
        & \le 2^t\cdot 4(64(2^ks_1)^{1/r}pw)^t \\
        & = 4(128(2^ks_1)^{1/r}pw)^t.
\end{align*}
Therefore, we can finally bound \begin{align*}
        \Pr_\rho[f|_\rho\notin \DT(t-1)\circ \CKT(k;d-1;s_2,\dots, s_d)\circ \DT(r)] & \le \Pr_\rho[\neg\calE_1] + \Pr_\rho[\neg \calE_2] \\ &  \le (4ep)^{t/2} + 4(128(2^ks_1)^{1/r}pw)^t \\ &  \le 5(128(2^ks_1)^{1/r}pw)^t.\qedhere
    \end{align*}
\end{proof}

\cref{thm:combinedmultiswitch} shows a depth reduction by $1$ under random restriction. At a high level, our argument will repeat this process $d$ times to simplify the depth of the circuit to $1$ with high probability. When the depth has simplified to $1$, we will need the following form of the multi-switching lemma for $\GC^0$ to complete our argument.

\begin{theorem}[{\cite[Theorem 4.8, Lemma 4.9]{kumar2023tight}} restated]
\label{thm:basecase}
    Let $f\in \CKT(k;1; m)\circ \DT(w)$. Then 
    \[\Pr_{\rho\sim \calR_p}[f|_\rho\notin \DT(t-1)\circ\DT(r-1)^m]\le  4(64(2^km)^{1/r}pw)^t.\]
\end{theorem}
\begin{proof}
Like \cref{thm:gc0-multiswitch}, this follows immediately from \cite[Theorem 4.8, Lemma 4.9]{kumar2023tight}. 
\end{proof}

We are now ready to prove our multi-output multi-switching lemma for $\GC^0(k)$, the main theorem of this subsection. 

\begin{theorem}[Multi-Output Multi-Switching Lemma for $\GC$]
\label{thm:multioutputlemma}
   Let $f\in\CKT(k;d;s_1,\dots, s_{d-1}, m)$ with $n$ inputs and $m$ outputs. Let $s = s_1+\dots+s_{d-1} + m$. Let $p = p_1\cdot p_2\cdots p_d$ and $w \coloneqq \lceil \log s\rceil + 1$. 
   Then \[\Pr_{\rho\sim \calR_p} [f|_\rho\notin \DT(2t-2)\circ \DT(r-1)^m]\le 5(128\cdot 2^{k/w}p_1)^t + \sum_{i=2}^{d-1} 5(128\cdot 2^{k/w}p_iw)^t + 4(128(2^km)^{1/r}pw)^t.\] 
\end{theorem}

\begin{proof}
    Let $s_d\coloneqq m$.  Notice we can factor $\rho\sim \calR_p$ as $\rho_1\circ\dots\circ \rho_d$, where each $\rho_i\sim \calR_{p_i}$. Now for each $i\in [d-1]$, define the event 
    \[\calE_i \iff f|_{\rho_1\cdot\dots\cdot \rho_i}\in \DT(t-1)\circ\CKT (d-i; s_{i+1},\dots, s_d)\circ \DT(w), \] and define \[\calE_d \iff f|_{\rho_1\circ\dots\circ \rho_d} \in \DT(2t-2)\circ \DT(r-1)^m.\] Notice that \[\bigwedge_{i=1}^d \calE_i \implies \calE_d\iff  f|_\rho\in \DT(2t-2)\circ\DT(q-1)^m.\]

    We will bound the complement of this event. Notice that since \[f\in \CKT(k;d; s_1,\dots, s_d)\subset \DT(t-1)\circ \CKT(k;d; s_1,\dots, s_d)\circ \DT(1), \]
     we have by \Cref{thm:combinedmultiswitch} that \[\Pr_\rho[\neg \calE_1]\le 5(64(2^ks_1)^{1/w}p_1)^t \le 5(128\cdot 2^{k/w}p_1)^t .\] For $i = 2,\dots, d-1$, \Cref{thm:combinedmultiswitch} gives us \[\Pr[\neg\calE_i|\calE_1,\dots, \calE_{i-1}]\le 5(64(2^ks_i)^{1/w}p_iw)^t\le 5(128\cdot 2^{k/w}p_iw)^t.\]

     We now bound $\Pr[\neg\calE_d|\calE_1,\dots, \calE_{d-1}]$. Let $g\coloneqq f|_{\rho_1\circ\dots\circ \rho_{d-1}}$. Conditioning on $\calE_1,\dots,\calE_{d-1}$, we have \[g\in \DT(t-1)\circ \CKT(k;1;m)\circ \DT(w).\] For each leaf $\ell$ of the partial decision tree of depth $t-1$ for $g$,
     define $g_\ell$ to be $g$ restricted by the root-to-leaf path in the tree to $\ell$. It follows that each $g_\ell$, by definition, is $\CKT(k;1;m)\circ \DT(w)$. Consequently, by \Cref{thm:basecase}, 
     we have for each $\ell$, \[\Pr[g_\ell|_{\rho_d}\notin \DT(t-1)\circ \DT(r-1)^m]\le 4(64(2^km)^{1/r}pw)^t.\] 
     As there are $2^{t-1}$ leaves, by a union bound it follows that the probability \emph{some} $g_\ell$ doesn't simplify is at most \[4(128(2^km)^{1/r}pw)^t.\] 
     In the complementary event, we have \[g|_{\rho_d} = f|_{\rho_1\circ\dots\circ \rho_d}\in \DT(t-1)\circ \DT(t-1)\circ \DT(q-1)^m = \DT(2t-2)\circ \DT(q-1)^m, \] so event $\calE_d$ holds. We now finally bound \begin{align*}
         \Pr_{\rho\sim \calR_p} [f|_\rho\notin \DT(2t-2)\circ \DT(r-1)^m]  & = \Pr[\neg\calE_1, \dots, \neg\calE_d] \\ &  = \sum_{i=1}^d \Pr[\neg\calE_i|\calE_1,\dots, \calE_{i-1}] \\ & \le 5(128\cdot 2^{k/w}p_1)^t + \sum_{i=2}^{d-1} 5(128\cdot 2^{k/w}p_iw)^t + 4(128(2^km)^{1/r}pw)^t.
     \end{align*}
\end{proof}

\begin{corollary}\label{cor:muli-output-multi-switching}
    Let $f:\bits^n\to\bits^m$ be computable by a $\GC^0(k)$ circuit of size $s$, depth $d$, and $k = O(\log s)$. Let $p = \frac{1}{m^{1/q}O(\log s)^{d-1}}$. Then, for all $t \in \N$, 
    \[\Pr_{\rho\sim \calR_p}[f|_\rho\notin \DT(2t-2)\circ \DT(q-1)^m] \le 2^{-t}.
    \]
\end{corollary}

\begin{proof}
    For $w\coloneqq \lceil \log s\rceil$, we have that $2^{k/w} = O(1)$. Using this and applying \Cref{thm:multioutputlemma} with $p_1 = \Omega(1)$, $p_2 = \dots = p_{d-1} = \Omega(1/w)$, and $p_d = 1/O(m^{1/q}w)$ yields the desired result.
\end{proof}

\subsubsection{\texorpdfstring{$\GC^0$}{GC0} Lower Bound}\label{subsec:gc0-lower-bound}

We can now use our multi-output multi-switching lemma to prove that $\php$ (\cref{def:php}) is hard $\GC^0$ circuits.

\begin{theorem}[$\php_{n,m}^r \notin \GC^0(k)$, Generalization of {\cite[Theorem 25]{watts2019exponential}}]\label{theorem:php-notin-gc0}
Let $r = n$ and $m \in [n,n^2]$. Any $\GC^0_d(k)$ circuit $F:\bitz^{nr}\to\bitz^{mr}$ with size $s \leq \exp\left((nr)^{\frac{1}{2d}}\right)$ and $k = O((nr)^{\frac{1}{2d}})$ solves $\php_{n,m}^r$ with probability at most $\exp\left(-n^2/(m^{1+o(1)}O(\log s)^{2(d-1)}\right)$.
\end{theorem}
\begin{proof}
Set $q = \sqrt{\log(mr)}$, $p = 1/(O(\log s)^{d-1} (mr)^{1/q})$, and $t = pnr/8$. Let $\rho$ be a $p$-random restriction. 
The only fact about $\AC^0$ used in the proof of \cite[Theorem 25]{watts2019exponential} is that a function $F$ computable by a size-$s$ $\AC^0$ circuit simplifies to an element of $\DT(2t) \circ \DT(q)^m$ under $\rho$ with probability at least $1 - \exp(-\Omega(pnr))$.
By \cref{cor:muli-output-multi-switching}, this holds for functions computable by size-$s$ $\GC^0(k)$ circuits with $k = O(\log s)$. Hence, the rest of the argument in \cite[Theorem 25]{watts2019exponential} goes through.
\end{proof}

With that, the main result follows. 

\begin{proof}[Proof of \cref{thm:watts-main}]
The result follows from combining \cref{lem:php-reduction-2dhlf} and \cref{theorem:php-notin-gc0}. 
\end{proof}

\subsubsection{\texorpdfstring{Separation Between Noisy $\QNC^0$ and $\GC^0$}{Separation Between Noisy QNC0 and GC0}}\label{subsec:noisy-qnc0}

Our separation between $\GC^0$ and $\QNC^0$ holds even when the $\QNC^0$ circuits are subjected to noise. 
The noise model considered is the \emph{local stochastic quantum noise model} \cite{fawzi2020constant,bravyi2020quantum} (see also \cite[Section 2.2]{grier2021interactive}). As in prior works, the noise rate is assumed to be below some constant threshold. Here and throughout, ``noisy $\QNC^0$'' refers to $\QNC^0$ subjected to local stochastic quantum noise with a certain constant noise rate.

Bravyi, Gosset, K\"onig, and Temamichel \cite{bravyi2020quantum} show that for any relation problem solvable by $\QNC^0$, one can construct a ``noisy version'' of that relation problem that is solvable by noisy $\QNC^0$ (\cite[Definition 15, Theorem 17]{bravyi2020quantum}, \cite[Definition 14, Theorem 15]{grier2021interactive}. 
Additionally, \cite[Lemma 16]{grier2021interactive} implies that any classical circuit solving the noisy version of the relation problem can solve the original relation problem with the overhead of first running a quasipolynomial-size $\AC^0$ circuit.

We can apply this framework to separate $\GC^0(k)$ and noisy $\QNC^0$. 

\begin{theorem}[Generalization of {\cite[Proposition 18, Theorem 19]{grier2021interactive}}]\label{thm:noisy-qnc0-gc0}
There is a search problem that is solvable by noisy $\QNC^0$ with probability $1 - \exp(-\Omega(\polylog(n)))$, but any size-$s$ depth-$d$ $\GC^0(k)$ circuit with $k = O(\log s)$ cannot solve the search problem with probability exceeding 
\[
\exp\left( \frac{-n^{1/2 - o(1)}}{O\left(\log( s + \exp(\polylog(n))) \right)^{2d+O(1)}} \right).
\]
\end{theorem}
\begin{proof}
Let the noisy \tdhlf\ be the relation problem obtained from applying \cite[Definition 14]{grier2021interactive} to the \tdhlf\ (\cref{def:tdhlf}). 
The quantum upper bound is precisely 
\cite[Proposition 18]{grier2021interactive}. 

Towards the classical lower bound, assume there exists a size-$s$, depth-$d$ $\GC^0(k)$ circuit with $k = O(\log s)$ that solves noisy \tdhlf\ with probability at most $\eps$.
Then, by \cite[Lemma 16]{grier2021interactive}, there exists a size-$(s + \exp(\polylog(n)))$, depth-$(d+O(1))$ $\GC^0(k)$ circuit with $k= O(\log(s + \exp(\polylog(n)))$ that solves \tdhlf\ with probability at most $\eps$.
But, by \cref{theorem:php-notin-gc0} and \cite[Theorem 26, Corollary 30]{watts2019exponential}, any $\GC^0(k)$ circuit of size $s$, depth $d$, and $k = O(\log s)$ for \tdhlf\  succeeds with probability at most 
\[
\exp\left(\frac{-n^{1/2-o(1)}}{O(\log s)^{2d}}\right).
\]
Therefore, we can conclude that 
\[
\eps \leq \exp\left(\frac{-n^{1/2-o(1)}}{O\left(\log\left(s + \exp(\polylog(n))\right) 
\right)^{2d+O(1)}}\right).\qedhere
\]
\end{proof}

\subsection{Separation Between \texorpdfstring{$\QNC^0\mathsf{/qpoly}$}{QNC0/qpoly} and \texorpdfstring{$\GC^0(k)[2]$}{GC0(k)[2]}}\label{subsec:gc2}
We exhibit a relation problem that can be solved with high probability by a $\QNC^0\mathsf{/qpoly}$ circuit but is average-case hard for $\GC^0(k)[2]\mathsf{/rpoly}$. 
Recall that $\QNC^0\mathsf{/qpoly}$ is the class of $\QNC^0$ circuits with quantum advice, i.e., polynomial-size, constant-depth quantum circuits with bounded fan-in gates that can start with any quantum state as long as it is independent of the input.  

Our argument follows the same structure as Watts et al.\ \cite{watts2019exponential}. However, we obtain an \emph{exponential separation} between $\GC^0(k)[2]\mathsf{/rpoly}$ and $\QNC^0\mathsf{/qpoly}$. Previously, the best separation was between $\QNC^0\mathsf{/qpoly}$ and \emph{polynomial-size} $\AC^0[2]$ circuits. 

\begin{theorem}[Generalization of {\cite[Theorem 6]{watts2019exponential}}]\label{thm:gc2}
There is a search problem that is solvable by $\QNC^0\mathsf{/qpoly}$ with probability $1-o(1)$, but any $\GC^0(k)[2]\mathsf{/rpoly}$ circuit of depth $d$ and size at most $\exp\left(O\left(n^{1/2.01d} \right) \right)$ with $k=O(n^{1/2d})$ cannot solve the search problem with probability exceeding $n^{-\Omega(1)}$.
\end{theorem}

The remainder of this subsection is devoted to proving \cref{thm:gc2}.
The quantum upper bound is given in \cite[Section 6.1, Section 6.3]{watts2019exponential}.
We will show an average-case lower bound for the following problem.

\begin{definition}[$r$-Parallel Parity Bending Problem {\cite[Problem 8]{watts2019exponential}}]\label{def:r-parallel-parity-bending}
  Given inputs $x_1, \dots, x_r$ with $x_i \in \{0,1,2\}^n$ for all $i \in [r]$, produce outputs $y_1,\dots,y_r \in \{0,1\}^n$ such that $y_i$ satisfies: 
  \begin{align*}
      &\abs{y_i} \equiv 0 \pmod{2} \quad\text{if}\quad \abs{x_i} \equiv 0 \pmod{3}\quad \text{or} \\
      &\abs{y_i} \equiv 1 \pmod{2} \quad\text{if}\quad \abs{x_i} \not\equiv 0 \pmod{3}.
  \end{align*}
  for at least a $\frac{2}{3} + 0.005$ fraction of the $i \in [k]$.
\end{definition}

Note that this problem takes input over $\{0,1,2\}$. Ultimately we are studying Boolean circuits, so, technically speaking, trits are encoded with two bits (e.g., $0 \mapsto 00$, $1\mapsto 01$, $2 \mapsto 10$). We use $\{0,1,2\}$ for notational convenience.

On the way to our lower bound, we first prove lower bounds for the following problem. 

\begin{definition}[3 Output Mod 3 {\cite[Problem 9]{watts2019exponential}}]\label{def:3-output-mod-3}
Given an input $x \in \{0, 1, 2\}^n$, output a trit $y \in \{0, 1, 2\}$ such that $y \equiv \abs{x} \pmod 3$. 
\end{definition}

To prove 3 Output Mod 3 is hard for $\GC^0(k)[2]$, we use the following worst-case to average-case reduction, given in \cite{watts2019exponential}.

\begin{lemma}\label{lemma:worst-to-average}
Suppose there is a $\GC^0(k)[2]\mathsf{/rpoly}$ circuit of size $s$ and depth $d$ that solves 3 Output Mod 3 (\cref{def:3-output-mod-3}) on a uniformly random input with probability $1/3 + \eps$ for some $\eps > 0$. 
Then there exists a $\GC^0(k)[2]\mathsf{/rpoly}$ circuit $C$ of depth $d + O(1)$ and size $s + O(n)$ such that for any $x \in \{0,1,2\}^n$,
\begin{align*}
    &\Pr[C(x) \equiv \abs{x} \pmod 3] = \frac{1}{3} + \eps, \text{ and} \\
    &\Pr[C(x) \equiv \abs{x} + 1 \pmod 3] = \Pr[C(x) \equiv \abs{x} + 2 \pmod 3] = \frac{1}{3} - \frac{\eps}{2}.
\end{align*}
\end{lemma}
\begin{proof}
    The proof is exactly the same as \cite[Lemma 35]{watts2019exponential}.
\end{proof}

We can now show that 3 Output Mod 3 is average-case hard for exponential-size $\GC^0(k)[2]$ circuits.

\begin{lemma}[Generalization of {\cite[Lemma 36]{watts2019exponential}}]\label{lem:3-output-mod-3-lower-bound}
Let $k = O(n^{1/2d})$.
Any $\GC^0(k)[2]\mathsf{/rpoly}$ circuit of depth $d$ and size $s \leq \exp\left(O\left(n^{1/2.01d} \right) \right)$ solves 3 Output Mod 3 (\cref{def:3-output-mod-3}) on the uniform distribution with probability at most $\frac{1}{3} + \frac{1}{n^{\Omega(1)}}$.
\end{lemma}
\begin{proof}
Let $C$ be the $\GC^0(k)[2]\mathsf{/rpoly}$ circuit that solves 3 Output Mod 3 on the uniform distribution with probability $\frac{1}{3} + \eps$. 
\cref{lemma:worst-to-average} implies that there is a circuit $C'$ that succeeds with probability $\frac{1}{3} + \eps$ and outputs each wrong answer with probability $\frac{1}{3} - \frac{\eps}{2}$.

Let $E : \{0,1,2\} \to \{0,1\}$ be the circuit that maps $0$ to $0$ and everything else to $1$. 
Define $C''$ to be the circuit that, given input $x$, outputs $0$ with probability $\frac{1}{4}$, and outputs $E(C''(x))$ otherwise. 
Observe that, if $\abs{x} \equiv 0 \pmod 3$, then $C''$ correctly outputs $0$ with probability $\frac{1}{4} + \frac{3}{4}(\frac{1}{3} + \eps) = \frac{1}{2} +\frac{3\eps}{4}$. Similarly, if $\abs{x} \not\equiv 0 \pmod 3$, then $C''$ correctly outputs $1$ with probability $\frac{3}{4}(\frac{1}{3} + \eps + \frac{1}{3} - \frac{\eps}{2}) = \frac{1}{2} + \frac{3\eps}{8}$. Hence $C''$ computes $\Mod_3$ with probability $\frac{1}{2} + \frac{3\eps}{8}$, so \cref{thm:mod-p-correlation-bounds} implies that $\eps \in \frac{1}{n^{\Omega(1)}}$.
\end{proof}

The average-case lower bound in \cref{lem:3-output-mod-3-lower-bound} implies the following corollary. 

\begin{corollary}[Generalization of {\cite[Corollary 37]{watts2019exponential}}]\label{cor:3-output-3-lb}
Let $k = O(n^{1/2d})$.
Let $C$ be a $\GC^0(k)[2]/\mathsf{rpoly}$ circuit of depth $d$ and size $s \leq \exp\left(O\left(n^{1/2.01d} \right) \right)$ outputting a trit.
Then, for all $i \in \{0, 1, 2\}$,
\[
\frac{1}{3} - \frac{1}{n^{\Omega(1)}} \leq \Pr_{x \in \{0,1,2\}^n}\left[C(x) - \abs{x} \equiv i \pmod 3 \right] \leq \frac{1}{3} + \frac{1}{n^{\Omega(1)}}.
\]
\end{corollary}
\begin{proof}
Because 
\[
\sum_{i \in \{0,1,2\}} \Pr_{x \in \{0,1,2\}^n}\left[C(x) - \abs{x} \equiv i \pmod 3 \right] =1,
\]
it suffices to prove 
\[
\Pr_{x \in \{0,1,2\}^n}\left[C(x) - \abs{x} \equiv i \pmod 3 \right] \leq \frac{1}{3} + \frac{1}{n^{\Omega(1)}}
\]
for each $i \in \{0,1,2\}$.
For $i = 0$, the desired bound is exactly shown in \cref{lem:3-output-mod-3-lower-bound}. 
For $i \in \{ 1, 2\}$, observe that if there is a $\GC^0(k)[2]\mathsf{/rpoly}$ circuit $D$ of depth $d$ and size at most $\exp\left(O\left(n^{1/2.01d} \right) \right)$ for which 
\[
\Pr[D(x) - \abs{x} \equiv i \pmod 3] \geq \frac{1}{3} + \frac{1}{n^{o(1)}}, 
\]
then one could construct a circuit $D'$ for which
\[
\Pr[D'(x)  \equiv \abs{x} \pmod 3] \geq \frac{1}{3} + \frac{1}{n^{o(1)}}
\]
by subtracting by the trit $i$ at the end of the circuit. Subtracting by a fixed trit only adds a constant overhead to the size and depth of the circuit, so such a $D'$ contradicts \cref{lem:3-output-mod-3-lower-bound}. 
\end{proof}

We note that \cite[Corollary 37]{watts2019exponential} is only stated for polynomial-size $\AC^0[2]\mathsf{/rpoly}$ circuits. However, we observe the statement also holds for exponential-size circuits, as demonstrated in \cref{cor:3-output-3-lb}. 
This allows us to obtain exponentially stronger lower bounds than the ones obtained in \cite{watts2019exponential}. 

Now we study the difficulty of solving $r$ instances of the 3 Output Mod 3 Problem. 

\begin{definition}[$r$-Parallel 3 Output Mod 3]\label{def:paralel-3-output-mod}
  Given inputs $x_1,\dots,x_r \in \{0,1,2\}^{n}$, output a vector $\Vec{y}\in \{0,1,2\}^r$ such that 
  \[
  y_i \equiv \abs{x_i} \pmod 3
  \]
  for at least a $\frac{1}{3} + 0.01$ fraction of the $i \in [k]$.
\end{definition}

To prove lower bounds for this problem, we use the XOR lemma for finite abelian groups.

\begin{lemma}[{\cite[Lemma 4.2]{rao2007exposition}}, XOR lemma for finite abelian groups]\label{lem:xor-lemma-abelian}
   Let $\calD$ be a distribution over a finite abelian group $G$ such that $\abs{\E[\psi(X)]} \leq \eps$  for every non-trivial character $\psi$. Then $\calD$ is $\eps \sqrt{\abs{G}}$-close (in total variation distance) to the uniform distribution over $G$.  
\end{lemma}

\begin{theorem}[Generalization of {\cite[Theorem 39]{watts2019exponential}}]\label{thm:3-output-mod-3}
Let $k=O(n^{1/2d})$.
   There exists an $r\in \Theta(\log n)$ for which any $\GC^0(k)[2]\mathsf{/rpoly}$ circuit of depth $d$ and size $s \leq \exp\left( O\left(n^{1/2.01d} \right) \right)$ solves the $r$-Parallel 3 Output Mod 3 Problem (\cref{def:paralel-3-output-mod}) with probability at most $n^{-\Omega(1)}$. 
\end{theorem}
\begin{proof}
   For $k= O(n^{1/2d})$, let $C$ be a $\GC^0(k)[2]\mathsf{/rpoly}$ circuit of depth $d$ and size at most $\exp\left(O\left(n^{1/2.01d} \right) \right)$ that solves the $r$-Parallel 3 Output Mod 3 problem with probability $\eps$. 
    Throughout this proof, let $x_1,\dots,x_r \in \{0,1,2\}^n$ be chosen uniformly at random, and let $(y_1, \dots, y_r)$ be the output trits of the circuit $C$.
    Let $\calD$ be the distribution over $r$ trits defined by 
    \[
    \bigotimes_{i=1}^r (\abs{x_i} - y_i \pmod 3).
    \]

    We begin by showing that $\calD$ is close to the uniform distribution over $\{0,1,2\}^r$ in total variation distance.
    Let $\chi_a$ be the character of $\F_3^r$ corresponding to $a \in \F_3^r$. Recall that $\chi_a(z) \coloneqq \omega^{\sum_{i=1}^r a_i z_i}$, where $z \in \F_3^r$ and $\omega$ is a third root of unity.  
    To show that $\calD$ is close to uniform, it suffices to show that $\abs{\E[\chi_a(\calD)]}$ is small for all nonzero $a$.

    For $a \in \F_3^r$, let $S$ denote the set of indices on which $a_i \neq 0$. Consider the problem where, given a nonzero $a\in \F_3^r$ and strings $x_1,\dots,x_r \in \{0,1,2\}^n$, the goal is to find trits $y_1,\dots,y_r$ such that 
    \[
    \sum_{i \in S} a_i \abs{x_i} \equiv \sum_{i \in S} a_i y_i \pmod 3.
    \]
    This problem reduces to 3 Output Mod 3 on the concatenated input $\widetilde{x} \coloneqq (a_i x_{j,i})_{i \in S, j\in[r]} \in \{0,1,2\}^{n\abs{S}}$. Specifically, given any circuit $A$ solving the former problem, one can solve the latter problem by first running the circuit $A$ to obtain the trits $y_1,\dots,y_r$. Then, add a circuit to compute the sum $\sum_{i \in S} a_i y_i \pmod 3$, which is the correct answer to the 3 Output Mod 3 problem on input $\widetilde{x}$. This last step can be done with a depth-$2$ $\AC^0$ circuit with $\exp(\abs{S})\leq \exp(r) \leq \poly(n)$ many gates.

    Now, because we are choosing $x_1,\dots,x_r$ uniformly at random, the concatenated input $\widetilde{x} \in \{0,1,2\}^{n \abs{S}}$ is uniformly random. Therefore, \cref{cor:3-output-3-lb} implies that the distribution 
    \[
    \sum_{i \in S} a_i (\abs{x_i} - y_i) \pmod 3
    \]
    is at most $n^{-\Omega(1)}$-far from the uniform distribution over a trit $\{0,1,2\}$ in total variation distance. Hence, $\abs{\E[\chi_a(\calD)]} \leq n^{-\Omega(1)}$ for each nonzero $a$. 
    Then, \cref{lem:xor-lemma-abelian} implies that $\calD$ is $n^{-\Omega(1)}\sqrt{3^r}$-close to the uniform distribution on $\{0,1,2\}^r$.

    Because $\calD$ is close to uniform, the probability $\eps$ that the circuit $C$ solves the $r$-Parallel 3 Output Mod 3 problem is (almost) equivalent to the probability that a uniformly random string in $\{0,1,2\}^r$ has more than a $\frac{1}{3} + 0.01$ fraction of its trits set to $0$. 
    By a Chernoff bound, this probability is bounded above by $\exp\left(-\Omega(r)\right)$. More carefully, we see that the probability of $C$ solving the $r$-Parallel 3 Output Mod 3 problem is at most 
    \[
    n^{-\Omega(1)} \sqrt{3^r} + \exp\left(-\Omega(r)\right),
    \]
    which is bounded above by $n^{-\Omega(1)}$ for some $r \in \Theta(\log n)$.
\end{proof}

In \cite[Theorem 40]{watts2019exponential} they show that the $r$-Parallel Parity Bending Problem (\cref{def:r-parallel-parity-bending}) is as hard as the $r$-Parallel 3 Output Mod 3 Problem (\cref{def:paralel-3-output-mod}).
Their reduction and \cref{thm:3-output-mod-3} imply the following corollary.

\begin{corollary}\label{cor:gc02-lb}
Let $k = O(n^{1/2d})$.
   There exists an $r \in \Theta(\log n)$ for which any $\GC^0(k)[2]\mathsf{/rpoly}$ circuit of depth $d$ and size at most $\exp\left(O\left(n^{1/2.01d}\right) \right)$ solves the $r$-Parallel Parity Bending Problem with probability at most $n^{-\Omega(1)}$.
\end{corollary}

Combining \cref{cor:gc02-lb} with the quantum upper bound in \cite[Section 6]{watts2019exponential} implies \cref{thm:gc2}.

\subsection{Separation Between \texorpdfstring{$\QNC^0\mathsf{/qpoly}$}{QNC0/qpoly} and \texorpdfstring{$\GC^0(k)[p]$}{GC0(k)[p]}}\label{subsec:gcp}

We exhibit relation problems that can all be solved by $\QNC^0\mathsf{/qpoly}$ but each one is average-case hard for $\GC^0(k)[p]$ for some prime $p \neq 2$. Since we proved a separation when $p=2$ in the previous subsection (\cref{thm:gc2}), we have an exponential separation between $\QNC^0\mathsf{/qpoly}$ and $\GC^0(k)[p]$ for all primes $p$.

\begin{theorem}{\label{thm:gcp}}
    For any prime $p$, there is a search problem that is solvable by $\QNC^0\mathsf{/qpoly}$ with probability $1-o(1)$, but any $\GC^0(k)[p]\mathsf{/rpoly}$ circuit of depth $d$ and size at most $\exp\left(O(n^{1/2.01d}) \right)$ with $k = O(n^{1/2d})$ cannot solve the search problem with probability exceeding $n^{-\Omega(1)}$.
\end{theorem}

We also note that we use the case where $p=2$ to obtain separations for primes $p \neq 2$, which is why the $p=2$ case is handled in a separate subsection.

Previously, the best separation known was between $\QNC^0\mathsf{/qpoly}$ and polynomial-size $\AC^0[p]$ circuits, which was shown in the recent work of Grilo, Kashefi, Markham, and Oliveira \cite{grilo2024power}.
The case where $p=2$ was shown in \cref{subsec:gc2}. We handle all other primes in this subsection.
We will show lower bounds for the following problem, which is a natural generalization of the $r$-Parallel Parity Bending Problem introduced by \cite[Problem Problem 8]{watts2019exponential}.

\begin{definition}[$(q,r)$-Parallel Parity Bending Problem {\cite[Definition 4]{grilo2024power}}]\label{def:q-r-parallel-parity-bending}
  Given inputs $x_1, \dots, x_r$ with $x_i \in \{0,1\}^n$ for all $i \in [r]$, produce outputs $y_1,\dots,y_r \in \{0,1\}^n$ such that $y_i$ satisfies: 
  \begin{align*}
      &\abs{y_i} \equiv 0 \pmod{q} \quad\text{if}\quad \abs{x_i} \equiv 0 \pmod{2}\quad \text{or} \\
      &\abs{y_i} \not\equiv 0 \pmod{q} \quad\text{if}\quad \abs{x_i} \equiv 1 \pmod{2}.
  \end{align*}
  for at least a $\frac{2}{3} + 0.005$ fraction of the $i \in [k]$.
\end{definition}

Grilo et al.\ \cite{grilo2024power} prove that $\QNC^0\mathsf{/qpoly}$ can solve this problem. We prove that the problem is average-case hard for $\GC^0(k)[p]$ for all primes $p \neq 2$. We begin with the following corollary of \cref{thm:mod-p-correlation-bounds}.

\begin{corollary}\label{cor:2-output-2-uniform}
Let $k = O(n^{1/2d})$.
For a prime $p \neq 2$, let $C$ be a $\GC^0(k)[p]/\mathsf{rpoly}$ circuit of depth $d$ and size $s \leq \exp\left(O\left(n^{1/2.01d} \right) \right)$. 
Then, for all $i \in \{0, 1\}$,
\[
\frac{1}{2} - \frac{1}{n^{\Omega(1)}} \leq \Pr_{x \in \{0,1\}^n}\left[C(x) - \abs{x} \equiv i \pmod 2 \right] \leq \frac{1}{2} + \frac{1}{n^{\Omega(1)}}.
\]
\end{corollary}
\begin{proof}
The proof is similar to \cref{cor:3-output-3-lb}.
Because 
\[
\sum_{i \in \{0,1\}} \Pr_{x \in \{0,1\}^n}\left[C(x) - \abs{x} \equiv i \pmod 2 \right] =1,
\]
it suffices to prove 
\[
\Pr_{x \in \{0,1\}^n}\left[C(x) - \abs{x} \equiv i \pmod 2 \right] \leq \frac{1}{2} + \frac{1}{n^{\Omega(1)}}
\]
for each $i \in \{0,1\}$.
For $i = 0$, the desired bound is exactly shown in \cref{thm:mod-p-correlation-bounds}. 
For $i = 1$, observe that if there is a $\GC^0(k)[p]\mathsf{/rpoly}$ circuit $D$ of depth $d$ and size at most $\exp\left(O\left(n^{1/2.01d} \right) \right)$ for which 
\[
\Pr[D(x) - \abs{x} \equiv 1 \pmod 2] \geq \frac{1}{2} + \frac{1}{n^{o(1)}}, 
\]
then one could construct a circuit $D'$ for which
\[
\Pr[D'(x)  \equiv \abs{x} \pmod 2] \geq \frac{1}{2} + \frac{1}{n^{o(1)}}
\]
adding a $\NOT$ gate to the end of the circuit. 
However, such a $D'$ cannot exist as it contradicts \cref{thm:mod-p-correlation-bounds}. 
\end{proof}

We now prove our average-case lower bound.

\begin{theorem}
Let $p \neq 2$ be a prime, and let $k=O(n^{1/2d})$.
There exists an $r \in \Theta(\log n)$ for which any $\GC^0(k)[p]\mathsf{/rpoly}$ circuit of depth $d$ and size at most $\exp\left(O\left(n^{1/2.01d} \right)\right)$ solves the $(q,r)$-Parallel Parity Bending Problem (\cref{def:q-r-parallel-parity-bending}) with probability at most $n^{-\Omega(1)}$.
\end{theorem}
\begin{proof}
The proof is similar to \cref{thm:3-output-mod-3}. 
For $k= O(n^{1/2d})$, let $C$ be a $\GC^0(k)[2]\mathsf{/rpoly}$ circuit of depth $d$ and size at most $\exp\left(O\left(n^{1/2.01d} \right) \right)$ that, on input $x_1,\dots,x_r \in \{0,1\}^n$, outputs $y_1,\dots,y_r \in \{0,1\}$ such that, for at least a $\frac{1}{2} + 0.01$ fraction of $i \in [r]$, $y_i \equiv \abs{x_i} \pmod 2$. Let $\eps$ denote the probability that $C$ succeeds at this task.
Throughout this proof, consider $x_1,\dots,x_r$ to be chosen uniformly at random.
Let $\calD$ be the distribution over $r$ bits defined by 
\[
\bigotimes_{i=1}^r (\abs{x_i} - y_i \pmod 2).
\]

We will show that $\calD$ is close to the uniform distribution over $\{0,1\}^r$ in total variation distance.
Let $\chi_a$ be the character of $\F_2^r$ corresponding to $a \in \F_2^r$.
Recall that $\chi_a(z) \coloneqq (-1)^{\sum_{i=1}^r a_i z_i}$, where $z \in \F_2^n$.
We will show that $\abs{\E[\chi_a(\calD)]}$ is small for all nonzero $a$, which implies that $\calD$ is close to the uniform distribution in total variation distance.

For $a \in \F_2^r$, let $S$ denote the set of indices on which $a_i \neq 0$.
Consider the problem where, given a nonzero $a \in \F_2^r$ and strings $x_1,\dots,x_r \in \{0,1\}^n$, the goal is to find $y_1,\dots,y_r \in \{0,1\}$ such that 
\[
\sum_{i \in S} a_i \abs{x_i} \equiv \sum_{i \in S} a_i y_i \pmod 2.
\]
Let $\widetilde{x} \coloneqq (a_i x_{j,i})_{i\in S, j\in [r]} \in \{0,1\}^{n\abs{S}}$, i.e., the bits chosen by $a$ for each $x_i$ for $i \in [r]$. 
The problem above reduces to computing $\Mod_2$ on $\widetilde{x}$.
Specifically, let $y_1, \dots, y_r$ be the output of a circuit solving the former problem. Then, add a circuit that computes the sum $\sum_{i\in S} a_i y_i \pmod 2$, which is equal to $\Mod_2(\widetilde{x})$.
Note this last step requires at most a depth-$2$ $\AC^0$ circuit with $\exp(\abs{S}) \leq \exp(r) \leq \poly(n)$ many gates.

Next, because $x_1,\dots,x_r$ are uniformly random, so too is the concatenated input $\widetilde{x}$.
Therefore, \cref{cor:2-output-2-uniform} implies that the distribution 
\[
\sum_{i\in S} a_i (\abs{x_i} - y_i) \pmod 2
\]
is at most $n^{-\Omega(1)}$-far from the uniform distribution over a single bit in total variation distance. Hence, $\abs{\E[\chi_a(\calD)]}\leq n^{-\Omega(1)}$. Then, \cref{lem:xor-lemma-abelian} implies that $\calD$ is $n^{-\Omega(1)}\sqrt{2^r}$-close to the uniform distribution on $\{0,1\}^r$.

Note that for a sample drawn from $\calD$, the bits that are $0$ correspond to the circuit successfully computing $\Mod_2$ on the corresponding input. Hence, the success probability $\eps$ of $C$ is precisely the probability that a sample drawn from $\calD$ has more than a $\frac{1}{2}+0.01$ fraction of the bits set to $0$.
By a Chernoff bound, the probability that a uniformly random string in $\{0,1\}^r$ has more than a $\frac{1}{2} + 0.01$ of its bits set to $0$ is at most $\exp(-\Omega(r))$. 
Because $\calD$ is $n^{-\Omega(1)}\sqrt{2^r}$-close to uniform (in variation distance), we have that $\eps$ (i.e., the probability that the number of bits in a sample drawn from $\calD$ has more than a $\frac{1}{2} + 0.01$ fraction of its bits set to 0) is at most 
\[
n^{-\Omega(1)} \sqrt{3^r} + \exp(-\Omega(r)), 
\]
which is bounded above by $n^{-\Omega(1)}$ for some $r \in \Theta(\log n)$.

At this point, we have shown that any $\GC^0(k)[p]\mathsf{/rpoly}$ circuit of depth $d$ and size at most $\exp(O(n^{1/2.01d}))$ trying to compute $\Mod_2(x_i)$ on a $\frac{1}{2} + 0.01$ fraction of inputs $x_1,\dots,x_r \in \{0,1\}^n$ will succeed with probability at most $n^{-\Omega(1)}$. To complete the proof, we give a reduction from this problem to the $(q,r)$-Parallel Parity Bending Problem, following the reduction given in \cite[Theorem 40]{watts2019exponential}.
Suppose we have a solution $y_1,\dots,y_r$ to $(q,r)$-Parallel Parity Bending Problem. Then we can output $y_1',\dots,y_r'$ solving the above problem as follows.
For $y_i$, set $y_i' = 0$ when $\abs{y_i}\equiv 0 \pmod q $, and set $y_i' = 1$ otherwise. 
This transformation preserves the number of successes, i.e., if $y_i$ is correct for the $(q,r)$-Parallel Parity Bending Problem, then $y_i'$ will equal $\Mod_2(x_i)$. 
\end{proof}

\subsection{On Interactive \texorpdfstring{$\QNC^0$}{QNC0} Circuits}

Grier and Schaeffer \cite{grier2019interactive} obtain quantum-classical separations for two-round interactive problems. 
We provide a high-level overview of their interactive problems and refer readers to \cite{grier2019interactive} for further detail. 
The problems involve a simple quantum state $\ket{G}$ that is fixed (independent of the input).
In the first round, the input specifies a sequence of Clifford gates to be applied to $\ket{G}$, along with a subset of $n - O(1)$ qubits to measure in the standard basis. A valid output for this round is any measurement outcome that could have been observed if the measurement was performed on an actual quantum computer. 

In the second round, a similar process occurs: the input specifies a sequence of Clifford gates to be applied to the $O(1)$ qubits that were not measured in the first round. Again, a valid output is any measurement outcome that could have been observed if the measurement was performed on a quantum computer. 

To summarize, all the interactive problems in \cite{grier2019interactive} revolve around simulating a Clifford circuit on $n$ qubits, and the simulation is broken into two rounds. The \emph{point} is that this problem caters to quantum devices, and the interactive aspect is crucial for proving lower bounds. 

In more detail, Grier and Schaeffer give three different interactive tasks $\mathsf{T_1}, \mathsf{T_2},$ and $\mathsf{T_3}$ that follow the above structure. The differences between the three tasks come from, e.g., the geometry of the starting state $\ket{G}$. 
It is not too surprising that Grier and Schaeffer show that $\QNC^0$ can solve their interactive tasks. 
On the other hand, they prove that any classical model that can solve these interactive tasks (i.e., \emph{simulate} the action on the fixed state $\ket{G}$) must be fairly powerful.
A bit more carefully, \cite[Theorem 1]{grier2019interactive} shows that $\AC^0[6] \subseteq \left(\AC^0\right)^{\mathsf{T_1}}$, $\NC^1 \subseteq \left(\AC^0\right)^{\mathsf{T_2}}$, and $\oplus \mathsf{L} \subseteq \left(\AC^0\right)^{\mathsf{T_3}}$.

To illustrate the usefulness of their theorem, let us explain how it implies a separation between $\AC^0[2]$ and $\QNC^0$. For the upper bound, they show that $\QNC^0$ can solve any of the tasks $\mathsf{T_i}$. For the lower bound, suppose towards a contradiction that $\AC^0[2]$ can solve $\mathsf{T_2}$. Then, by Grier and Schaeffer's theorem, this implies that $\NC^1 \subseteq \left(\AC^0\right)^{\AC^0[2]} = \AC^0[2]$, but this is a contradiction because the containment of $\AC^0[2]$ in $\NC^1$ is known to be strict.

The remainder of this subsection will use Grier and Schaeffer's framework to show that there is an interactive task that $\QNC^0$ circuits can solve but $\GC^0(k)[p]$ circuits cannot.
We begin by showing that even a single $\g(k)$ gate can compute functions that are not computable by $\NC = \AC = \TC$.

\begin{theorem}\label{thm:gk-incomp-nc}
   There is a function $f:\bitz^n \to \bitz$ computable by a single $\g(k)$ gate that is not computable in $\NC^i$ for any constant $i$ and $k = \omega(\log^{i-1}(n))$.     
   When $k \in \log^{\omega(1)}(n)$, then there are functions $f:\bitz^n \to \bitz$ that are computable by a single $\g(k)$ gate that cannot be computed in $\NC = \AC = \TC$.
\end{theorem}
\begin{proof}
We count the functions computable by $\NC^i$ and a single $\g(k)$.
For $\NC^i$, since the circuit has depth $O(\log^i n)$ with fan-in $2$, there are $\le 2^{O(\log^i n)}$ gates in any $\NC^i$ circuit. 
Furthermore, all fan-in points of these gates are connected by a wire to the fan-out of another gate or an input bit. 

Each gate can be one of $\{\AND, \OR, \NOT\}$, giving $6^{O(\log^i n)}$ many options. 
For each fan-in point of a gate, there exists $\le 2^{O(\log^i n)}+n+2$ many choices of wires that will connect this fan-in point to either the fan-out of another gate, an input variable, or a constant $0/1$ bit. This gives a total of $6^{O(\log^i n)}(2^{O(\log^i n)}+n+2)^{2^{O(\log^i n)}} = 2^{\widetilde{O}(2^{\log^i n})}$ $\NC^1$ circuits.
Meanwhile, the number of $\G(k)$ gates of fan-in $n$ is at least $2^{\binom{n}{\leq k}}$. 
To see this, note that ${\binom{n}{\leq k}}$ many inputs can be assigned arbitrarily, giving $2^{\binom{n}{\leq k}}$ many options.\footnote{By counting carefully, one can show that the number of $\g(k)$ gates is $2\cdot 2^{\binom{n}{\leq k}}$ for $0 \leq k \leq n-1$, and $2^{2^n}$ for $k=n$. We do not need this for our argument.} 
This number exceeds 
$2^{\widetilde{O}(2^{\log^i n})}$ as long as $k = \omega(\log^{i-1}(n))$.
The final part of the theorem follows from setting $k = \log^{\omega(1)}(n)$.
\end{proof}

As another form of \cref{thm:gk-incomp-nc}, we can also show that, e.g., a single $\g(k)$ gate can compute functions that require exponential-size circuits. 
We find this interesting in its own right because proving superlinear circuit lower bounds is currently beyond our techniques. 

\begin{theorem}\label{thm:gk-incomp-TC}
   There is a function $f:\bitz^n \to \bitz$ computable by a single $\g(k)$ gate that requires $\mathsf{SIZE}\left(2^{\widetilde{\Omega}(n^\eps)}\right)$ for $k = \Omega(n^\eps)$  and $\eps > 0$.
\end{theorem}
\begin{proof}
   We use a counting argument. In the proof of \cref{thm:gk-incomp-nc}, we showed that there are at least $2^{\binom{n}{\le k}}$ functions computable by $\g(k)$ gates. 
   We will give a loose upper bound on the number of size-$s$ circuits, which suffices for our purposes. There are $3$ choices each gate could be (from $\{\AND,\OR,\NOT\}$), and each gate has at most $\binom{s+n}2$ choices of two gates to feed into it (including the $n$ input bits). Finally there are $s$ ways to pick one gate to be the output. Thus the total number of ways to pick our $s$ gates are at most $3^s\binom{n+s}2^ss = (n+s)^{O(s)}$. 
   
   For $s = \Omega(n^{k-1})$ and $k\ge 2$, this quantity is $\le (n+s))^{O(s)} = 2^{O(s\log s)} = 2^{\widetilde{O}(n^{k-1})}$.  But $2^{\widetilde{O}(n^{k-1})} = o(2^{\binom{n}{\le k}})$, so the number of size-$s$ circuits is smaller than the number of $\g(k)$ gates for $s = \Omega(n^{k-1})$ and $k \ge 2$. Hence, there exists a $\g(k)$ gate that cannot be computed by size $n^{k-1}$ circuits. 
   In particular, by setting $k = n^\eps$, we see that $\GC^0(n^\eps)\notin \mathsf{SIZE}(2^{\widetilde{\Omega}(n^\eps)})$.
\end{proof}

We say two circuit classes $\mathsf{C}$ and $\mathsf{D}$ are incomparable when there are functions $f, g :\bitz^n\to \bitz$ such that $f \in \mathsf{C}$ but $f \notin \mathsf{D}$ and $g \notin \mathsf{C}$ but $g \in \mathsf{D}$.

\begin{corollary}\label{cor:gc0k-nc1-incomp}
Let $p$ be a prime number. For $k \in \omega(1)$, the class of depth-$d$ $\GC^0(k)[p]$ circuits of size at most $\exp\left(O(n^{1/2.01d})\right)$ is incomparable to $\NC^1$. 
\end{corollary}
\begin{proof}
\cref{thm:mod-p-correlation-bounds} says that $\MAJ$ cannot be computed by $\GC^0(k)[p]$ for any prime $p$. $\MAJ$ can be computed by $\NC^1$ because $\NC^1\supseteq \TC^0$.
\cref{thm:gk-incomp-nc} implies that there is a function that can be computed by $\GC^0(k)[p]$ but not $\NC^1$.
\end{proof}

We can now use Grier and Schaeffer's framework to get a separation between $\QNC^0$ and $\GC^0(k)[p]$ for an interactive problem. 

\begin{theorem}[Generalization of {\cite[Corollary 2]{grier2019interactive}}]\label{thm:grier-schaeffer}
Let $k = O(n^{1/2d})$.
There is an interactive task that $\QNC^0$ circuits can solve that depth-$d$, size-$s$ $\GC^0(k)[p]$ circuits cannot for $s \leq \exp\left(O(n^{1/2.01d})\right)$.
\end{theorem}
\begin{proof}
Grier and Schaeffer's task $\mathsf{T}_2$ (\cite[Problem 12]{grier2019interactive}) can be solved by $\QNC^0$. Suppose it can be solved by $\GC^0(k)[p]$ circuits for some prime $p$. Then, by \cite[Theorem 1]{grier2019interactive}, $\NC^1 \subseteq (\AC^0)^{\GC^0(k)[p]} = \GC^0(k)[p]$ but this contradicts \cref{cor:gc0k-nc1-incomp}.
\end{proof}

\section*{Acknowledgements}
We thank Scott Aaronson, Srinivasan Arunachalam, Anna G\'{a}l, Uma Girish, Jesse Goodman, Daniel Grier, Siddhartha Jain, Nathan Ju, William Kretschmer, Shyamal Patel, Avishay Tal, Ryan Williams, and Justin Yirka for helpful conversations.
This work was done in part while the authors were visiting the Simons Institute for the Theory of Computing, supported by NSF QLCI Grant No. 2016245, and in part while SG was an intern at IBM Quantum.

SG is supported by the NSF QLCI Award OMA-2016245 (Scott Aaronson).
VMK is supported by NSF Grants CCF-2008076 and CCF-2312573, and a Simons Investigator
Award (\#409864, David Zuckerman).

\bibliographystyle{alphaurl}
\bibliography{refs}

\newcommand{\etalchar}[1]{$^{#1}$}
\begin{thebibliography}{GKMdO24}

\bibitem[AA15]{aaronson2015forrelation}
Scott Aaronson and Andris Ambainis.
\newblock {\textsc{Forrelation:} A Problem that Optimally Separates Quantum from Classical Computing}.
\newblock In {\em Proceedings of the Forty-Seventh Annual ACM Symposium on Theory of Computing}, pages 307--316, 2015.
\newblock \href {https://arxiv.org/abs/1411.5729} {\path{arXiv:1411.5729}}, \href {https://doi.org/10.1145/2746539.2746547} {\path{doi:10.1145/2746539.2746547}}.

\bibitem[Aar10]{aaronson2010bqp}
Scott Aaronson.
\newblock {$\BQP$ and the Polynomial Hierarchy}.
\newblock In {\em Proceedings of the Forty-Second ACM Symposium on Theory of Computing}, pages 141--150, 2010.
\newblock \href {https://arxiv.org/abs/2010.05846} {\path{arXiv:2010.05846}}, \href {https://doi.org/10.1145/1806689.1806711} {\path{doi:10.1145/1806689.1806711}}.

\bibitem[Aar16]{aaronson2016p}
Scott Aaronson.
\newblock {$\PTIME \stackrel{?}{=} \NP$}.
\newblock {\em Open Problems in Mathematics}, pages 1--122, 2016.
\newblock \texttt{\href{https://www.scottaaronson.com/papers/pnp.pdf}{scottaaronson.com}}.
\newblock \href {https://doi.org/10.1007/978-3-319-32162-2_1} {\path{doi:10.1007/978-3-319-32162-2_1}}.

\bibitem[AB09]{arora2009computational}
Sanjeev Arora and Boaz Barak.
\newblock {\em {Computational Complexity: A Modern Approach}}.
\newblock Cambridge University Press, 2009.
\newblock \href {https://doi.org/10.5555/1540612} {\path{doi:10.5555/1540612}}.

\bibitem[ABO84]{ajtai1984theorem}
Mikl{\'o}s Ajtai and Michael Ben-Or.
\newblock {A Theorem on Probabilistic Constant Depth Computations}.
\newblock In {\em Proceedings of the Sixteenth Annual ACM Symposium on Theory of Computing}, pages 471--474, 1984.
\newblock \href {https://doi.org/10.1145/800057.808715} {\path{doi:10.1145/800057.808715}}.

\bibitem[ADH97]{adleman1997quantum}
Leonard~M. Adleman, Jonathan DeMarrais, and Ming-Deh~A. Huang.
\newblock {Quantum Computability}.
\newblock {\em SIAM Journal on Computing}, 26(5):1524--1540, 1997.
\newblock \href {https://doi.org/10.1137/S0097539795293639} {\path{doi:10.1137/S0097539795293639}}.

\bibitem[AG94]{allender1994uniform}
Eric Allender and Vivek Gore.
\newblock {A Uniform Circuit Lower Bound for the Permanent}.
\newblock {\em SIAM Journal on Computing}, 23(5):1026--1049, 1994.
\newblock \href {https://doi.org/10.1137/S0097539792233907} {\path{doi:10.1137/S0097539792233907}}.

\bibitem[AGS21]{arunachalam2021quantum}
Srinivasan Arunachalam, Alex~Bredariol Grilo, and Aarthi Sundaram.
\newblock {Quantum Hardness of Learning Shallow Classical Circuits}.
\newblock {\em SIAM Journal on Computing}, 50(3):972--1013, 2021.
\newblock \href {https://arxiv.org/abs/1903.02840} {\path{arXiv:1903.02840}}, \href {https://doi.org/10.1137/20M1344202} {\path{doi:10.1137/20M1344202}}.

\bibitem[AH94]{AHdepthreduction94}
Eric Allender and Ulrich Hertrampf.
\newblock {Depth Reduction for Circuits of Unbounded Fan-in}.
\newblock {\em Information and Computation}, 112(2):217--238, 1994.
\newblock \href {https://doi.org/10.1006/inco.1994.1057} {\path{doi:10.1006/inco.1994.1057}}.

\bibitem[AIK22]{aaronson2022acrobatics}
Scott Aaronson, DeVon Ingram, and William Kretschmer.
\newblock {The Acrobatics of $\BQP$}.
\newblock In {\em 37th Computational Complexity Conference (CCC 2022)}, volume 234 of {\em Leibniz International Proceedings in Informatics (LIPIcs)}, pages 20:1--20:17, 2022.
\newblock \href {https://doi.org/10.4230/LIPIcs.CCC.2022.20} {\path{doi:10.4230/LIPIcs.CCC.2022.20}}.

\bibitem[Ajt83]{ajtai198311}
Mikl{\'o}s Ajtai.
\newblock {$\Sigma_1$-Formulae on Finite Structures}.
\newblock {\em Annals of Pure and Applied Logic}, 24(1):1--48, 1983.
\newblock \href {https://doi.org/10.1016/0168-0072(83)90038-6} {\path{doi:10.1016/0168-0072(83)90038-6}}.

\bibitem[AW93]{allender1993counting}
Eric~W. Allender and Klaus~W. Wagner.
\newblock {Counting Hierarchies: Polynomial Time and Constant Depth Circuits}.
\newblock In {\em Current Trends in Theoretical Computer Science: Essays and Tutorials}, pages 469--483. World Scientific, 1993.
\newblock \href {https://doi.org/10.1142/9789812794499_0035} {\path{doi:10.1142/9789812794499_0035}}.

\bibitem[AW09]{aaronson2009algebrization}
Scott Aaronson and Avi Wigderson.
\newblock {Algebrization: A New Barrier in Complexity Theory}.
\newblock {\em ACM Transactions on Computation Theory}, 1(1):1--54, 2009.
\newblock \href {https://doi.org/10.1145/1490270.1490272} {\path{doi:10.1145/1490270.1490272}}.

\bibitem[BG81]{charles1981relative}
Charles~H. Bennett and John Gill.
\newblock {Relative to a Random Oracle $A$, $\PTIME^A\neq\NP^A \neq \coNP^A$ with Probability $1$}.
\newblock {\em SIAM Journal on Computing}, 10(1):96--113, 1981.
\newblock \href {https://doi.org/10.1137/0210008} {\path{doi:10.1137/0210008}}.

\bibitem[BGK18]{bravyi2018quantum}
Sergey Bravyi, David Gosset, and Robert K{\"o}nig.
\newblock {Quantum advantage with shallow circuits}.
\newblock {\em Science}, 362(6412):308--311, 2018.
\newblock \href {https://arxiv.org/abs/1704.00690} {\path{arXiv:1704.00690}}, \href {https://doi.org/10.1126/science.aar3106} {\path{doi:10.1126/science.aar3106}}.

\bibitem[BGKT20]{bravyi2020quantum}
Sergey Bravyi, David Gosset, Robert K{\"o}nig, and Marco Tomamichel.
\newblock Quantum advantage with noisy shallow circuits.
\newblock {\em Nature Physics}, 16(10):1040--1045, 2020.
\newblock \href {https://arxiv.org/abs/1904.01502} {\path{arXiv:1904.01502}}, \href {https://doi.org/10.1038/s41567-020-0948-z} {\path{doi:10.1038/s41567-020-0948-z}}.

\bibitem[BGS75]{baker1975relativizations}
Theodore Baker, John Gill, and Robert Solovay.
\newblock {Relativizations of the $\PTIME \stackrel{?}{=} \NP$ Question}.
\newblock {\em SIAM Journal on computing}, 4(4):431--442, 1975.
\newblock \href {https://doi.org/10.1137/0204037} {\path{doi:10.1137/0204037}}.

\bibitem[Blu81]{blum19812}
Norbert Blum.
\newblock {\em {A $2.75n$-lower bound on the network complexity of Boolean functions}}.
\newblock Technical Report A81/05, Universit\"at des Saarlandes, 1981.

\bibitem[Blu83]{blum1983boolean}
Norbert Blum.
\newblock {A Boolean function requiring $3n$ network size}.
\newblock {\em Theoretical Computer Science}, 28(3):337--345, 1983.
\newblock \href {https://doi.org/10.1016/0304-3975(83)90029-4} {\path{doi:10.1016/0304-3975(83)90029-4}}.

\bibitem[BV97]{bernstein1993quantum}
Ethan Bernstein and Umesh Vazirani.
\newblock {Quantum Complexity Theory}.
\newblock {\em SIAM Journal on Computing}, 26(5):1411--1473, 1997.
\newblock \href {https://doi.org/10.1137/S0097539796300921} {\path{doi:10.1137/S0097539796300921}}.

\bibitem[CHO{\etalchar{+}}22]{chen2022beyond}
Lijie Chen, Shuichi Hirahara, Igor~Carboni Oliveira, J\'{a}n Pich, Ninad Rajgopal, and Rahul Santhanam.
\newblock Beyond natural proofs: Hardness magnification and locality.
\newblock {\em Journal of the ACM}, 69(4), 2022.
\newblock \href {https://doi.org/10.1145/3538391} {\path{doi:10.1145/3538391}}.

\bibitem[CIKK16]{carmosino2016learning}
Marco~L. Carmosino, Russell Impagliazzo, Valentine Kabanets, and Antonina Kolokolova.
\newblock {Learning Algorithms from Natural Proofs}.
\newblock In {\em 31st Conference on Computational Complexity (CCC 2016)}, 2016.
\newblock \href {https://doi.org/10.4230/LIPIcs.CCC.2016.1} {\path{doi:10.4230/LIPIcs.CCC.2016.1}}.

\bibitem[CSV19]{coudron2019trading}
Matthew Coudron, Jalex Stark, and Thomas Vidick.
\newblock {Trading Locality for Time: Certifiable Randomness from Low-Depth Circuits}.
\newblock {\em Communications in Mathematical Physics}, 2019.
\newblock \href {https://arxiv.org/abs/1810.04233} {\path{arXiv:1810.04233}}, \href {https://doi.org/10.1007/s00220-021-03963-w} {\path{doi:10.1007/s00220-021-03963-w}}.

\bibitem[DK11]{demenkov2011elementary}
Evgeny Demenkov and Alexander~S Kulikov.
\newblock {An Elementary Proof of a $3n- o(n)$ Lower Bound on the Circuit Complexity of Affine Dispersers}.
\newblock In {\em International Symposium on Mathematical Foundations of Computer Science}, pages 256--265. Springer, 2011.
\newblock \texttt{\href{https://eccc.weizmann.ac.il/report/2011/026/}{eccc:TR11-026}}.
\newblock \href {https://doi.org/10.1007/978-3-642-22993-0_25} {\path{doi:10.1007/978-3-642-22993-0_25}}.

\bibitem[FGHK16]{find2016better}
Magnus~Gausdal Find, Alexander Golovnev, Edward~A. Hirsch, and Alexander~S Kulikov.
\newblock {A Better-Than-$3n$ Lower Bound for the Circuit Complexity of an Explicit Function}.
\newblock In {\em 2016 IEEE 57th Annual Symposium on Foundations of Computer Science}, pages 89--98, 2016.
\newblock \texttt{\href{https://eccc.weizmann.ac.il/report/2015/166/}{eccc:TR15-166}}.
\newblock \href {https://doi.org/10.1109/FOCS.2016.19} {\path{doi:10.1109/FOCS.2016.19}}.

\bibitem[FGL20]{fawzi2020constant}
Omar Fawzi, Antoine Grospellier, and Anthony Leverrier.
\newblock {Constant overhead quantum fault tolerance with quantum expander codes}.
\newblock {\em Communications of the ACM}, 64(1):106--114, 2020.
\newblock \href {https://arxiv.org/abs/1808.03821} {\path{arXiv:1808.03821}}, \href {https://doi.org/10.1145/3434163} {\path{doi:10.1145/3434163}}.

\bibitem[FSS84]{furst1984parity}
Merrick Furst, James~B. Saxe, and Michael Sipser.
\newblock {Parity, Circuits, and the Polynomial-Time Hierarchy}.
\newblock {\em Mathematical Systems Theory}, 17(1):13--27, 1984.
\newblock \href {https://doi.org/10.1007/BF01744431} {\path{doi:10.1007/BF01744431}}.

\bibitem[FSUV13]{fefferman2012beating}
Bill Fefferman, Ronen Shaltiel, Christopher Umans, and Emanuele Viola.
\newblock {On Beating the Hybrid Argument}.
\newblock {\em Theory of Computing}, 9(26):809--843, 2013.
\newblock \href {https://doi.org/10.4086/toc.2013.v009a026} {\path{doi:10.4086/toc.2013.v009a026}}.

\bibitem[GJS21]{grier2021interactive}
Daniel Grier, Nathan Ju, and Luke Schaeffer.
\newblock {Interactive quantum advantage with noisy, shallow Clifford circuits}, 2021.
\newblock \href {https://arxiv.org/abs/2102.06833} {\path{arXiv:2102.06833}}.

\bibitem[GK16]{golovnev2016weighted}
Alexander Golovnev and Alexander~S. Kulikov.
\newblock {Weighted Gate Elimination: Boolean Dispersers for Quadratic Varieties Imply Improved Circuit Lower Bounds}.
\newblock In {\em Proceedings of the 2016 ACM Conference on Innovations in Theoretical Computer Science}, pages 405--411, 2016.
\newblock \href {https://doi.org/10.1145/2840728.2840755} {\path{doi:10.1145/2840728.2840755}}.

\bibitem[GKMdO24]{grilo2024power}
Alex~Bredariol Grilo, Elham Kashefi, Damian Markham, and Michael de~Oliveira.
\newblock {The power of shallow-depth Toffoli and qudit quantum circuits}, 2024.
\newblock \href {https://arxiv.org/abs/2404.18104} {\path{arXiv:2404.18104}}.

\bibitem[Gol08]{goldreich2008computational}
Oded Goldreich.
\newblock {\em {Computational Complexity: A Conceptual Perspective}}.
\newblock Cambridge University Press, 2008.
\newblock \href {https://doi.org/10.1017/CBO9780511804106} {\path{doi:10.1017/CBO9780511804106}}.

\bibitem[GS20]{grier2019interactive}
Daniel Grier and Luke Schaeffer.
\newblock {Interactive shallow Clifford circuits: Quantum advantage against $\NC^1$ and beyond}, 2020.
\newblock \href {https://arxiv.org/abs/1911.02555} {\path{arXiv:1911.02555}}, \href {https://doi.org/10.1145/3357713.3384332} {\path{doi:10.1145/3357713.3384332}}.

\bibitem[H{\aa}s86]{haastad1986computational}
Johan H{\aa}stad.
\newblock {\em Computational limitations for small depth circuits}.
\newblock PhD thesis, Massachusetts Institute of Technology, 1986.

\bibitem[H{\aa}s14]{haastad2014correlation}
Johan H{\aa}stad.
\newblock {On the Correlation of Parity and Small-Depth Circuits}.
\newblock {\em SIAM Journal on Computing}, 43(5):1699--1708, 2014.
\newblock \href {https://doi.org/10.1137/120897432} {\path{doi:10.1137/120897432}}.

\bibitem[HMdOS24]{hsieh2024concurrent}
Min-Hsiu Hsieh, Leandro Mendes, Michael de~Oliveira, and Sathyawageeswar Subramanian.
\newblock {Unconditionally separating noisy $\QNC^0$ from bounded polynomial threshold circuits of constant depth}, 2024.

\bibitem[HRST17]{rossman2015average}
Johan H{\aa}stad, Benjamin Rossman, Rocco~A. Servedio, and Li-Yang Tan.
\newblock {An Average-Case Depth Hierarchy Theorem for Boolean Circuits}.
\newblock {\em Journal of the ACM}, 64(5), 2017.
\newblock \href {https://doi.org/10.1145/3095799} {\path{doi:10.1145/3095799}}.

\bibitem[IW97]{impagliazzo1997p}
Russell Impagliazzo and Avi Wigderson.
\newblock {$\PTIME = \BPP$ if $\e$ requires exponential circuits: Derandomizing the XOR lemma}.
\newblock In {\em Proceedings of the Twenty-Ninth Annual ACM Symposium on Theory of Computing}, pages 220--229, 1997.
\newblock \href {https://doi.org/10.1145/258533.258590} {\path{doi:10.1145/258533.258590}}.

\bibitem[Juk90]{jukna1990monotone}
Stasys Jukna.
\newblock {Monotone Circuits and Local Computations}.
\newblock In {\em Proceedings of the 31th Conference of Lithuanian Mathematical Society}, pages 28--29, 1990.

\bibitem[Kum23]{kumar2023tight}
Vinayak~M. Kumar.
\newblock {Tight Correlation Bounds for Circuits Between $\AC^0$ and $\TC^0$}.
\newblock In {\em 38th Computational Complexity Conference}, volume 264, pages 18:1--18:40, 2023.
\newblock \href {https://arxiv.org/abs/2304.02770} {\path{arXiv:2304.02770}}, \href {https://doi.org/10.4230/LIPIcs.CCC.2023.18} {\path{doi:10.4230/LIPIcs.CCC.2023.18}}.

\bibitem[LG19]{legall2019average}
Fran\c{c}ois Le~Gall.
\newblock {Average-Case Quantum Advantage with Shallow Circuits}.
\newblock In {\em 34th Computational Complexity Conference}, volume 137 of {\em Leibniz International Proceedings in Informatics (LIPIcs)}, pages 21:1--21:20, 2019.
\newblock \href {https://arxiv.org/abs/1810.12792} {\path{arXiv:1810.12792}}, \href {https://doi.org/10.4230/LIPICS.CCC.2019.21} {\path{doi:10.4230/LIPICS.CCC.2019.21}}.

\bibitem[LMN93]{linial1993constant}
Nathan Linial, Yishay Mansour, and Noam Nisan.
\newblock {Constant Depth Circuits, Fourier Transform, and Learnability}.
\newblock {\em Journal of the ACM}, 40(3):607--620, 1993.
\newblock \href {https://doi.org/10.1145/174130.174138} {\path{doi:10.1145/174130.174138}}.

\bibitem[LY22]{li20221}
Jiatu Li and Tianqi Yang.
\newblock {$3.1n-o(n)$ Circuit Lower Bounds for Explicit Functions}.
\newblock In {\em Proceedings of the 54th Annual ACM SIGACT Symposium on Theory of Computing}, pages 1180--1193, 2022.
\newblock \texttt{\href{https://eccc.weizmann.ac.il/report/2021/023/}{eccc:TR21-023}}.
\newblock \href {https://doi.org/10.1145/3519935.3519976} {\path{doi:10.1145/3519935.3519976}}.

\bibitem[MP43]{mcculloch1943logical}
Warren~S McCulloch and Walter Pitts.
\newblock {A logical calculus of the ideas immanent in nervous activity}.
\newblock {\em The Bulletin of Mathematical Biophysics}, 5:115--133, 1943.
\newblock \href {https://doi.org/10.1007/BF02478259} {\path{doi:10.1007/BF02478259}}.

\bibitem[MSS91]{maass1991computational}
Wolfgang Maass, Georg Schnitger, and Eduardo~D. Sontag.
\newblock {On the Computational Power of Sigmoid versus Boolean Threshold Circuits}.
\newblock In {\em Proceedings of the Thirty Second Annual Symposium of Foundations of Computer Science}, pages 767--776, 1991.
\newblock \href {https://doi.org/10.1109/SFCS.1991.185447} {\path{doi:10.1109/SFCS.1991.185447}}.

\bibitem[Mur71]{muroga1971threshold}
Saburo Muroga.
\newblock {Threshold Logic and Its Applications}.
\newblock {\em John Wiley \& Sons, Inc.}, 1971.

\bibitem[NC02]{nielsen2002quantum}
Michael~A. Nielsen and Isaac Chuang.
\newblock {Quantum Computation and Quantum Information}, 2002.
\newblock \href {https://doi.org/10.1017/CBO9780511976667} {\path{doi:10.1017/CBO9780511976667}}.

\bibitem[NW94]{nisan1994hardness}
Noam Nisan and Avi Wigderson.
\newblock {Hardness vs. Randomness}.
\newblock {\em Journal of Computer and System Sciences}, 49(2):149--167, 1994.
\newblock \href {https://doi.org/10.1016/S0022-0000(05)80043-1} {\path{doi:10.1016/S0022-0000(05)80043-1}}.

\bibitem[Pau75]{paul19752}
Wolfgang~J Paul.
\newblock {A $2.5n$-lower bound on the combinational complexity of Boolean functions}.
\newblock In {\em Proceedings of the Seventh Annual ACM Symposium on Theory of Computing}, pages 27--36, 1975.
\newblock \href {https://doi.org/10.1145/800116.803750} {\path{doi:10.1145/800116.803750}}.

\bibitem[Rao07]{rao2007exposition}
Anup Rao.
\newblock {An Exposition of Bourgain’s 2-Source Extractor}.
\newblock In {\em Electronic Colloquium on Computational Complexity (ECCC)}, 2007.
\newblock \texttt{\href{https://eccc.weizmann.ac.il/report/2007/034/}{eccc:TR07-034}}.

\bibitem[Raz85]{razborov1985lower}
Alexander Razborov.
\newblock Lower bounds on the monotone complexity of some boolean function.
\newblock In {\em Doklady Mathematics}, volume~31, pages 354--357, 1985.

\bibitem[Raz87]{razborov1987lower}
Alexander~A. Razborov.
\newblock {Lower bounds on the size of bounded depth circuits over a complete basis with logical addition}.
\newblock {\em Mathematical notes of the Academy of Sciences of the USSR}, 41(4):598--607, 1987.
\newblock \href {https://doi.org/10.1007/BF01137685} {\path{doi:10.1007/BF01137685}}.

\bibitem[Reg24]{regev2024efficient}
Oded Regev.
\newblock {An Efficient Quantum Factoring Algorithm}, 2024.
\newblock \href {https://arxiv.org/abs/2308.06572} {\path{arXiv:2308.06572}}.

\bibitem[Ros17]{ros17entropyswitch}
Benjamin Rossman.
\newblock {An entropy proof of the switching lemma and tight bounds on the decision-tree size of $\AC^0$ }, 2017.
\newblock URL: \url{https://users.cs.duke.edu/~br148/logsize.pdf}.

\bibitem[RR97]{razborov1994natural}
Alexander~A. Razborov and Steven Rudich.
\newblock {Natural Proofs}.
\newblock {\em Journal of Computer and System Sciences}, 55(1):24--35, 1997.
\newblock \href {https://doi.org/10.1006/jcss.1997.1494} {\path{doi:10.1006/jcss.1997.1494}}.

\bibitem[RT22]{raz2022oracle}
Ran Raz and Avishay Tal.
\newblock {Oracle separation of $\BQP$ and $\PH$}.
\newblock {\em Journal of the ACM}, 69(4):1--21, 2022.
\newblock \texttt{\href{https://eccc.weizmann.ac.il/report/2018/107/}{eccc:TR18-107}}.
\newblock \href {https://doi.org/10.1145/3313276.3316315} {\path{doi:10.1145/3313276.3316315}}.

\bibitem[Sch74]{schnorr1974zwei}
Claus-Peter Schnorr.
\newblock {Zwei lineare untere Schranken fur die komplexitat Boolescher funktionen}.
\newblock {\em Computing}, 13:155--171, 1974.
\newblock \href {https://doi.org/10.1007/BF02246615} {\path{doi:10.1007/BF02246615}}.

\bibitem[Sch80]{schnorr19803n}
Claus-Peter Schnorr.
\newblock {A $3n$-lower bound on the network complexity of Boolean functions}.
\newblock {\em Theoretical Computer Science}, 10(1):83--92, 1980.
\newblock \href {https://doi.org/10.1016/0304-3975(80)90074-2} {\path{doi:10.1016/0304-3975(80)90074-2}}.

\bibitem[Sho97]{shor1997polynomial}
Peter~W. Shor.
\newblock {Polynomial-Time Algorithms for Prime Factorization and Discrete Logarithms on a Quantum Computer}.
\newblock {\em SIAM Journal on Computing}, 26(5):1484--1509, 1997.
\newblock \href {https://arxiv.org/abs/quant-ph/9508027} {\path{arXiv:quant-ph/9508027}}, \href {https://doi.org/10.1137/S0097539795293172} {\path{doi:10.1137/S0097539795293172}}.

\bibitem[Sim97]{simon1997power}
Daniel~R. Simon.
\newblock {On the Power of Quantum Computation}.
\newblock {\em SIAM Journal on Computing}, 26(5):1474--1483, 1997.
\newblock \href {https://doi.org/10.1137/S0097539796298637} {\path{doi:10.1137/S0097539796298637}}.

\bibitem[Smo87]{smolensky1987algebraic}
Roman Smolensky.
\newblock {Algebraic methods in the theory of lower bounds for Boolean circuit complexity}.
\newblock In {\em Proceedings of the Nineteenth Annual ACM Symposium on Theory of Computing}, pages 77--82, 1987.
\newblock \href {https://doi.org/10.1145/28395.28404} {\path{doi:10.1145/28395.28404}}.

\bibitem[Sto76]{stockmeyer1976combinational}
Larry~J. Stockmeyer.
\newblock {On the combinational complexity of certain symmetric Boolean functions}.
\newblock {\em Mathematical Systems Theory}, 10(1):323--336, 1976.
\newblock \href {https://doi.org/10.1007/BF01683282} {\path{doi:10.1007/BF01683282}}.

\bibitem[STV21]{srinivasan_et_al:LIPIcs.FSTTCS.2019.28}
Srikanth Srinivasan, Utkarsh Tripathi, and S.~Venkitesh.
\newblock {On the Probabilistic Degrees of Symmetric Boolean Functions}.
\newblock {\em SIAM Journal on Discrete Mathematics}, 35(3):2070--2092, 2021.
\newblock \href {https://doi.org/10.1137/19M1294162} {\path{doi:10.1137/19M1294162}}.

\bibitem[Tor91]{toran1991}
Jacobo Tor\'{a}n.
\newblock {Complexity Classes Defined by Counting Quantifiers}.
\newblock {\em Journal of the ACM}, 38(3):752–773, 1991.
\newblock \href {https://doi.org/10.1145/116825.116858} {\path{doi:10.1145/116825.116858}}.

\bibitem[VV85]{valiantvazirani}
Leslie~G. Valiant and Vijay~V. Vazirani.
\newblock $\np$ is as easy as detecting unique solutions.
\newblock In {\em Proceedings of the Seventeenth Annual ACM Symposium on Theory of Computing}, page 458–463. Association for Computing Machinery, 1985.
\newblock \href {https://doi.org/10.1145/22145.22196} {\path{doi:10.1145/22145.22196}}.

\bibitem[Wag86]{wagner1986complexity}
Klaus~W. Wagner.
\newblock {The Complexity of Combinatorial Problems with Succinct Input Representation}.
\newblock {\em Acta Informatica}, 23:325--356, 1986.
\newblock \href {https://doi.org/10.1007/BF00289117} {\path{doi:10.1007/BF00289117}}.

\bibitem[Wil14]{wil14acc0}
Ryan Williams.
\newblock {Nonuniform $\ACC$ Circuit Lower Bounds}.
\newblock {\em Journal of the ACM}, 61(1), 2014.
\newblock \href {https://doi.org/10.1145/2559903} {\path{doi:10.1145/2559903}}.

\bibitem[WKST19]{watts2019exponential}
Adam~Bene Watts, Robin Kothari, Luke Schaeffer, and Avishay Tal.
\newblock Exponential separation between shallow quantum circuits and unbounded fan-in shallow classical circuits.
\newblock In {\em Proceedings of the 51st Annual ACM SIGACT Symposium on Theory of Computing}, pages 515--526, 2019.
\newblock \href {https://arxiv.org/abs/1906.08890} {\path{arXiv:1906.08890}}, \href {https://doi.org/10.1145/3313276.3316404} {\path{doi:10.1145/3313276.3316404}}.

\bibitem[WP24]{watts2024unconditionalquantumadvantagesampling}
Adam~Bene Watts and Natalie Parham.
\newblock {Unconditional Quantum Advantage for Sampling with Shallow Circuits}, 2024.
\newblock \href {https://arxiv.org/abs/2301.00995} {\path{arXiv:2301.00995}}.

\bibitem[Wu22]{wu2022stoch}
Xinyu Wu.
\newblock {A Stochastic Calculus Approach to the Oracle Separation of $\mathsf{BQP}$ and $\mathsf{PH}$}.
\newblock {\em Theory of Computing}, 18(17):1--11, 2022.
\newblock \href {https://arxiv.org/abs/2007.02431} {\path{arXiv:2007.02431}}, \href {https://doi.org/10.4086/toc.2022.v018a017} {\path{doi:10.4086/toc.2022.v018a017}}.

\bibitem[Yao85]{yao1985separating}
Andrew Chi-Chih Yao.
\newblock {Separating the polynomial-time hierarchy by oracles}.
\newblock In {\em 26th Annual Symposium on Foundations of Computer Science}, pages 1--10, 1985.
\newblock \href {https://doi.org/10.1109/SFCS.1985.49} {\path{doi:10.1109/SFCS.1985.49}}.

\bibitem[Yao89]{yao1989circuits}
A.~C. Yao.
\newblock Circuits and local computation.
\newblock In {\em Proceedings of the Twenty-First Annual ACM Symposium on Theory of Computing}, page 186–196. Association for Computing Machinery, 1989.
\newblock \href {https://doi.org/10.1145/73007.73025} {\path{doi:10.1145/73007.73025}}.

\end{thebibliography}

\end{document}